\documentclass[
 pra,
 onecolumn,  
 superscriptaddress,
 amsmath,
 amssymb,
 amsthm,
 longbibliography
]{revtex4-2}

\usepackage{amsthm}

\usepackage{enumitem}
\usepackage{graphicx}
\usepackage{dcolumn}
\usepackage{bm}
\usepackage{tikz}
\usepackage{qcircuit}
\usepackage{tikzit}

\usepackage{braket}
\usepackage{appendix}
\usepackage{xcolor}
\usepackage{float}
\usepackage[caption=false]{subfig}
\usepackage{mathtools}
\usepackage{algorithmicx}
\usepackage{amsmath}
\usepackage{amsfonts}
\usepackage{tabularx}
\usepackage{xltabular}
\usepackage{siunitx}
\usepackage{textcomp}
\usepackage{longtable}
\usepackage{array}
\usepackage{enumitem}
\usepackage{scalerel}

\theoremstyle{definition}
\newtheorem{theorem}{Theorem}[section]

\newtheorem{lemma}[theorem]{Lemma}
\newtheorem*{lemma*}{Lemma}
\newtheorem*{proposition*}{Proposition}

\newtheorem{definition}[theorem]{Definition}
\newtheorem{example}[theorem]{Example}

\makeatletter
\newcommand{\pnrelbar}{%
  \linethickness{\dimen2}%
  \sbox\z@{$\m@th\prec$}%
  \dimen@=1.1\ht\z@
  \begin{picture}(\dimen@,.4ex)
  \roundcap
  \put(0,.2ex){\line(1,0){\dimen@}}
  \put(\dimexpr 0.5\dimen@-.2ex\relax,0){\line(1,1){.4ex}}
  \end{picture}%
}

\newcommand{\precneq}{\mathrel{\vcenter{\hbox{\text{\prec@neq}}}}}
\newcommand{\prec@neq}{%
  \dimen2=\f@size\dimexpr.04pt\relax
  \oalign{%
    \noalign{\kern\dimexpr.2ex-.5\dimen2\relax}
    $\m@th\prec$\cr
    \noalign{\kern-.5\dimen2}
    \hidewidth\pnrelbar\hidewidth\cr
  }%
}
\makeatother

\usepackage{tabularx}

\usepackage{algorithm}
\usepackage{algpseudocode}
\newcounter{protocol}
\newenvironment{protocol}[1]
  {\par\addvspace{\topsep}
   \noindent
   \tabularx{\linewidth}{@{} X @{}}
    \refstepcounter{protocol}\textbf{Protocol \theprotocol} #1 \\
    }
  { 
   \endtabularx
   \par\addvspace{\topsep}}

\input{config/zx.tikzdefs}

\tikzstyle{box}=[shape=rectangle, text height=1.5ex, text depth=0.25ex, yshift=0.5mm, fill=white, draw=black, minimum height=4mm, yshift=-0.5mm, minimum width=4mm, font={\footnotesize}]
\tikzstyle{Z dot}=[inner sep=0mm, minimum size=2mm, shape=circle, draw=black, fill={rgb,255: red,221; green,255; blue,221}]
\tikzstyle{Z phase dot}=[minimum size=5mm, font={\footnotesize\boldmath}, shape=rectangle, rounded corners=2mm, inner sep=0.2mm, outer sep=-2mm, scale=0.8, tikzit shape=circle, draw=black, fill={rgb,255: red,221; green,255; blue,221}, tikzit draw=blue]
\tikzstyle{X dot}=[Z dot, shape=circle, draw=black, fill={rgb,255: red,255; green,136; blue,136}]
\tikzstyle{X phase dot}=[Z phase dot, tikzit shape=circle, tikzit draw=blue, fill={rgb,255: red,255; green,136; blue,136}, font={\footnotesize\boldmath}]
\tikzstyle{hadamard}=[fill=yellow, draw=black, shape=rectangle, inner sep=0.6mm, minimum height=1.5mm, minimum width=1.5mm]
\tikzstyle{vertex}=[inner sep=0mm, minimum size=1.5mm, shape=circle, draw=black, fill=black]
\tikzstyle{vertex set}=[inner sep=0mm, minimum size=2.25mm, shape=circle, draw=black, fill=white, font={\footnotesize\boldmath}]
\tikzstyle{vertex set phase}=[minimum size=5mm, font={\scriptsize\boldmath}, shape=rectangle, rounded corners=2mm, inner sep=0.2mm, outer sep=-2mm, scale=0.8, tikzit shape=circle, draw=black, fill=white, tikzit draw=red]
\tikzstyle{target}=[draw=black,shape=circle,minimum size=3.5mm,append after command={
        [shorten >=\pgflinewidth, shorten <=\pgflinewidth,]
        (\tikzlastnode.north) edge (\tikzlastnode.south)
        (\tikzlastnode.east) edge (\tikzlastnode.west)
        }]

\tikzstyle{meter}=[draw=black, inner sep=10, shape=rectangle,path picture={
                \draw[black] ([shift={(.1,.3)}] path picture bounding box.south west) to[bend left=50] ([shift={(-.1,.3)}] path picture bounding box.south east);
                \draw[black,-latex] ([shift={(0,.1)}]path picture bounding box.south) -- ([shift={(.3,-.1)}]path picture bounding box.north);
                }]

\tikzstyle{black dash edge}=[-, dashed, dash pattern=on 2pt off 0.5pt, thick, draw={rgb,255: red,200; green,200; blue,200}]
\tikzstyle{hadamard edge}=[-, dashed, dash pattern=on 2pt off 0.5pt, thick, draw={rgb,255: red,68; green,136; blue,255}]
\tikzstyle{dash edge}=[-, dashed, dash pattern=on 2pt off 0.5pt, thick, draw={rgb,255: red,180; green,180; blue,180}]
\tikzstyle{red dash edge}=[-, dashed, dash pattern=on 2pt off 0.5pt, thick, draw={rgb,255: red,255; green,60; blue,60}]
\tikzstyle{red edge}=[-, thick, draw={rgb,255: red,255; green,60; blue,60}]
\tikzstyle{brace edge}=[-, tikzit draw=blue, decorate, decoration={brace,amplitude=1mm,raise=-1mm}]
\tikzstyle{diredge}=[->]
\tikzstyle{mdiredge}=[<->]

\newcommand{\SXE}[1]{#1}

\begin{document}

\title{Multi-agent blind quantum computation \\ without universal cluster states}

\author{Shuxiang Cao}
\email{shuxiang.cao@physics.ox.ac.uk}
\affiliation{Department of Physics, Clarendon Laboratory, University of Oxford, OX1 3PU, UK}

\begin{abstract}
Blind quantum computation (BQC) protocols enable quantum algorithms to be executed on third-party quantum agents while keeping the data and algorithm confidential. The previous proposals for measurement-based BQC require preparing a highly entangled cluster state. In this paper, we show that such a requirement is not necessary. Our protocol only requires pre-shared bell pairs between delegated quantum agents, and there is no requirement for any classical or quantum information exchange between agents during the execution. Our proposal requires fewer quantum resources than previous proposals by eliminating the need for a universal cluster state.
\end{abstract}

\maketitle

\section{Introduction}

Quantum computers built from current technology are difficult to be miniaturized, and unlikely to become personal electronics such as a laptop or a cellphone~\cite{Krantz2019AQubits,Bruzewicz2019Trapped-ionChallenges,Kloeffel2013ProspectsDots,Slussarenko2019PhotonicReview}. Therefore, cloud-based services are considered the most applicable approach to offer the general public access to quantum computers. It is natural to ask whether the privacy of the quantum algorithm can be kept when one does not have complete control of the quantum hardware. Blind quantum computing (BQC) aims to solve this problem. Quantum algorithms can be executed with BQC protocols on third-party quantum agents while keeping the algorithm, data, and results confidential~\cite{Fitzsimons2017PrivateProtocols, Broadbent2008UniversalComputation}.

Here we discuss two ways to implement universal quantum computation. One is gate-based quantum computing (GBQC)~\cite{Michielsen2017BenchmarkingComputers}. This method starts with a pure quantum state, usually by resetting all qubits to zero. Then it transforms the quantum state using a sequence of quantum gates. The final output state carries the processed information. The other method is called measurement-based quantum computing (MBQC) or one-way quantum computation~\cite{Briegel2009Measurement-basedComputation,Raussendorf2006AComputer,Gross2007Measurement-basedModel,Kissinger2019UniversalMeasurements}. This method prepares a highly entangled state of multiple qubits, often referred to as a \textit{cluster state}~\cite{Nielsen2006Cluster-stateComputation}, then performs a sequence of measurements and corrections to implement computation. Eventually it can give the same result as the GBQC.

The Universal Blind Quantum Computing (UBQC) protocol was proposed in \cite{Broadbent2008UniversalComputation} based on the MBQC framework. UBQC protocol utilizes a universal cluster state and can be implemented by a semi-classical client with a single agent or an entirely classical client with multiple agents. 
There are other proposals implementing BQC with a single agent and an entirely classical client are possible, however, these proposals require some computational assumptions~\cite{Mahadev2018ClassicalComputations,BrakerskiQuantumClassical,Cojocaru2019QFactory:Preparation}. 

In this paper, we make use of a quantum graphical reasoning method, ZX-Calculus, to derive a BQC protocol that can be implemented with multiple agents and an entirely classical client. The UBQC protocol utilizes a universal cluster state, forcing all the information describing the algorithm to be encoded in the measurement axis. It sacrifices the ability to encode information into the entanglement structure between qubits. Contrarily, our method does have information encoded in the entanglement structure, and does not require a universal cluster state. This makes our protocol more resource-efficient.

This paper is arranged as follows:
Section \ref{sec:zx} describes ZX-calculus, a graphical quantum reasoning technique that we use to derive our result. Section \ref{sec:zx_bqc} explains our BQC protocol. Section \ref{sec:proof} gives proof of the correctness and secureness of our protocol.  
Section \ref{sec:discussion} discusses the compatibility with existing verification protocols, and quantifies the resource cost of our protocol and the UBQC protocol. Section \ref{sec:conclusion} summarizes the paper.

\section{Background}

\subsection{Universal Blind quantum computation\label{sec:BQC}}

The Universal Blind Quantum Computing (UBQC) protocol employs the MBQC method to implement BQC \cite{Broadbent2008UniversalComputation}. Under the MBQC framework, the algorithm can be described with only the entanglement structure between qubits and each qubit's measurement axis. To make the algorithm blind to the agents, the information each agent possesses, the entanglement structure, the measurement axis, and the measurement output of the agents, must not reveal any information about the algorithm. A valid BQC protocol must make this information independent from the delegated task.
 
To make the entanglement structure of the delegated task independent from the quantum algorithm, UBQC utilises a universal cluster state, which can implement arbitrary quantum algorithms with the same entanglement structure but a different measurement axis. Such a method concentrates the information describing the quantum algorithm on the measurement axis. UBQC protocol uses the brickwork cluster state to implement MBQC. Different quantum gates can be implemented by measuring the cluster state with different angles in sequence. For the brickwork state, the qubits are measured from left to right. Based on the measurement result, \textit{corrections} is applied to the following qubits on each step. The calculated result is then stored in the qubit on the right end of the brickwork cluster state and can be further processed by piling up more elementary components, or more ``bricks''.

\begin{example}\label{fig:UBC_brick_work}
Brickwork cluster state and MBQC with brickwork resource state. Each node denotes a qubit prepared in $\ket{+}$ state. Wires connecting two qubits denote an entanglement that is generated by applying a CZ gate between two qubits. The result is stored inside the very right qubits after measuring each qubit from left to right. The angle inside each node represents the angle of measurement that would be applied to the corresponding qubit. (a) The layout of a typical brickwork cluster state. The grey square shows a fundamental element of the brickwork cluster state. (b) Implement a Hadamard gate. The square on the left side denotes the qubits that hold the computation output. (c) Implement T ($\pi/8$) gate. (d) Implement identity gate. (e) Implement a CNOT gate.
\vskip 1em
\centering

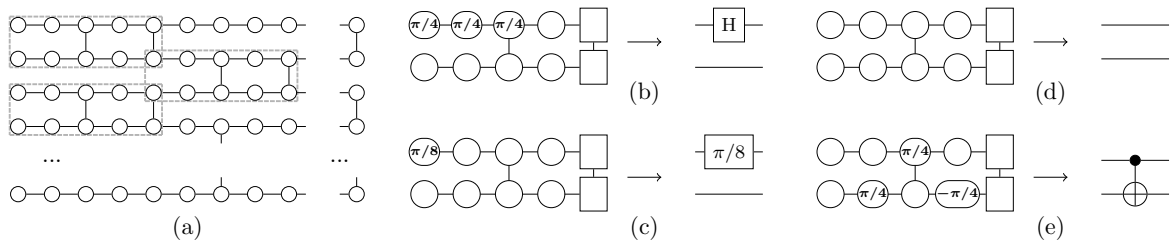
\[\scalebox{0.9}{\begin{tikzpicture}
	\begin{pgfonlayer}{nodelayer}
		\node [style=vertex set] (0) at (-11, 7) {};
		\node [style=vertex set] (1) at (-10, 7) {};
		\node [style=vertex set] (2) at (-11, 6) {};
		\node [style=vertex set] (3) at (-10, 6) {};
		\node [style=vertex set] (4) at (-9, 7) {};
		\node [style=vertex set] (5) at (-8, 7) {};
		\node [style=vertex set] (6) at (-7, 7) {};
		\node [style=vertex set] (7) at (-9, 6) {};
		\node [style=vertex set] (8) at (-8, 6) {};
		\node [style=vertex set] (9) at (-7, 6) {};
		\node [style=vertex set] (10) at (-6, 7) {};
		\node [style=vertex set] (11) at (-6, 6) {};
		\node [style=vertex set] (12) at (-5, 7) {};
		\node [style=vertex set] (13) at (-5, 6) {};
		\node [style=vertex set] (14) at (-3, 7) {};
		\node [style=vertex set] (16) at (-3, 6) {};
		\node [style=vertex set] (28) at (-7, 5) {};
		\node [style=vertex set] (29) at (-6, 5) {};
		\node [style=vertex set] (30) at (-7, 4) {};
		\node [style=vertex set] (31) at (-6, 4) {};
		\node [style=vertex set] (32) at (-5, 5) {};
		\node [style=vertex set] (33) at (-4, 5) {};
		\node [style=vertex set] (34) at (-3, 5) {};
		\node [style=vertex set] (35) at (-5, 4) {};
		\node [style=vertex set] (36) at (-4, 4) {};
		\node [style=vertex set] (37) at (-3, 4) {};
		\node [style=vertex set] (42) at (-4, 7) {};
		\node [style=vertex set] (43) at (-4, 6) {};
		\node [style=vertex set] (44) at (-9, 5) {};
		\node [style=vertex set] (45) at (-8, 5) {};
		\node [style=vertex set] (46) at (-9, 4) {};
		\node [style=vertex set] (47) at (-8, 4) {};
		\node [style=vertex set] (48) at (-11, 5) {};
		\node [style=vertex set] (49) at (-10, 5) {};
		\node [style=vertex set] (50) at (-11, 4) {};
		\node [style=vertex set] (51) at (-10, 4) {};
		\node [style=none] (60) at (-5, 3.5) {};
		\node [style=vertex set] (62) at (-7, 2) {};
		\node [style=vertex set] (63) at (-6, 2) {};
		\node [style=vertex set] (64) at (-5, 2) {};
		\node [style=vertex set] (65) at (-4, 2) {};
		\node [style=vertex set] (66) at (-3, 2) {};
		\node [style=vertex set] (69) at (-9, 2) {};
		\node [style=vertex set] (70) at (-8, 2) {};
		\node [style=vertex set] (71) at (-11, 2) {};
		\node [style=vertex set] (72) at (-10, 2) {};
		\node [style=none] (77) at (-5, 2.5) {};
		\node [style=none] (79) at (-10, 3) {...};
		\node [style=none] (81) at (-1.5, 3) {...};
		\node [style=vertex set] (82) at (-1, 7) {};
		\node [style=vertex set] (83) at (-1, 6) {};
		\node [style=vertex set] (84) at (-1, 5) {};
		\node [style=vertex set] (85) at (-1, 4) {};
		\node [style=vertex set] (86) at (-1, 2) {};
		\node [style=none] (87) at (-1.5, 7) {};
		\node [style=none] (88) at (-1.5, 6) {};
		\node [style=none] (89) at (-1.5, 5) {};
		\node [style=none] (90) at (-1.5, 4) {};
		\node [style=none] (91) at (-1.5, 2) {};
		\node [style=none] (97) at (-1, 2.5) {};
		\node [style=none] (98) at (-6, 1) {(a)};
		\node [style=none] (99) at (-11.25, 7.25) {};
		\node [style=none] (100) at (-11.25, 5.75) {};
		\node [style=none] (101) at (-6.75, 5.75) {};
		\node [style=none] (102) at (-6.75, 7.25) {};
		\node [style=none] (103) at (-7.25, 6.25) {};
		\node [style=none] (104) at (-7.25, 4.75) {};
		\node [style=none] (105) at (-2.75, 4.75) {};
		\node [style=none] (106) at (-2.75, 6.25) {};
		\node [style=none] (107) at (-6.75, 5.25) {};
		\node [style=none] (108) at (-11.25, 5.25) {};
		\node [style=none] (109) at (-11.25, 3.75) {};
		\node [style=none] (110) at (-6.75, 3.75) {};
		\node [style=none] (111) at (-2.5, 2) {};
		\node [style=none] (112) at (-2.5, 4) {};
		\node [style=none] (113) at (-2.5, 5) {};
		\node [style=none] (114) at (-2.5, 6) {};
		\node [style=none] (115) at (-2.5, 7) {};
		\node [style=box] (145) at (6, 7) {};
		\node [style=box] (146) at (6, 5.75) {};
		\node [style=vertex set phase] (157) at (1, 5.75) {};
		\node [style=vertex set phase] (158) at (2.25, 5.75) {};
		\node [style=vertex set phase] (159) at (3.5, 5.75) {};
		\node [style=vertex set phase] (160) at (4.75, 5.75) {};
		\node [style=vertex set phase] (161) at (4.75, 7) {};
		\node [style=vertex set phase] (162) at (3.5, 7) {$\pi/4$};
		\node [style=vertex set phase] (163) at (2.25, 7) {$\pi/4$};
		\node [style=vertex set phase] (164) at (1, 7) {$\pi/4$};
		\node [style=box] (165) at (6, 3.25) {};
		\node [style=box] (166) at (6, 2) {};
		\node [style=vertex set phase] (167) at (1, 2) {};
		\node [style=vertex set phase] (168) at (2.25, 2) {};
		\node [style=vertex set phase] (169) at (3.5, 2) {};
		\node [style=vertex set phase] (170) at (4.75, 2) {};
		\node [style=vertex set phase] (171) at (4.75, 3.25) {};
		\node [style=vertex set phase] (172) at (3.5, 3.25) {};
		\node [style=vertex set phase] (173) at (2.25, 3.25) {};
		\node [style=vertex set phase] (174) at (1, 3.25) {$\pi/8$};
		\node [style=box] (175) at (18, 7) {};
		\node [style=box] (176) at (18, 5.75) {};
		\node [style=vertex set phase] (177) at (13, 5.75) {};
		\node [style=vertex set phase] (178) at (14.25, 5.75) {};
		\node [style=vertex set phase] (179) at (15.5, 5.75) {};
		\node [style=vertex set phase] (180) at (16.75, 5.75) {};
		\node [style=vertex set phase] (181) at (16.75, 7) {};
		\node [style=vertex set phase] (182) at (15.5, 7) {};
		\node [style=vertex set phase] (183) at (14.25, 7) {};
		\node [style=vertex set phase] (184) at (13, 7) {};
		\node [style=box] (185) at (18, 3.25) {};
		\node [style=box] (186) at (18, 2) {};
		\node [style=vertex set phase] (187) at (13, 2) {};
		\node [style=vertex set phase] (188) at (14.25, 2) {$\pi/4$};
		\node [style=vertex set phase] (189) at (15.5, 2) {};
		\node [style=vertex set phase] (190) at (16.75, 2) {$-\pi/4$};
		\node [style=vertex set phase] (191) at (16.75, 3.25) {};
		\node [style=vertex set phase] (192) at (15.5, 3.25) {$\pi/4$};
		\node [style=vertex set phase] (193) at (14.25, 3.25) {};
		\node [style=vertex set phase] (194) at (13, 3.25) {};
		\node [style=none] (195) at (7, 6.5) {};
		\node [style=none] (196) at (8, 6.5) {};
		\node [style=none] (197) at (7, 2.5) {};
		\node [style=none] (198) at (8, 2.5) {};
		\node [style=none] (199) at (19, 6.5) {};
		\node [style=none] (200) at (20, 6.5) {};
		\node [style=none] (201) at (19, 2.5) {};
		\node [style=none] (202) at (20, 2.5) {};
		\node [style=none] (203) at (9, 7) {};
		\node [style=none] (204) at (11, 7) {};
		\node [style=none] (205) at (9, 5.75) {};
		\node [style=none] (206) at (11, 5.75) {};
		\node [style=none] (207) at (9, 3.25) {};
		\node [style=none] (208) at (11, 3.25) {};
		\node [style=none] (209) at (9, 2) {};
		\node [style=none] (210) at (11, 2) {};
		\node [style=none] (211) at (21, 7) {};
		\node [style=none] (212) at (23, 7) {};
		\node [style=none] (213) at (21, 6) {};
		\node [style=none] (214) at (23, 6) {};
		\node [style=none] (215) at (21, 3) {};
		\node [style=none] (216) at (23, 3) {};
		\node [style=none] (217) at (21, 2) {};
		\node [style=none] (218) at (23, 2) {};
		\node [style=box] (219) at (10, 7) {H};
		\node [style=box] (220) at (10, 3.25) {$\pi/8$};
		\node [style=vertex] (221) at (22, 3) {};
		\node [style=target] (222) at (22, 2) {};
		\node [style=none] (223) at (7.5, 5) {(b)};
		\node [style=none] (224) at (19.5, 5) {(d)};
		\node [style=none] (225) at (7.5, 1) {(c)};
		\node [style=none] (226) at (19.5, 1) {(e)};
	\end{pgfonlayer}
	\begin{pgfonlayer}{edgelayer}
		\draw (4) to (7);
		\draw (6) to (9);
		\draw (0) to (1);
		\draw (1) to (4);
		\draw (4) to (5);
		\draw (5) to (6);
		\draw (6) to (10);
		\draw (10) to (12);
		\draw (13) to (11);
		\draw (11) to (9);
		\draw (9) to (8);
		\draw (8) to (7);
		\draw (7) to (3);
		\draw (3) to (2);
		\draw (28) to (29);
		\draw (29) to (32);
		\draw (32) to (33);
		\draw (33) to (34);
		\draw (37) to (36);
		\draw (36) to (35);
		\draw (35) to (31);
		\draw (31) to (30);
		\draw (13) to (43);
		\draw (12) to (42);
		\draw (42) to (14);
		\draw (43) to (16);
		\draw (44) to (45);
		\draw (47) to (46);
		\draw (48) to (49);
		\draw (51) to (50);
		\draw (49) to (44);
		\draw (51) to (46);
		\draw (45) to (28);
		\draw (47) to (30);
		\draw (13) to (32);
		\draw (16) to (34);
		\draw (44) to (46);
		\draw (28) to (30);
		\draw (35) to (60.center);
		\draw (66) to (65);
		\draw (65) to (64);
		\draw (64) to (63);
		\draw (63) to (62);
		\draw (70) to (69);
		\draw (72) to (71);
		\draw (72) to (69);
		\draw (70) to (62);
		\draw (64) to (77.center);
		\draw (87.center) to (82);
		\draw (88.center) to (83);
		\draw (89.center) to (84);
		\draw (90.center) to (85);
		\draw (91.center) to (86);
		\draw (82) to (83);
		\draw (84) to (85);
		\draw (97.center) to (86);
		\draw [style=dash edge] (99.center) to (102.center);
		\draw [style=dash edge] (102.center) to (101.center);
		\draw [style=dash edge] (101.center) to (100.center);
		\draw [style=dash edge] (100.center) to (99.center);
		\draw [style=dash edge] (103.center) to (104.center);
		\draw [style=dash edge] (104.center) to (105.center);
		\draw [style=dash edge] (105.center) to (106.center);
		\draw [style=dash edge] (106.center) to (103.center);
		\draw [style=dash edge] (108.center) to (109.center);
		\draw [style=dash edge] (109.center) to (110.center);
		\draw [style=dash edge] (110.center) to (107.center);
		\draw [style=dash edge] (107.center) to (108.center);
		\draw (14) to (115.center);
		\draw (16) to (114.center);
		\draw (34) to (113.center);
		\draw (37) to (112.center);
		\draw (66) to (111.center);
		\draw (145) to (146);
		\draw (157) to (158);
		\draw (164) to (163);
		\draw (163) to (162);
		\draw (158) to (159);
		\draw (159) to (160);
		\draw (162) to (161);
		\draw (161) to (145);
		\draw (160) to (146);
		\draw (162) to (159);
		\draw (165) to (166);
		\draw (167) to (168);
		\draw (174) to (173);
		\draw (173) to (172);
		\draw (168) to (169);
		\draw (169) to (170);
		\draw (172) to (171);
		\draw (171) to (165);
		\draw (170) to (166);
		\draw (172) to (169);
		\draw (175) to (176);
		\draw (177) to (178);
		\draw (184) to (183);
		\draw (183) to (182);
		\draw (178) to (179);
		\draw (179) to (180);
		\draw (182) to (181);
		\draw (181) to (175);
		\draw (180) to (176);
		\draw (182) to (179);
		\draw (185) to (186);
		\draw (187) to (188);
		\draw (194) to (193);
		\draw (193) to (192);
		\draw (188) to (189);
		\draw (189) to (190);
		\draw (192) to (191);
		\draw (191) to (185);
		\draw (190) to (186);
		\draw (192) to (189);
		\draw [style=diredge] (195.center) to (196.center);
		\draw [style=diredge] (197.center) to (198.center);
		\draw [style=diredge] (199.center) to (200.center);
		\draw [style=diredge] (201.center) to (202.center);
		\draw (211.center) to (212.center);
		\draw (213.center) to (214.center);
		\draw (215.center) to (221);
		\draw (221) to (216.center);
		\draw (217.center) to (222);
		\draw (222) to (218.center);
		\draw (221) to (222);
		\draw (209.center) to (210.center);
		\draw (207.center) to (220);
		\draw (220) to (208.center);
		\draw (203.center) to (219);
		\draw (219) to (204.center);
		\draw (205.center) to (206.center);
	\end{pgfonlayer}
\end{tikzpicture}
}\]
\end{example}

Usually the measurement axis is defined by doing a single qubit rotation before the physically implementable measurement, which is usually the Pauli Z axis. To make the measurement axis independent to the quantum algorithm, a second agent is introduced to implement all or part of the single qubit rotation. When only one remote agent is available, UBQC protocol requires the client to be semi-classic; that is, the client can manipulate a minimum of a single qubit. It also requires the remote server to exchange quantum information by physically swapping qubits or establishing new entanglement. See example \ref{fig:UBQC_protocols}(a). In the original proposal, known as the ``prepare-and-measure'' method, the semi-classical agent effectively prepares the measurement angle. The agent and the client share the entanglement of each qubit, and the client measures its qubit at a random angle. This random angle would be ``teleported'' to the agent and affect the cluster state. The agent only needs to initialise the qubit into a superposition state and directly measure the qubit without rotating the qubit ~\cite{Broadbent2008UniversalComputation}. Alternatively, the agent can provide the cluster state and send the state back to the agent, known as the ``measurement-only'' method. Only the agent has access to the measurement angle~\cite{Morimae2013BlindMeasurements}. The UBQC protocols can also be implemented with multiple remote quantum agents and a purely classical client when the two agents' communication is restricted. See example \ref{fig:UBQC_protocols}(b). A uniformly distributed measurement axis for the delegated agent can be implemented on the first agent by simply requesting the second agent to measure their entangled qubits from a random axis. Then the computation can continue with the same method for a single agent UBQC. 

The measurement outcome for the "prepare-and-measure" approach is obfuscated by randomly flipping the outcome distribution during the measurement. Such obfuscation can be done by randomly choosing to measure at its original or with a $\pi$ difference. The measurement outcome would flip when the measurement angle is chosen with $\pi$ difference. Then the client classically restores the distribution after the measurement. Since the agent does not know if the distribution has been flipped or not, it can only observe a uniform distribution of 0 and 1 outcomes.

\begin{example}\label{fig:UBQC_protocols}
    Two protocols of universal blind quantum computation (UBQC). Both methods execute quantum algorithms with MBQC on the brickwork cluster state prepared on the remote agents. (a) Protocol with a semi-classical client and a single remote agent. The client can manipulate only one qubit and exchange qubits with the agent—the client prepares phase or measures the qubit at a random angle. The agent would know the actual rotation angle obfuscated by this random angle. (b) Protocol with a full classical client and multiple remote agents with shared entanglement. A second agent is introduced to replace the semi-classical client.

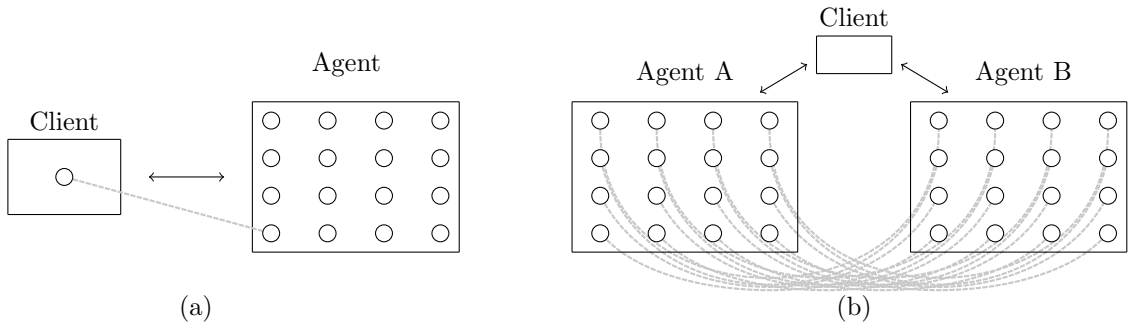
\centering
\begin{tikzpicture}
	\begin{pgfonlayer}{nodelayer}
		\node [style=vertex set] (44) at (5.75, 4) {};
		\node [style=vertex set] (46) at (5.75, 3) {};
		\node [style=vertex set] (64) at (5.75, 2) {};
		\node [style=vertex set] (66) at (5.75, 1) {};
		\node [style=none] (72) at (5, 4.5) {};
		\node [style=none] (73) at (5, 0.5) {};
		\node [style=none] (74) at (11, 0.5) {};
		\node [style=none] (75) at (11, 4.5) {};
		\node [style=vertex set] (76) at (14.75, 4) {};
		\node [style=vertex set] (78) at (14.75, 3) {};
		\node [style=vertex set] (96) at (14.75, 2) {};
		\node [style=vertex set] (98) at (14.75, 1) {};
		\node [style=none] (104) at (14, 4.5) {};
		\node [style=none] (105) at (14, 0.5) {};
		\node [style=none] (106) at (19.75, 0.5) {};
		\node [style=none] (107) at (19.75, 4.5) {};
		\node [style=vertex set] (108) at (7.25, 4) {};
		\node [style=vertex set] (109) at (7.25, 3) {};
		\node [style=vertex set] (110) at (7.25, 2) {};
		\node [style=vertex set] (111) at (7.25, 1) {};
		\node [style=vertex set] (112) at (16.25, 4) {};
		\node [style=vertex set] (113) at (16.25, 3) {};
		\node [style=vertex set] (114) at (16.25, 2) {};
		\node [style=vertex set] (115) at (16.25, 1) {};
		\node [style=vertex set] (116) at (8.75, 4) {};
		\node [style=vertex set] (117) at (8.75, 3) {};
		\node [style=vertex set] (118) at (8.75, 2) {};
		\node [style=vertex set] (119) at (8.75, 1) {};
		\node [style=vertex set] (120) at (17.75, 4) {};
		\node [style=vertex set] (121) at (17.75, 3) {};
		\node [style=vertex set] (122) at (17.75, 2) {};
		\node [style=vertex set] (123) at (17.75, 1) {};
		\node [style=vertex set] (124) at (10.25, 4) {};
		\node [style=vertex set] (125) at (10.25, 3) {};
		\node [style=vertex set] (126) at (10.25, 2) {};
		\node [style=vertex set] (127) at (10.25, 1) {};
		\node [style=vertex set] (128) at (19.25, 4) {};
		\node [style=vertex set] (129) at (19.25, 3) {};
		\node [style=vertex set] (130) at (19.25, 2) {};
		\node [style=vertex set] (131) at (19.25, 1) {};
		\node [style=vertex set] (132) at (-3, 4) {};
		\node [style=vertex set] (133) at (-3, 3) {};
		\node [style=vertex set] (134) at (-3, 2) {};
		\node [style=vertex set] (135) at (-3, 1) {};
		\node [style=none] (136) at (-3.5, 4.5) {};
		\node [style=none] (137) at (-3.5, 0.5) {};
		\node [style=none] (138) at (2, 0.5) {};
		\node [style=none] (139) at (2, 4.5) {};
		\node [style=vertex set] (140) at (-1.5, 4) {};
		\node [style=vertex set] (141) at (-1.5, 3) {};
		\node [style=vertex set] (142) at (-1.5, 2) {};
		\node [style=vertex set] (143) at (-1.5, 1) {};
		\node [style=vertex set] (144) at (0, 4) {};
		\node [style=vertex set] (145) at (0, 3) {};
		\node [style=vertex set] (146) at (0, 2) {};
		\node [style=vertex set] (147) at (0, 1) {};
		\node [style=vertex set] (148) at (1.5, 4) {};
		\node [style=vertex set] (149) at (1.5, 3) {};
		\node [style=vertex set] (150) at (1.5, 2) {};
		\node [style=vertex set] (151) at (1.5, 1) {};
		\node [style=none] (152) at (11.5, 6.25) {};
		\node [style=none] (153) at (11.5, 5.25) {};
		\node [style=none] (154) at (13.5, 5.25) {};
		\node [style=none] (155) at (13.5, 6.25) {};
		\node [style=none] (156) at (10, 4.75) {};
		\node [style=none] (157) at (15, 4.75) {};
		\node [style=none] (158) at (11.25, 5.5) {};
		\node [style=none] (159) at (13.75, 5.5) {};
		\node [style=none] (160) at (12.5, 6.75) {Client};
		\node [style=none] (161) at (8, 5.25) {Agent A};
		\node [style=none] (162) at (17, 5.25) {Agent B};
		\node [style=none] (163) at (-10, 3.5) {};
		\node [style=none] (164) at (-10, 1.5) {};
		\node [style=none] (165) at (-7, 1.5) {};
		\node [style=none] (166) at (-7, 3.5) {};
		\node [style=none] (167) at (-8.5, 4) {Client};
		\node [style=none] (168) at (-1, 5.5) {Agent};
		\node [style=none] (170) at (-6.25, 2.5) {};
		\node [style=none] (171) at (-4.25, 2.5) {};
		\node [style=vertex set] (172) at (-8.5, 2.5) {};
		\node [style=none] (173) at (-5, -1) {(a)};
		\node [style=none] (174) at (12.5, -1) {(b)};
	\end{pgfonlayer}
	\begin{pgfonlayer}{edgelayer}
		\draw [style=black dash edge, bend right=45, looseness=0.75] (66) to (98);
		\draw [style=black dash edge, bend right=60] (64) to (96);
		\draw [style=black dash edge, bend right=75, looseness=1.25] (46) to (78);
		\draw [style=black dash edge, bend right=90, looseness=1.50] (44) to (76);
		\draw [style=black dash edge, bend right=45, looseness=0.75] (111) to (115);
		\draw [style=black dash edge, bend right=60] (110) to (114);
		\draw [style=black dash edge, bend right=75, looseness=1.25] (109) to (113);
		\draw [style=black dash edge, bend right=90, looseness=1.50] (108) to (112);
		\draw [style=black dash edge, bend right=45, looseness=0.75] (119) to (123);
		\draw [style=black dash edge, bend right=60] (118) to (122);
		\draw [style=black dash edge, bend right=75, looseness=1.25] (117) to (121);
		\draw [style=black dash edge, bend right=90, looseness=1.50] (116) to (120);
		\draw [style=black dash edge, bend right=45, looseness=0.75] (127) to (131);
		\draw [style=black dash edge, bend right=60] (126) to (130);
		\draw [style=black dash edge, bend right=75, looseness=1.25] (125) to (129);
		\draw [style=black dash edge, bend right=90, looseness=1.50] (124) to (128);
		\draw (72.center) to (73.center);
		\draw (73.center) to (74.center);
		\draw (75.center) to (74.center);
		\draw (75.center) to (72.center);
		\draw (104.center) to (105.center);
		\draw (107.center) to (104.center);
		\draw (107.center) to (106.center);
		\draw (106.center) to (105.center);
		\draw (136.center) to (137.center);
		\draw (139.center) to (136.center);
		\draw (139.center) to (138.center);
		\draw (138.center) to (137.center);
		\draw (152.center) to (155.center);
		\draw (153.center) to (154.center);
		\draw (153.center) to (152.center);
		\draw (155.center) to (154.center);
		\draw [style=mdiredge] (158.center) to (156.center);
		\draw [style=mdiredge] (159.center) to (157.center);
		\draw (163.center) to (166.center);
		\draw (164.center) to (165.center);
		\draw (164.center) to (163.center);
		\draw (166.center) to (165.center);
		\draw [style=mdiredge] (170.center) to (171.center);
		\draw [style=black dash edge] (172) to (135);
	\end{pgfonlayer}
\end{tikzpicture}
\end{example}

\SXE{
It is worth mentioning that a circuit-based BQC method proposed in \cite{Sano2020BlindComputer} utilises a similar philosophy as the UBQC protocol. A ``universal circuit'' that can implement arbitrary operation by modifying the single-qubit gate rotation angle has been introduced in the proposal. The entanglement structure is then irrelevant to the circuit on the agent, and the rotation angles are obfuscated with quantum computing on encrypted data (QCED) \cite{Broadbent2018HowComputation,Fisher2014QuantumData,Broadbent2015DelegatingComputations}, which requires the agent to exchange quantum information with the client. The circuit-based protocol computes the cluster state in a circuit-based manner. However, it still requires exchanging the same amount of quantum information between the client and agent as the UBQC protocol to implement ``correction''.}

\subsection{ZX-Calculus\label{sec:zx}}

In this section, we provide a brief review of the ZX-Calculus~\cite{Coecke2018PicturingProcesses}. The ZX-Calculus is a diagrammatic method for reasoning the linear maps of quantum operations. With the gate representation of quantum computation, we decompose a unitary operation into a sequence of predefined gates; with ZX-Calculus, we decompose the unitary into a network, the so-called \textit{ZX-diagram}, consists of red and green spiders. In the following discussion, we ignore the scalars of the ZX-diagrams.

\paragraph{ZX-diagram}

A ZX-diagram consists of \textit{wires} and \textit{spiders}, corresponding to legs and tensors in the tensor network language~\SXE{\cite{Biamonte2017TensorNutshell, PhysRevLett.125.150504}}. There are two types of spiders: $Z$ spiders and $X$ spiders noted as green and red dots. The spiders are defined as a tensor parameterized by a single phase variable. The opened wire can be considered an input or output of the ZX-diagram. The summary of the definition of the basic building blocks of ZX-diagram is shown in example \ref{fig:zx-definations}. A quantum circuit can be easily rewritten to a ZX-diagram with rules provided in example \ref{fig:zx-definations}. The $X$ gate can be replaced with a red spider, and the $Z$ gate can be replaced with a green spider. The rotation angle is represented as the phase of each spider. 

\begin{example}\label{fig:zx-definations}
Basic building blocks of ZX-diagram. ZX-diagram is a notation representing tensor networks. Any quantum circuit can be converted into a quantum tensor network and further represented by ZX-diagram~\cite{Backens2016CompletenessZX-calculus,Jeandel2020COMPLETENESSZX-CALCULUS}. ZX-diagram consists of \textit{spiders}, which is a tensor with constraints above. Standard quantum circuits can be converted into ZX-diagram with the following rules:

\centering
$\begin{tikzpicture}
	\begin{pgfonlayer}{nodelayer}
		\node [style=Z phase dot] (0) at (0, 0) {$\alpha$};
		\node [style=none] (1) at (1.25, 1) {};
		\node [style=none] (2) at (-1.25, 1) {};
		\node [style=none] (3) at (-1.25, -1) {};
		\node [style=none] (4) at (1.25, -1) {};
		\node [style=none] (5) at (1.25, 0.5) {};
		\node [style=none] (6) at (-1.25, 0.5) {};
		\node [style=none, rotate=90] (7) at (-1, -0.25) {...};
		\node [style=none, rotate=90] (8) at (1, -0.25) {...};
	\end{pgfonlayer}
	\begin{pgfonlayer}{edgelayer}
		\draw [in=-141, out=0, looseness=0.75] (3.center) to (0);
		\draw [in=180, out=-39, looseness=0.75] (0) to (4.center);
		\draw [in=180, out=22, looseness=0.75] (0) to (5.center);
		\draw [in=180, out=39, looseness=0.75] (0) to (1.center);
		\draw [in=0, out=158, looseness=0.75] (0) to (6.center);
		\draw [in=141, out=0, looseness=0.75] (2.center) to (0);
	\end{pgfonlayer}
\end{tikzpicture} \ := \ket{0...0}\bra{0...0} +
e^{i \alpha} \ket{1...1}\bra{1...1} = \begin{bmatrix}
1 &  & \\
 & ... & \\
 &  & e^{i \alpha}
\end{bmatrix}$

$\begin{tikzpicture}
	\begin{pgfonlayer}{nodelayer}
		\node [style=X phase dot] (0) at (0, 0) {$\alpha$};
		\node [style=none] (1) at (1.25, 1) {};
		\node [style=none] (2) at (-1.25, 1) {};
		\node [style=none] (3) at (-1.25, -1) {};
		\node [style=none] (4) at (1.25, -1) {};
		\node [style=none] (5) at (1.25, 0.5) {};
		\node [style=none] (6) at (-1.25, 0.5) {};
		\node [style=none, rotate=90] (7) at (-1, -0.25) {...};
		\node [style=none, rotate=90] (8) at (1, -0.25) {...};
	\end{pgfonlayer}
	\begin{pgfonlayer}{edgelayer}
		\draw [in=-141, out=0, looseness=0.75] (3.center) to (0);
		\draw [in=180, out=-39, looseness=0.75] (0) to (4.center);
		\draw [in=180, out=22, looseness=0.75] (0) to (5.center);
		\draw [in=180, out=39, looseness=0.75] (0) to (1.center);
		\draw [in=0, out=158, looseness=0.75] (0) to (6.center);
		\draw [in=141, out=0, looseness=0.75] (2.center) to (0);
	\end{pgfonlayer}
\end{tikzpicture} \ := \ \ket{+...+}\bra{+...+} +
e^{i \alpha} \ket{-...-}\bra{-...-} = \ H^{\otimes n} \begin{bmatrix}
1 &  & \\
 & ... & \\
 &  & e^{i \alpha}
\end{bmatrix} H^{\otimes n}$ \\

\vskip 1em

    \begin{tikzpicture}
	\begin{pgfonlayer}{nodelayer}
		\node [style=none] (0) at (3, -2) {};
		\node [style=none] (1) at (6, -2) {};
		\node [style=none] (2) at (3, -4) {};
		\node [style=none] (3) at (6, -4) {};
		\node [style=vertex] (4) at (4.5, -2) {};
		\node [style=target] (5) at (4.5, -4) {};
		\node [style=none] (6) at (6.5, -3) {=};
		\node [style=none] (7) at (7, -2) {};
		\node [style=none] (8) at (10, -2) {};
		\node [style=none] (9) at (7, -4) {};
		\node [style=none] (10) at (10, -4) {};
		\node [style=none] (19) at (2.5, 0) {=};
		\node [style=none] (52) at (-1, 0) {};
		\node [style=none] (53) at (2, 0) {};
		\node [style=box] (54) at (0.5, 0) {$Z(\alpha)$};
		\node [style=Z phase dot] (55) at (4.5, 0) {$\alpha$};
		\node [style=none] (56) at (3, 0) {};
		\node [style=none] (57) at (6, 0) {};
		\node [style=none] (58) at (10.5, 0) {=};
		\node [style=none] (59) at (7, 0) {};
		\node [style=none] (60) at (10, 0) {};
		\node [style=box] (61) at (8.5, 0) {$X(\alpha)$};
		\node [style=none] (63) at (11, 0) {};
		\node [style=none] (64) at (14, 0) {};
		\node [style=X phase dot] (65) at (12.5, 0) {$\alpha$};
		\node [style=none] (66) at (15, 0) {};
		\node [style=none] (67) at (17, 0) {};
		\node [style=box] (68) at (16, 0) {H};
		\node [style=none] (69) at (17.5, 0) {=};
		\node [style=none] (70) at (18, 0) {};
		\node [style=none] (71) at (20, 0) {};
		\node [style=hadamard] (72) at (19, 0) {};
		\node [style=none] (73) at (20.5, 0) {or};
		\node [style=none] (74) at (21, 0) {};
		\node [style=none] (75) at (22.5, 0) {};
		\node [style=none] (76) at (11, -2) {};
		\node [style=none] (77) at (14, -2) {};
		\node [style=none] (78) at (11, -4) {};
		\node [style=none] (79) at (14, -4) {};
		\node [style=vertex] (80) at (12.5, -2) {};
		\node [style=none] (82) at (14.5, -3) {=};
		\node [style=none] (83) at (15, -2) {};
		\node [style=none] (84) at (18, -2) {};
		\node [style=none] (85) at (15, -4) {};
		\node [style=none] (86) at (18, -4) {};
		\node [style=box] (89) at (12.5, -4) {$Z(\alpha)$};
		\node [style=Z phase dot] (90) at (16.5, -2) {$\frac{\alpha}{2}$};
		\node [style=Z phase dot] (91) at (16.5, -4) {$\frac{\alpha}{2}$};
		\node [style=X phase dot] (92) at (16.5, -3) {$\pi$};
		\node [style=Z phase dot] (93) at (17.5, -3) {$\frac{\alpha}{2}$};
		\node [style=Z dot] (94) at (8.5, -2) {};
		\node [style=X dot] (95) at (8.5, -4) {};
	\end{pgfonlayer}
	\begin{pgfonlayer}{edgelayer}
		\draw (4) to (0.center);
		\draw (4) to (1.center);
		\draw (2.center) to (5);
		\draw (5) to (3.center);
		\draw (4) to (5);
		\draw (52.center) to (54);
		\draw (54) to (53.center);
		\draw (56.center) to (55);
		\draw (55) to (57.center);
		\draw (59.center) to (61);
		\draw (61) to (60.center);
		\draw (63.center) to (65);
		\draw (65) to (64.center);
		\draw (66.center) to (68);
		\draw (68) to (67.center);
		\draw (70.center) to (72);
		\draw (72) to (71.center);
		\draw [style=hadamard edge] (74.center) to (75.center);
		\draw (80) to (76.center);
		\draw (80) to (77.center);
		\draw (80) to (89);
		\draw (89) to (78.center);
		\draw (89) to (79.center);
		\draw (90) to (92);
		\draw (92) to (93);
		\draw (92) to (91);
		\draw (83.center) to (90);
		\draw (90) to (84.center);
		\draw (85.center) to (91);
		\draw (91) to (86.center);
		\draw (7.center) to (94);
		\draw (94) to (8.center);
		\draw (94) to (95);
		\draw (95) to (9.center);
		\draw (95) to (10.center);
	\end{pgfonlayer}
\end{tikzpicture}
\end{example}

Instead of directly contracting the tensor network, ZX-Calculus provided a set of rules to manipulate a ZX-diagram while keeping them equivalent. ZX-calculus is complete on Clifford+T language with a set of rules are specified\cite{Backens2016CompletenessZX-calculus,SchroderdeWitt2014TheMechanics}. Some of these rules are shown in example~\ref{fig:zx-rules}.

\begin{example}\label{fig:zx-rules}
The power of ZX-Calculus is it derives a set of rules to transfer a ZX-diagram to another one while keeping them equivalent. Here we show several rules that we will use later. The $\spiderrule$ rule indicates that any two spiders with the same colour can be merged. The $\hadamardrule$ indicates that spider colour can be changed by adding a Hadamard spider on every wire of the spider. $\pirule$ indicates that a Pi operation with a different colour can be copied and moved to other wires while changing the sign of the spider's phase. Also, $\identityrule$ and $\hcancelrule$ can help generate or remove redundant spiders and Hadamard nodes. $\bialgrule$ the bialgebra rule. The Hopf law or Antipode law $\hopfrule$ shows two parallel Hadamard wire results cancelling each other.

\vskip 1em

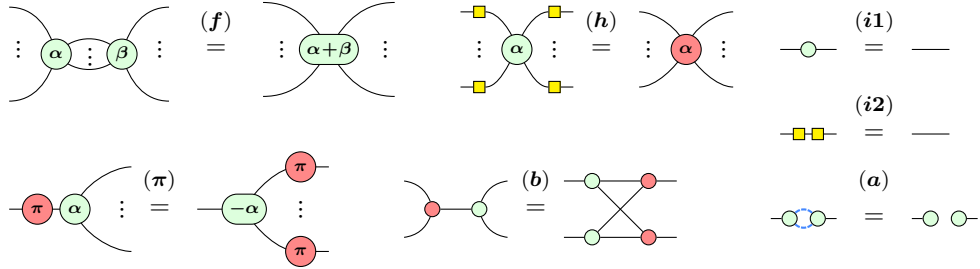
\centering
\begin{tikzpicture}
	\begin{pgfonlayer}{nodelayer}
		\node [style=Z phase dot] (0) at (-3, -0.125) {$\beta$};
		\node [style=none] (1) at (-1.75, 1.125) {};
		\node [style=none] (4) at (-1.75, -1.375) {};
		\node [style=none, rotate=90] (8) at (-2, 0) {...};
		\node [style=Z phase dot] (10) at (-4.75, -0.125) {$\alpha$};
		\node [style=none] (12) at (-6, 1.125) {};
		\node [style=none] (16) at (-6, -1.375) {};
		\node [style=none, rotate=90] (17) at (-5.75, 0) {...};
		\node [style=none] (24) at (-0.5, 0) {$=$};
		\node [style=none, rotate=90] (25) at (-3.85, -0.2) {...};
		\node [style=none] (26) at (0.75, 1.25) {};
		\node [style=none, rotate=90] (28) at (1.25, 0) {...};
		\node [style=none] (29) at (4.25, -1.25) {};
		\node [style=none, rotate=90] (30) at (4, 0) {...};
		\node [style=Z phase dot] (32) at (2.5, 0) {$\ \alpha\!+\!\beta\ $};
		\node [style=none] (33) at (4.25, 1.25) {};
		\node [style=none] (34) at (0.75, -1.25) {};
		\node [style=none] (35) at (-0.5, 0.75) {\spiderrule};
		\node [style=Z phase dot] (36) at (0.25, -4.25) {$\ -\alpha\ $};
		\node [style=none] (37) at (2.5, -3.125) {};
		\node [style=none] (38) at (-2, -4.25) {$=$};
		\node [style=none] (39) at (-1, -4.25) {};
		\node [style=none] (40) at (2.5, -5.375) {};
		\node [style=X phase dot] (42) at (1.75, -5.375) {$\pi$};
		\node [style=X phase dot] (43) at (-5.25, -4.25) {$\pi$};
		\node [style=none] (44) at (-2.75, -3.125) {};
		\node [style=Z phase dot] (45) at (-4.25, -4.25) {$\alpha$};
		\node [style=none, rotate=90] (46) at (-3, -4.25) {...};
		\node [style=none, rotate=90] (48) at (1.75, -4.25) {...};
		\node [style=none] (50) at (-6, -4.25) {};
		\node [style=none] (51) at (-2.75, -5.375) {};
		\node [style=X phase dot] (52) at (1.75, -3.125) {$\pi$};
		\node [style=none] (53) at (-2, -3.5) {\pirule};
		\node [style=hadamard] (69) at (6.5, 1) {};
		\node [style=none, rotate=90] (72) at (6.5, 0) {...};
		\node [style=none] (73) at (9.75, 0) {$=$};
		\node [style=hadamard] (74) at (6.5, -1) {};
		\node [style=none] (77) at (6, 1) {};
		\node [style=none] (81) at (10.75, -1) {};
		\node [style=none] (82) at (6, -1) {};
		\node [style=none, rotate=90] (83) at (11, 0) {...};
		\node [style=none] (84) at (10.75, 1) {};
		\node [style=none] (86) at (9.75, 0.75) {\hadamardrule};
		\node [style=Z dot] (87) at (15.25, 0) {};
		\node [style=none] (88) at (14.5, 0) {};
		\node [style=none] (89) at (16, 0) {};
		\node [style=none] (90) at (18, 0) {};
		\node [style=none] (91) at (19, 0) {};
		\node [style=none] (92) at (17, 0.75) {\identityrule};
		\node [style=none] (93) at (17, 0) {$=$};
		\node [style=none] (94) at (17, -2.25) {$=$};
		\node [style=none] (95) at (17, -1.5) {\hcancelrule};
		\node [style=none] (96) at (19, -2.25) {};
		\node [style=none] (97) at (16, -2.25) {};
		\node [style=none] (98) at (14.5, -2.25) {};
		\node [style=none] (99) at (18, -2.25) {};
		\node [style=hadamard] (100) at (15, -2.25) {};
		\node [style=hadamard] (101) at (15.5, -2.25) {};
		\node [style=hadamard] (118) at (8.5, 1) {};
		\node [style=none, rotate=90] (120) at (8.5, 0) {...};
		\node [style=hadamard] (121) at (8.5, -1) {};
		\node [style=Z phase dot] (122) at (7.5, 0) {$\alpha$};
		\node [style=none] (123) at (9, 1) {};
		\node [style=none] (126) at (9, -1) {};
		\node [style=X phase dot] (127) at (12, 0) {$\alpha$};
		\node [style=none] (128) at (13.25, -1) {};
		\node [style=none, rotate=90] (129) at (13, 0) {...};
		\node [style=none] (130) at (13.25, 1) {};
		\node [style=none] (131) at (14.25, -4.5) {};
		\node [style=none] (132) at (16, -4.5) {};
		\node [style=Z dot] (133) at (14.75, -4.5) {};
		\node [style=Z dot] (134) at (15.5, -4.5) {};
		\node [style=none] (135) at (17, -4.5) {=};
		\node [style=none] (136) at (17, -3.5) {\hopfrule};
		\node [style=none] (137) at (18, -4.5) {};
		\node [style=none] (138) at (19.75, -4.5) {};
		\node [style=Z dot] (139) at (18.5, -4.5) {};
		\node [style=Z dot] (140) at (19.25, -4.5) {};
		\node [style=none] (141) at (8, -3.5) {\bialgrule};
		\node [style=X dot] (142) at (11, -5) {};
		\node [style=none] (143) at (11.75, -3.5) {};
		\node [style=X dot] (144) at (5.25, -4.25) {};
		\node [style=Z dot] (145) at (6.5, -4.25) {};
		\node [style=none] (146) at (8.75, -5) {};
		\node [style=none] (147) at (4.5, -3.5) {};
		\node [style=none] (148) at (8.75, -3.5) {};
		\node [style=Z dot] (149) at (9.5, -3.5) {};
		\node [style=none] (150) at (7.25, -3.5) {};
		\node [style=none] (151) at (8, -4.25) {$=$};
		\node [style=none] (152) at (4.5, -5) {};
		\node [style=X dot] (153) at (11, -3.5) {};
		\node [style=none] (154) at (11.75, -5) {};
		\node [style=Z dot] (155) at (9.5, -5) {};
		\node [style=none] (156) at (7.25, -5) {};
	\end{pgfonlayer}
	\begin{pgfonlayer}{edgelayer}
		\draw [in=180, out=-75] (0) to (4.center);
		\draw [in=180, out=75] (0) to (1.center);
		\draw [in=-105, out=0] (16.center) to (10);
		\draw [in=105, out=0] (12.center) to (10);
		\draw [bend left=45] (10) to (0);
		\draw [bend left=45] (0) to (10);
		\draw [in=-120, out=0] (34.center) to (32);
		\draw [in=180, out=-60] (32) to (29.center);
		\draw [in=180, out=60] (32) to (33.center);
		\draw [in=120, out=0] (26.center) to (32);
		\draw [in=180, out=-75, looseness=0.75] (45) to (51.center);
		\draw [in=180, out=75, looseness=0.75] (45) to (44.center);
		\draw [in=180, out=0, looseness=0.50] (43) to (45);
		\draw (50.center) to (43);
		\draw [in=-60, out=180, looseness=0.75] (42) to (36);
		\draw [in=0, out=180, looseness=0.75] (36) to (39.center);
		\draw [in=60, out=180, looseness=0.75] (52) to (36);
		\draw (37.center) to (52);
		\draw (40.center) to (42);
		\draw (77.center) to (69);
		\draw (82.center) to (74);
		\draw (88.center) to (89.center);
		\draw (90.center) to (91.center);
		\draw (98.center) to (97.center);
		\draw (99.center) to (96.center);
		\draw [in=-60, out=180, looseness=0.75] (121) to (122);
		\draw [in=75, out=180, looseness=0.75] (118) to (122);
		\draw (123.center) to (118);
		\draw (126.center) to (121);
		\draw [in=-120, out=0, looseness=0.75] (74) to (122);
		\draw [in=105, out=0, looseness=0.75] (69) to (122);
		\draw [in=180, out=-60, looseness=0.75] (127) to (128.center);
		\draw [in=180, out=75, looseness=0.75] (127) to (130.center);
		\draw [in=0, out=-120, looseness=0.75] (127) to (81.center);
		\draw [in=0, out=105, looseness=0.75] (127) to (84.center);
		\draw (131.center) to (133);
		\draw (134) to (132.center);
		\draw [style=hadamard edge, bend right=45] (133) to (134);
		\draw [style=hadamard edge, bend right=45] (134) to (133);
		\draw (137.center) to (139);
		\draw (140) to (138.center);
		\draw (155) to (153);
		\draw (142) to (149);
		\draw (149) to (153);
		\draw (155) to (142);
		\draw (154.center) to (142);
		\draw (146.center) to (155);
		\draw (149) to (148.center);
		\draw (153) to (143.center);
		\draw [bend right] (152.center) to (144);
		\draw [bend right] (144) to (147.center);
		\draw (144) to (145);
		\draw [bend right] (145) to (156.center);
		\draw [bend left] (145) to (150.center);
	\end{pgfonlayer}
\end{tikzpicture}

\end{example}

Instead of representing the gates with both input wires and output wires, spiders can have just a single wire. With a single output wire, the spider represents a \textit{bra} notation. Such \textit{bra} notation denotes to a post-selection operation when extracting the diagram into a quantum circuit. While with a single input wire, the spider represents a \textit{ket} notation. Such \textit{ket} notation denotes preparation of the initial state of the quantum circuit. 

\begin{example}    \label{fig:post_selection}
Post-selection in ZX-diagram. The post-selection is represented by attaching a spider with no output to the output wires of the ZX-diagram. Apply a red dot denoting post select the $\ket{0}$ state, and a green dot means $\ket{+}$ state.

\centering
\[\scalebox{1}{\begin{tikzpicture}
	\begin{pgfonlayer}{nodelayer}
		\node [style=X dot] (101) at (-19.75, 1.5) {};
		\node [style=X dot] (102) at (-19.75, 0.5) {};
		\node [style=X dot] (103) at (-19.75, -1.5) {};
		\node [style=none] (104) at (-19.75, -0.5) {...};
		\node [style=none] (105) at (-19.25, 2.5) {};
		\node [style=none] (106) at (-19.25, -2.5) {};
		\node [style=none] (107) at (-16.5, 2.5) {};
		\node [style=none] (108) at (-16.5, -2.5) {};
		\node [style=none] (109) at (-18, -3) {$U$};
		\node [style=none] (110) at (-18, -0.5) {...};
		\node [style=Z dot] (111) at (-17, 1.5) {};
		\node [style=Z dot] (112) at (-17, 0.5) {};
		\node [style=Z dot] (113) at (-17, -1.5) {};
		\node [style=X dot] (114) at (-18.75, 1.5) {};
		\node [style=X dot] (115) at (-18.75, 0.5) {};
		\node [style=X dot] (116) at (-18.75, -1.5) {};
		\node [style=none] (117) at (-18.25, 1.75) {};
		\node [style=none] (118) at (-18.25, 0) {};
		\node [style=none] (119) at (-18.25, -0.75) {};
		\node [style=none] (120) at (-17.5, -1.75) {};
		\node [style=none] (121) at (-17.5, 0) {};
		\node [style=none] (122) at (-17.5, 1.75) {};
		\node [style=none] (125) at (-16, -1.5) {};
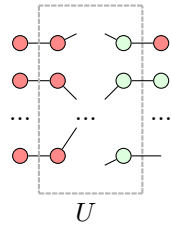
		\node [style=none] (126) at (-18, -4) {State $U\ket{0_n}$ with post-selection};
		\node [style=X dot] (127) at (-16, 1.5) {};
		\node [style=Z dot] (128) at (-16, 0.5) {};
		\node [style=none] (133) at (-16, -0.5) {...};
	\end{pgfonlayer}
	\begin{pgfonlayer}{edgelayer}
		\draw (101) to (114);
		\draw (102) to (115);
		\draw (103) to (116);
		\draw (114) to (117.center);
		\draw (115) to (118.center);
		\draw (122.center) to (111);
		\draw (116) to (119.center);
		\draw (121.center) to (112);
		\draw (120.center) to (113);
		\draw [style=dash edge] (105.center) to (106.center);
		\draw [style=dash edge] (106.center) to (108.center);
		\draw [style=dash edge] (108.center) to (107.center);
		\draw [style=dash edge] (107.center) to (105.center);
		\draw (113) to (125.center);
		\draw (111) to (127);
		\draw (112) to (128);
	\end{pgfonlayer}
\end{tikzpicture}}\]

\end{example}

A ZX-diagram can also represent a density matrix. For a pure state $\ket{\psi}$, the density matrix is $\rho = \ket{\psi}\bra{\psi}$, which is the tensor product of $\ket{\psi}$ and $\bra{\psi}$. In a ZX-diagram, a tensor product can be represented by putting two disconnected diagrams together. By writing $\ket{\psi}$ and $\bra{\psi}$ into the same diagram, we have the ZX-diagram of the density matrix $\rho$ shown in example~\ref{fig:density_matrix} (b).
 
\begin{example}    \label{fig:density_matrix}
The representation of a pure state density matrix. Putting two isolated ZX-diagram together gives the tensor product between to ZX-diagram. Suppose density matrix is $\rho=\ket{\psi}\otimes\bra{\psi}$, it can be represented by placing two ZX-diagram of $\ket{\psi}$ and $\bra{\psi}$ together.

\centering
\[\scalebox{1}{\begin{tikzpicture}
	\begin{pgfonlayer}{nodelayer}
		\node [style=X dot] (0) at (-11.75, 1.5) {};
		\node [style=X dot] (1) at (-11.75, 0.5) {};
		\node [style=X dot] (2) at (-11.75, -1.5) {};
		\node [style=none] (3) at (-11.75, -0.5) {...};
		\node [style=none] (4) at (-11.25, 2.5) {};
		\node [style=none] (5) at (-11.25, -2.5) {};
		\node [style=none] (6) at (-8.5, 2.5) {};
		\node [style=none] (7) at (-8.5, -2.5) {};
		\node [style=none] (8) at (-10, -3) {$U$};
		\node [style=none] (9) at (-10, -0.5) {...};
		\node [style=Z dot] (10) at (-9, 1.5) {};
		\node [style=Z dot] (11) at (-9, 0.5) {};
		\node [style=Z dot] (12) at (-9, -1.5) {};
		\node [style=X dot] (13) at (-10.75, 1.5) {};
		\node [style=X dot] (14) at (-10.75, 0.5) {};
		\node [style=X dot] (15) at (-10.75, -1.5) {};
		\node [style=none] (16) at (-10.25, 1.75) {};
		\node [style=none] (17) at (-10.25, 0) {};
		\node [style=none] (18) at (-10.25, -0.75) {};
		\node [style=none] (19) at (-9.5, -1.75) {};
		\node [style=none] (20) at (-9.5, 0) {};
		\node [style=none] (21) at (-9.5, 1.75) {};
		\node [style=none] (22) at (-8, 1.5) {};
		\node [style=none] (23) at (-8, 0.5) {};
		\node [style=none] (24) at (-8, -1.5) {};
		\node [style=X dot] (25) at (-3.25, 1.5) {};
		\node [style=X dot] (26) at (-3.25, 0.5) {};
		\node [style=X dot] (27) at (-3.25, -1.5) {};
		\node [style=none] (28) at (-3.25, -0.5) {...};
		\node [style=none] (29) at (-3.75, 2.5) {};
		\node [style=none] (30) at (-3.75, -2.5) {};
		\node [style=none] (31) at (-6.5, 2.5) {};
		\node [style=none] (32) at (-6.5, -2.5) {};
		\node [style=none] (33) at (-5, -3) {$U^{\dag}$};
		\node [style=none] (34) at (-5, -0.5) {...};
		\node [style=Z dot] (35) at (-6, 1.5) {};
		\node [style=Z dot] (36) at (-6, 0.5) {};
		\node [style=Z dot] (37) at (-6, -1.5) {};
		\node [style=X dot] (38) at (-4.25, 1.5) {};
		\node [style=X dot] (39) at (-4.25, 0.5) {};
		\node [style=X dot] (40) at (-4.25, -1.5) {};
		\node [style=none] (41) at (-4.75, 1.75) {};
		\node [style=none] (42) at (-4.75, 0) {};
		\node [style=none] (43) at (-4.75, -0.75) {};
		\node [style=none] (44) at (-5.5, -1.75) {};
		\node [style=none] (45) at (-5.5, 0) {};
		\node [style=none] (46) at (-5.5, 1.75) {};
		\node [style=none] (47) at (-7, 1.5) {};
		\node [style=none] (48) at (-7, 0.5) {};
		\node [style=none] (49) at (-7, -1.5) {};
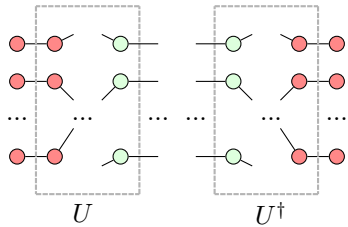
		\node [style=none] (99) at (-7.5, -4) {Density matrix of state $U\ket{0_n}$};
		\node [style=none] (131) at (-7, -0.5) {...};
		\node [style=none] (132) at (-8, -0.5) {...};
	\end{pgfonlayer}
	\begin{pgfonlayer}{edgelayer}
		\draw (0) to (13);
		\draw (1) to (14);
		\draw (2) to (15);
		\draw (13) to (16.center);
		\draw (14) to (17.center);
		\draw (21.center) to (10);
		\draw (15) to (18.center);
		\draw (20.center) to (11);
		\draw (19.center) to (12);
		\draw [style=dash edge] (4.center) to (5.center);
		\draw [style=dash edge] (5.center) to (7.center);
		\draw [style=dash edge] (7.center) to (6.center);
		\draw [style=dash edge] (6.center) to (4.center);
		\draw (10) to (22.center);
		\draw (11) to (23.center);
		\draw (12) to (24.center);
		\draw (25) to (38);
		\draw (26) to (39);
		\draw (27) to (40);
		\draw (38) to (41.center);
		\draw (39) to (42.center);
		\draw (46.center) to (35);
		\draw (40) to (43.center);
		\draw (45.center) to (36);
		\draw (44.center) to (37);
		\draw [style=dash edge] (29.center) to (30.center);
		\draw [style=dash edge] (30.center) to (32.center);
		\draw [style=dash edge] (32.center) to (31.center);
		\draw [style=dash edge] (31.center) to (29.center);
		\draw (35) to (47.center);
		\draw (36) to (48.center);
		\draw (37) to (49.center);
	\end{pgfonlayer}
\end{tikzpicture}}\]

\end{example}

 A mixed state can be generated by partially tracing away part of a pure system. The reduced density matrix of which some qubits are traced away can be represented by directly connecting the wire of the traced-away qubits between the $\ket{\psi}$ diagram and the $\bra{\psi}$ diagram. This is shown in example~\ref{fig:reduced_density_matrix}.

\begin{example}    \label{fig:reduced_density_matrix}
The ZX-diagram representation of a reduced density matrix by tracing away one of the qubits. The reduced density matrix can be represented by directly connecting the open wires of the qubit. In this example, the two open wires of the last qubit are connected, marked with a rectangle.

\centering

\[\scalebox{1}{\begin{tikzpicture}
	\begin{pgfonlayer}{nodelayer}
		\node [style=X dot] (50) at (2.25, 1.5) {};
		\node [style=X dot] (51) at (2.25, 0.5) {};
		\node [style=X dot] (52) at (2.25, -1.5) {};
		\node [style=none] (53) at (2.25, -0.5) {...};
		\node [style=none] (54) at (2.75, 2.5) {};
		\node [style=none] (55) at (2.75, -2.5) {};
		\node [style=none] (56) at (5.5, 2.5) {};
		\node [style=none] (57) at (5.5, -2.5) {};
		\node [style=none] (58) at (4, -3) {$U$};
		\node [style=none] (59) at (4, -0.5) {...};
		\node [style=Z dot] (60) at (5, 1.5) {};
		\node [style=Z dot] (61) at (5, 0.5) {};
		\node [style=Z dot] (62) at (5, -1.5) {};
		\node [style=X dot] (63) at (3.25, 1.5) {};
		\node [style=X dot] (64) at (3.25, 0.5) {};
		\node [style=X dot] (65) at (3.25, -1.5) {};
		\node [style=none] (66) at (3.75, 1.75) {};
		\node [style=none] (67) at (3.75, 0) {};
		\node [style=none] (68) at (3.75, -0.75) {};
		\node [style=none] (69) at (4.5, -1.75) {};
		\node [style=none] (70) at (4.5, 0) {};
		\node [style=none] (71) at (4.5, 1.75) {};
		\node [style=none] (72) at (6, 1.5) {};
		\node [style=none] (73) at (6, 0.5) {};
		\node [style=X dot] (75) at (10.75, 1.5) {};
		\node [style=X dot] (76) at (10.75, 0.5) {};
		\node [style=X dot] (77) at (10.75, -1.5) {};
		\node [style=none] (78) at (10.75, -0.5) {...};
		\node [style=none] (79) at (10.25, 2.5) {};
		\node [style=none] (80) at (10.25, -2.5) {};
		\node [style=none] (81) at (7.5, 2.5) {};
		\node [style=none] (82) at (7.5, -2.5) {};
		\node [style=none] (83) at (9, -3) {$U^{\dag}$};
		\node [style=none] (84) at (9, -0.5) {...};
		\node [style=Z dot] (85) at (8, 1.5) {};
		\node [style=Z dot] (86) at (8, 0.5) {};
		\node [style=Z dot] (87) at (8, -1.5) {};
		\node [style=X dot] (88) at (9.75, 1.5) {};
		\node [style=X dot] (89) at (9.75, 0.5) {};
		\node [style=X dot] (90) at (9.75, -1.5) {};
		\node [style=none] (91) at (9.25, 1.75) {};
		\node [style=none] (92) at (9.25, 0) {};
		\node [style=none] (93) at (9.25, -0.75) {};
		\node [style=none] (94) at (8.5, -1.75) {};
		\node [style=none] (95) at (8.5, 0) {};
		\node [style=none] (96) at (8.5, 1.75) {};
		\node [style=none] (97) at (7, 1.5) {};
		\node [style=none] (98) at (7, 0.5) {};
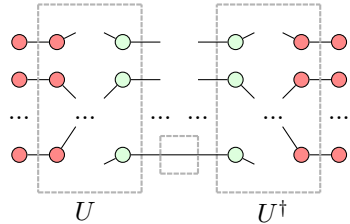
		\node [style=none] (100) at (6.5, -4) {Reduced density matrix by
tracing away a qubit};
		\node [style=none] (129) at (7, -0.5) {...};
		\node [style=none] (130) at (6, -0.5) {...};
		\node [style=none] (131) at (6, -1) {};
		\node [style=none] (132) at (6, -2) {};
		\node [style=none] (133) at (7, -2) {};
		\node [style=none] (134) at (7, -1) {};
	\end{pgfonlayer}
	\begin{pgfonlayer}{edgelayer}
		\draw (50) to (63);
		\draw (51) to (64);
		\draw (52) to (65);
		\draw (63) to (66.center);
		\draw (64) to (67.center);
		\draw (71.center) to (60);
		\draw (65) to (68.center);
		\draw (70.center) to (61);
		\draw (69.center) to (62);
		\draw [style=dash edge] (54.center) to (55.center);
		\draw [style=dash edge] (55.center) to (57.center);
		\draw [style=dash edge] (57.center) to (56.center);
		\draw [style=dash edge] (56.center) to (54.center);
		\draw (60) to (72.center);
		\draw (61) to (73.center);
		\draw (75) to (88);
		\draw (76) to (89);
		\draw (77) to (90);
		\draw (88) to (91.center);
		\draw (89) to (92.center);
		\draw (96.center) to (85);
		\draw (90) to (93.center);
		\draw (95.center) to (86);
		\draw (94.center) to (87);
		\draw [style=dash edge] (79.center) to (80.center);
		\draw [style=dash edge] (80.center) to (82.center);
		\draw [style=dash edge] (82.center) to (81.center);
		\draw [style=dash edge] (81.center) to (79.center);
		\draw (85) to (97.center);
		\draw (86) to (98.center);
		\draw (62) to (87);
		\draw [style=dash edge] (131.center) to (132.center);
		\draw [style=dash edge] (132.center) to (133.center);
		\draw [style=dash edge] (133.center) to (134.center);
		\draw [style=dash edge] (134.center) to (131.center);
	\end{pgfonlayer}
\end{tikzpicture}}\]

\end{example}
Although converting an arbitrary quantum circuit into a ZX-diagram is easy, it is not always trivial to convert a ZX-diagram back into a quantum circuit. A ZX-diagram can always represent an arbitrary gate in quantum circuits; however, a ZX-diagram can also represent a non-unitary tensor. For example, the number of input and output wires can be different. The process of converting a ZX-diagram into a quantum circuit is often referred to as \textit{circuit extraction}~\cite{Backens2020ThereTale}.

\paragraph{Garph-like ZX-diagram}

There is a special form of ZX-diagram that is particularly useful, called \textit{graph-like ZX-diagram}. ~\cite{Duncan2019Graph-theoreticZX-calculus}. 

\begin{definition}\label{definition:graph-like-diagram}
A diagram is called graph-like if 
\begin{enumerate}
    \item All spiders are Z-spiders.
    \item Spiders are only connected via Hadamard wires.
    \item There are no parallel Hadamard wires or self-loops.
    \item Every input or output is connected to a Z-spider.
    \item Every Z-spider is connected to at most one input or output.
\end{enumerate}
\label{definition:graph-like}
\end{definition}

An example is shown in example~\ref{fig:graph_like_diagram}.
 
From the Gottesman–Knill theorem, Quantum circuits containing gates only from the Clifford group can be simulated efficiently on a classical computer \cite{Aaronson2004ImprovedCircuits}. After a quantum circuit is written into a ZX-diagram, it is possible to simplify the diagram and remove Clifford operations before further modifications. This technique has been developed for circuit simplification~\cite{Duncan2019Graph-theoreticZX-calculus}.

\begin{example}\label{fig:graph_like_diagram}
A graph-like ZX-diagram. A graph-like ZX-diagram must have only Z-spiders (green), and the internal connections are only Hadamard wires (dashed-blue line). There are no self-loops or parallel wires, and each spider is connected to at most one input or output. (a) demonstrates \SXE{a schematic of} an original circuit written into ZX-diagram. \SXE{It can be created by substituting each quantum gate in the circuit with its corresponding ZX-diagram component. This diagram contains both X and Z spiders, with each row representing a qubit and each column denoting a layer of the circuit.} (b) \SXE{demonstrate schematics of a graph-like ZX-Diagram, comprised solely of Z spiders and Hadamard edges. A graph-like ZX-Diagram that is equivalent to a diagram like (a) can always be found \cite{Duncan2019Graph-theoreticZX-calculus}. In a graph-like ZX-diagram, the columns and rows no longer correspond directly to a gate layer or a qubit.} The open Hadamard edges with the ellipsis denote some arbitrary configuration that is abbreviated.
    
\vskip 1em    

\centering
\begin{tikzpicture}
	\begin{pgfonlayer}{nodelayer}
		\node [style=Z dot] (0) at (-22.75, -7.5) {};
		\node [style=Z dot] (1) at (-22.75, -8.5) {};
		\node [style=Z dot] (3) at (-22.75, -10.5) {};
		\node [style=none] (4) at (-22.75, -9.25) {...};
		\node [style=Z dot] (11) at (-20, -8.25) {};
		\node [style=Z dot] (12) at (-20, -10) {};
		\node [style=Z dot] (23) at (-14, -7.5) {};
		\node [style=Z dot] (24) at (-14, -8.5) {};
		\node [style=Z dot] (25) at (-14, -10.5) {};
		\node [style=none] (26) at (-14, -9.25) {...};
		\node [style=none] (27) at (-20, -9.25) {};
		\node [style=none] (28) at (-20, -9.25) {...};
		\node [style=none] (52) at (-13, -7.5) {};
		\node [style=none] (53) at (-13, -8.5) {};
		\node [style=none] (54) at (-13, -10.5) {};
		\node [style=none] (55) at (-23.25, -6.5) {};
		\node [style=none] (56) at (-22.25, -6.5) {};
		\node [style=none] (57) at (-23.25, -11.5) {};
		\node [style=none] (58) at (-22.25, -11.5) {};
		\node [style=none] (59) at (-14.5, -11.5) {};
		\node [style=none] (60) at (-13.5, -11.5) {};
		\node [style=none] (61) at (-14.5, -6.5) {};
		\node [style=none] (62) at (-13.5, -6.5) {};
		\node [style=none] (213) at (-20.5, -6.5) {};
		\node [style=none] (214) at (-19.5, -6.5) {};
		\node [style=none] (215) at (-19.5, -11.5) {};
		\node [style=none] (216) at (-20.5, -11.5) {};
		\node [style=none] (217) at (-21.25, -9.25) {...};
		\node [style=none] (218) at (-15.5, -9.25) {...};
		\node [style=X dot] (219) at (-23.75, 0) {};
		\node [style=X dot] (220) at (-23.75, -1) {};
		\node [style=X dot] (221) at (-23.75, -3) {};
		\node [style=none] (222) at (-23.75, -2) {...};
		\node [style=none] (223) at (-23.25, 1) {};
		\node [style=none] (224) at (-23.25, -4) {};
		\node [style=none] (227) at (-18.5, -5) {(a)Original circuit in ZX-diagram};
		\node [style=none] (228) at (-22, -2) {...};
		\node [style=none] (231) at (-13.5, 1) {};
		\node [style=none] (232) at (-13.5, -4) {};
		\node [style=none] (234) at (-15, -2) {...};
		\node [style=none] (235) at (-13, 0) {};
		\node [style=none] (236) at (-13, -1) {};
		\node [style=none] (237) at (-13, -3) {};
		\node [style=none] (238) at (-13, -3) {};
		\node [style=none] (239) at (-18.5, -2) {...};
		\node [style=none] (240) at (-13, -2) {...};
		\node [style=Z dot] (241) at (-21, 0) {};
		\node [style=Z dot] (242) at (-21, -1) {};
		\node [style=Z dot] (243) at (-21, -3) {};
		\node [style=Z dot] (244) at (-16, -3) {};
		\node [style=Z dot] (245) at (-16, -1) {};
		\node [style=Z dot] (246) at (-14, 0) {};
		\node [style=Z dot] (247) at (-14, -1) {};
		\node [style=Z dot] (248) at (-14, -3) {};
		\node [style=none] (249) at (-16, -1.5) {};
		\node [style=none] (250) at (-16.5, 0) {};
		\node [style=none] (251) at (-16, -2.5) {};
		\node [style=none] (252) at (-14, -2.5) {};
		\node [style=none] (253) at (-14, -1.5) {};
		\node [style=X dot] (255) at (-22.75, 0) {};
		\node [style=X dot] (256) at (-22.75, -1) {};
		\node [style=X dot] (257) at (-22.75, -3) {};
		\node [style=none] (259) at (-22.75, -1.5) {};
		\node [style=none] (260) at (-22.75, -2.5) {};
		\node [style=none] (262) at (-21, -1.5) {};
		\node [style=X dot] (265) at (-16, 0) {};
		\node [style=none] (269) at (-20.5, -3) {};
		\node [style=none] (270) at (-20.5, -1) {};
		\node [style=none] (271) at (-16.5, -1) {};
		\node [style=none] (274) at (-20.75, -9) {};
		\node [style=none] (275) at (-21, -9.75) {};
		\node [style=none] (276) at (-18.75, -9.75) {};
		\node [style=none] (277) at (-19, -9) {};
		\node [style=none] (278) at (-15, -9) {};
		\node [style=none] (279) at (-15, -10) {};
		\node [style=none] (281) at (-21.5, -8.75) {};
		\node [style=none] (282) at (-21.75, -9.25) {};
		\node [style=none] (283) at (-21.75, -9.75) {};
		\node [style=none] (285) at (-21, -2.5) {};
		\node [style=none] (286) at (-20.5, 0) {};
		\node [style=none] (287) at (-16.5, -3) {};
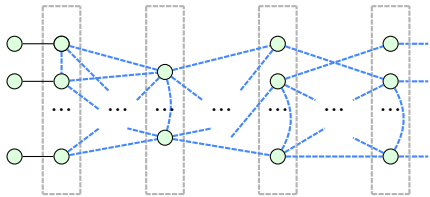
		\node [style=none] (288) at (-18.5, -12.5) {(b)Circuit represented in graph-like ZX-diagram};
		\node [style=Z dot] (289) at (-17, -7.5) {};
		\node [style=Z dot] (290) at (-17, -8.5) {};
		\node [style=Z dot] (291) at (-17, -10.5) {};
		\node [style=none] (292) at (-17, -9.25) {...};
		\node [style=none] (293) at (-16, -9) {};
		\node [style=none] (295) at (-15.75, -9.75) {};
		\node [style=none] (296) at (-17.5, -11.5) {};
		\node [style=none] (297) at (-16.5, -11.5) {};
		\node [style=none] (298) at (-17.5, -6.5) {};
		\node [style=none] (299) at (-16.5, -6.5) {};
		\node [style=none] (300) at (-18.5, -9.25) {...};
		\node [style=none] (301) at (-18.25, -10) {};
		\node [style=none] (303) at (-18, -8.75) {};
		\node [style=Z dot] (355) at (-24, -7.5) {};
		\node [style=Z dot] (356) at (-24, -8.5) {};
		\node [style=Z dot] (357) at (-24, -10.5) {};
		\node [style=Z dot] (358) at (-22.75, -7.5) {};
	\end{pgfonlayer}
	\begin{pgfonlayer}{edgelayer}
		\draw [style=hadamard edge] (23) to (52.center);
		\draw [style=hadamard edge] (24) to (53.center);
		\draw [style=hadamard edge] (25) to (54.center);
		\draw [style=dash edge] (55.center) to (57.center);
		\draw [style=dash edge] (57.center) to (58.center);
		\draw [style=dash edge] (58.center) to (56.center);
		\draw [style=dash edge] (56.center) to (55.center);
		\draw [style=dash edge] (61.center) to (59.center);
		\draw [style=dash edge] (59.center) to (60.center);
		\draw [style=dash edge] (60.center) to (62.center);
		\draw [style=dash edge] (62.center) to (61.center);
		\draw [style=dash edge] (213.center) to (216.center);
		\draw [style=dash edge] (215.center) to (216.center);
		\draw [style=dash edge] (215.center) to (214.center);
		\draw [style=dash edge] (213.center) to (214.center);
		\draw (253.center) to (247);
		\draw (252.center) to (248);
		\draw (245) to (249.center);
		\draw (244) to (251.center);
		\draw (246) to (235.center);
		\draw (247) to (236.center);
		\draw (248) to (238.center);
		\draw (219) to (255);
		\draw (220) to (256);
		\draw (221) to (257);
		\draw (256) to (259.center);
		\draw (257) to (260.center);
		\draw (262.center) to (242);
		\draw [style=dash edge] (223.center) to (224.center);
		\draw [style=dash edge] (232.center) to (231.center);
		\draw (265) to (250.center);
		\draw (271.center) to (245);
		\draw (242) to (270.center);
		\draw (243) to (269.center);
		\draw [style=dash edge] (224.center) to (232.center);
		\draw [style=dash edge] (223.center) to (231.center);
		\draw [style=hadamard edge] (3) to (12);
		\draw [style=hadamard edge] (3) to (283.center);
		\draw [style=hadamard edge] (12) to (275.center);
		\draw [style=hadamard edge] (1) to (282.center);
		\draw [style=hadamard edge] (0) to (281.center);
		\draw [style=hadamard edge] (11) to (274.center);
		\draw [style=hadamard edge] (1) to (11);
		\draw [style=hadamard edge] (0) to (11);
		\draw [style=hadamard edge, bend left=15] (11) to (12);
		\draw [style=hadamard edge] (11) to (277.center);
		\draw [style=hadamard edge] (12) to (276.center);
		\draw [style=hadamard edge] (24) to (278.center);
		\draw [style=hadamard edge] (25) to (279.center);
		\draw (255) to (241);
		\draw (256) to (242);
		\draw (245) to (247);
		\draw (265) to (246);
		\draw (257) to (243);
		\draw (244) to (248);
		\draw (285.center) to (243);
		\draw [style=hadamard edge] (241) to (242);
		\draw (241) to (286.center);
		\draw (287.center) to (244);
		\draw (265) to (245);
		\draw [style=hadamard edge] (0) to (1);
		\draw [style=hadamard edge, bend left] (24) to (25);
		\draw [style=hadamard edge] (291) to (295.center);
		\draw [style=dash edge] (298.center) to (296.center);
		\draw [style=dash edge] (296.center) to (297.center);
		\draw [style=dash edge] (297.center) to (299.center);
		\draw [style=dash edge] (299.center) to (298.center);
		\draw [style=hadamard edge] (290) to (301.center);
		\draw [style=hadamard edge] (289) to (303.center);
		\draw [style=hadamard edge, bend left] (290) to (291);
		\draw [style=hadamard edge] (12) to (291);
		\draw [style=hadamard edge] (11) to (289);
		\draw [style=hadamard edge] (291) to (25);
		\draw [style=hadamard edge] (290) to (293.center);
		\draw [style=hadamard edge] (289) to (24);
		\draw [style=hadamard edge] (290) to (23);
		\draw (355) to (358);
		\draw (1) to (356);
		\draw (357) to (3);
	\end{pgfonlayer}
\end{tikzpicture}

\end{example}

\subsection{Flow and determinism in MBQC}

Suppose a circuit is converted into a ZX-diagram and then transformed into a graph-like ZX-diagram.  Then we modify the ZX-diagram and give every spider an extra regular wire. See example \ref{fig:pi_propergation_measurement_sequence}(a). The new diagram equals the original diagram when all these wires terminate with a zero-phase Z spider. 

Now let us consider each spider corresponds to a qubit, and measure the qubit closes the added open wire with either a zero-phase spider or a $\pi$-phase spider. When we get an unwanted $\pi$-phase spider, see example \ref{fig:pi_propergation_measurement_sequence}(b), we restore the state by applying extra single-qubit operations on the related qubit which we have not been measured. This is equivalent to the \textit{correction} operation for MBQC. See example \ref{fig:pi_propergation_measurement_sequence}(c). In this way, we can always obtain the same distribution as the original quantum circuit.

\begin{example}\label{fig:pi_propergation_measurement_sequence}
      Measurement sequence of implementing correction. (a)The spiders have been grouped based on their distance from the output spiders. Shown in grey squares. (b)The measurement has been performed on the group with the largest distance; some unexpected outcome has been measured. (c) Apply the $\pirule$ and $\hadamardrule$ to recover the state. The phase can always get propagated into groups with a lower distance to the output spiders.

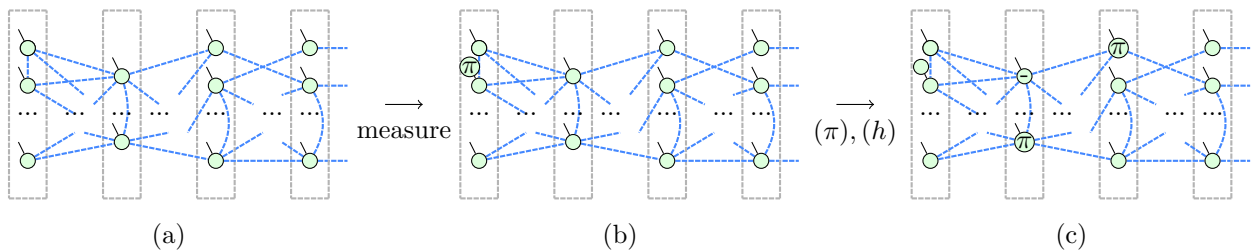
\centering
\[
\begin{tikzpicture}
	\begin{pgfonlayer}{nodelayer}
		\node [style=Z dot] (0) at (-15.75, 1.5) {};
		\node [style=Z dot] (1) at (-15.75, 0.5) {};
		\node [style=Z dot] (2) at (-15.75, -1.5) {};
		\node [style=none] (3) at (-15.75, -0.25) {...};
		\node [style=Z dot] (4) at (-13.25, 0.75) {};
		\node [style=Z dot] (5) at (-13.25, -1) {};
		\node [style=Z dot] (6) at (-8.25, 1.5) {};
		\node [style=Z dot] (7) at (-8.25, 0.5) {};
		\node [style=Z dot] (8) at (-8.25, -1.5) {};
		\node [style=none] (9) at (-8.25, -0.25) {...};
		\node [style=none] (10) at (-13.25, -0.25) {};
		\node [style=none] (11) at (-13.25, -0.25) {...};
		\node [style=none] (12) at (-7.25, 1.5) {};
		\node [style=none] (13) at (-7.25, 0.5) {};
		\node [style=none] (14) at (-7.25, -1.5) {};
		\node [style=none] (15) at (-16.25, 2.5) {};
		\node [style=none] (16) at (-15.25, 2.5) {};
		\node [style=none] (17) at (-16.25, -2.5) {};
		\node [style=none] (18) at (-15.25, -2.5) {};
		\node [style=none] (19) at (-8.75, -2.5) {};
		\node [style=none] (20) at (-7.75, -2.5) {};
		\node [style=none] (21) at (-8.75, 2.5) {};
		\node [style=none] (22) at (-7.75, 2.5) {};
		\node [style=none] (23) at (-13.75, 2.5) {};
		\node [style=none] (24) at (-12.75, 2.5) {};
		\node [style=none] (25) at (-12.75, -2.5) {};
		\node [style=none] (26) at (-13.75, -2.5) {};
		\node [style=none] (27) at (-14.5, -0.25) {...};
		\node [style=none] (28) at (-9.25, -0.25) {...};
		\node [style=none] (29) at (-14, 0) {};
		\node [style=none] (30) at (-14.25, -0.75) {};
		\node [style=none] (31) at (-12, -0.75) {};
		\node [style=none] (32) at (-12.25, 0) {};
		\node [style=none] (33) at (-9, 0.25) {};
		\node [style=none] (34) at (-9, -1) {};
		\node [style=none] (35) at (-14.25, 0.25) {};
		\node [style=none] (36) at (-14.5, -0.25) {};
		\node [style=none] (37) at (-14.5, -0.75) {};
		\node [style=Z dot] (38) at (-10.75, 1.5) {};
		\node [style=Z dot] (39) at (-10.75, 0.5) {};
		\node [style=Z dot] (40) at (-10.75, -1.5) {};
		\node [style=none] (41) at (-10.75, -0.25) {...};
		\node [style=none] (42) at (-9.75, 0) {};
		\node [style=none] (43) at (-9.5, -0.75) {};
		\node [style=none] (44) at (-11.25, -2.5) {};
		\node [style=none] (45) at (-10.25, -2.5) {};
		\node [style=none] (46) at (-11.25, 2.5) {};
		\node [style=none] (47) at (-10.25, 2.5) {};
		\node [style=none] (48) at (-12.25, -0.25) {...};
		\node [style=none] (49) at (-11.5, -0.75) {};
		\node [style=none] (50) at (-11.75, 0.25) {};
		\node [style=Z dot] (54) at (-15.75, 1.5) {};
		\node [style=none] (110) at (5.75, 0) {};
		\node [style=none] (111) at (6.75, 0) {};
		\node [style=none] (112) at (6.25, -0.75) {$(\pi),(h)$};
		\node [style=none] (113) at (-16, 2) {};
		\node [style=none] (114) at (-16, 1) {};
		\node [style=none] (115) at (-16, -1) {};
		\node [style=none] (116) at (-13.5, 1.25) {};
		\node [style=none] (117) at (-13.5, -0.5) {};
		\node [style=none] (118) at (-11, 2) {};
		\node [style=none] (119) at (-11, 1) {};
		\node [style=none] (120) at (-11, -1) {};
		\node [style=none] (121) at (-8.5, -1) {};
		\node [style=none] (122) at (-8.5, 1) {};
		\node [style=none] (123) at (-8.5, 2) {};
		\node [style=none] (190) at (-6.25, 0) {};
		\node [style=none] (191) at (-5.25, 0) {};
		\node [style=none] (192) at (-5.75, -0.75) {measure};
		\node [style=none] (262) at (-12, -3.5) {(a)};
		\node [style=none] (263) at (0, -3.5) {(b)};
		\node [style=none] (264) at (12, -3.5) {(c)};
		\node [style=Z dot] (265) at (-3.75, 1.5) {};
		\node [style=Z dot] (266) at (-3.75, 0.5) {};
		\node [style=Z dot] (267) at (-3.75, -1.5) {};
		\node [style=none] (268) at (-3.75, -0.25) {...};
		\node [style=Z dot] (269) at (-1.25, 0.75) {};
		\node [style=Z dot] (270) at (-1.25, -1) {};
		\node [style=Z dot] (271) at (3.75, 1.5) {};
		\node [style=Z dot] (272) at (3.75, 0.5) {};
		\node [style=Z dot] (273) at (3.75, -1.5) {};
		\node [style=none] (274) at (3.75, -0.25) {...};
		\node [style=none] (275) at (-1.25, -0.25) {};
		\node [style=none] (276) at (-1.25, -0.25) {...};
		\node [style=none] (277) at (4.75, 1.5) {};
		\node [style=none] (278) at (4.75, 0.5) {};
		\node [style=none] (279) at (4.75, -1.5) {};
		\node [style=none] (280) at (-4.25, 2.5) {};
		\node [style=none] (281) at (-3.25, 2.5) {};
		\node [style=none] (282) at (-4.25, -2.5) {};
		\node [style=none] (283) at (-3.25, -2.5) {};
		\node [style=none] (284) at (3.25, -2.5) {};
		\node [style=none] (285) at (4.25, -2.5) {};
		\node [style=none] (286) at (3.25, 2.5) {};
		\node [style=none] (287) at (4.25, 2.5) {};
		\node [style=none] (288) at (-1.75, 2.5) {};
		\node [style=none] (289) at (-0.75, 2.5) {};
		\node [style=none] (290) at (-0.75, -2.5) {};
		\node [style=none] (291) at (-1.75, -2.5) {};
		\node [style=none] (292) at (-2.5, -0.25) {...};
		\node [style=none] (293) at (2.75, -0.25) {...};
		\node [style=none] (294) at (-2, 0) {};
		\node [style=none] (295) at (-2.25, -0.75) {};
		\node [style=none] (296) at (0, -0.75) {};
		\node [style=none] (297) at (-0.25, 0) {};
		\node [style=none] (298) at (3, 0.25) {};
		\node [style=none] (299) at (3, -1) {};
		\node [style=none] (300) at (-2.25, 0.25) {};
		\node [style=none] (301) at (-2.5, -0.25) {};
		\node [style=none] (302) at (-2.5, -0.75) {};
		\node [style=Z dot] (303) at (1.25, 1.5) {};
		\node [style=Z dot] (304) at (1.25, 0.5) {};
		\node [style=Z dot] (305) at (1.25, -1.5) {};
		\node [style=none] (306) at (1.25, -0.25) {...};
		\node [style=none] (307) at (2.25, 0) {};
		\node [style=none] (308) at (2.5, -0.75) {};
		\node [style=none] (309) at (0.75, -2.5) {};
		\node [style=none] (310) at (1.75, -2.5) {};
		\node [style=none] (311) at (0.75, 2.5) {};
		\node [style=none] (312) at (1.75, 2.5) {};
		\node [style=none] (313) at (-0.25, -0.25) {...};
		\node [style=none] (314) at (0.5, -0.75) {};
		\node [style=none] (315) at (0.25, 0.25) {};
		\node [style=Z dot] (316) at (-3.75, 1.5) {};
		\node [style=none] (317) at (-4, 2) {};
		\node [style=none] (319) at (-4, -1) {};
		\node [style=none] (320) at (-1.5, 1.25) {};
		\node [style=none] (321) at (-1.5, -0.5) {};
		\node [style=none] (322) at (1, 2) {};
		\node [style=none] (323) at (1, 1) {};
		\node [style=none] (324) at (1, -1) {};
		\node [style=none] (325) at (3.5, -1) {};
		\node [style=none] (326) at (3.5, 1) {};
		\node [style=none] (327) at (3.5, 2) {};
		\node [style=Z dot] (328) at (-4, 1) {$\pi$};
		\node [style=Z dot] (332) at (8.25, 1.5) {};
		\node [style=Z dot] (333) at (8.25, 0.5) {};
		\node [style=Z dot] (334) at (8.25, -1.5) {};
		\node [style=none] (335) at (8.25, -0.25) {...};
		\node [style=Z dot] (336) at (10.75, 0.75) {-};
		\node [style=Z dot] (337) at (10.75, -1) {$\pi$};
		\node [style=Z dot] (338) at (15.75, 1.5) {};
		\node [style=Z dot] (339) at (15.75, 0.5) {};
		\node [style=Z dot] (340) at (15.75, -1.5) {};
		\node [style=none] (341) at (15.75, -0.25) {...};
		\node [style=none] (342) at (10.75, -0.25) {};
		\node [style=none] (343) at (10.75, -0.25) {...};
		\node [style=none] (344) at (16.75, 1.5) {};
		\node [style=none] (345) at (16.75, 0.5) {};
		\node [style=none] (346) at (16.75, -1.5) {};
		\node [style=none] (347) at (7.75, 2.5) {};
		\node [style=none] (348) at (8.75, 2.5) {};
		\node [style=none] (349) at (7.75, -2.5) {};
		\node [style=none] (350) at (8.75, -2.5) {};
		\node [style=none] (351) at (15.25, -2.5) {};
		\node [style=none] (352) at (16.25, -2.5) {};
		\node [style=none] (353) at (15.25, 2.5) {};
		\node [style=none] (354) at (16.25, 2.5) {};
		\node [style=none] (355) at (10.25, 2.5) {};
		\node [style=none] (356) at (11.25, 2.5) {};
		\node [style=none] (357) at (11.25, -2.5) {};
		\node [style=none] (358) at (10.25, -2.5) {};
		\node [style=none] (359) at (9.5, -0.25) {...};
		\node [style=none] (360) at (14.75, -0.25) {...};
		\node [style=none] (361) at (10, 0) {};
		\node [style=none] (362) at (9.75, -0.75) {};
		\node [style=none] (363) at (12, -0.75) {};
		\node [style=none] (364) at (11.75, 0) {};
		\node [style=none] (365) at (15, 0.25) {};
		\node [style=none] (366) at (15, -1) {};
		\node [style=none] (367) at (9.75, 0.25) {};
		\node [style=none] (368) at (9.5, -0.25) {};
		\node [style=none] (369) at (9.5, -0.75) {};
		\node [style=Z dot] (370) at (13.25, 1.5) {$\pi$};
		\node [style=Z dot] (371) at (13.25, 0.5) {};
		\node [style=Z dot] (372) at (13.25, -1.5) {};
		\node [style=none] (373) at (13.25, -0.25) {...};
		\node [style=none] (374) at (14.25, 0) {};
		\node [style=none] (375) at (14.5, -0.75) {};
		\node [style=none] (376) at (12.75, -2.5) {};
		\node [style=none] (377) at (13.75, -2.5) {};
		\node [style=none] (378) at (12.75, 2.5) {};
		\node [style=none] (379) at (13.75, 2.5) {};
		\node [style=none] (380) at (11.75, -0.25) {...};
		\node [style=none] (381) at (12.5, -0.75) {};
		\node [style=none] (382) at (12.25, 0.25) {};
		\node [style=Z dot] (383) at (8.25, 1.5) {};
		\node [style=none] (384) at (8, 2) {};
		\node [style=none] (385) at (8, -1) {};
		\node [style=none] (386) at (10.5, 1.25) {};
		\node [style=none] (387) at (10.5, -0.5) {};
		\node [style=none] (388) at (13, 2) {};
		\node [style=none] (389) at (13, 1) {};
		\node [style=none] (390) at (13, -1) {};
		\node [style=none] (391) at (15.5, -1) {};
		\node [style=none] (392) at (15.5, 1) {};
		\node [style=none] (393) at (15.5, 2) {};
		\node [style=Z dot] (394) at (8, 1) {};
	\end{pgfonlayer}
	\begin{pgfonlayer}{edgelayer}
		\draw [style=hadamard edge] (6) to (12.center);
		\draw [style=hadamard edge] (7) to (13.center);
		\draw [style=hadamard edge] (8) to (14.center);
		\draw [style=dash edge] (15.center) to (17.center);
		\draw [style=dash edge] (17.center) to (18.center);
		\draw [style=dash edge] (18.center) to (16.center);
		\draw [style=dash edge] (16.center) to (15.center);
		\draw [style=dash edge] (21.center) to (19.center);
		\draw [style=dash edge] (19.center) to (20.center);
		\draw [style=dash edge] (20.center) to (22.center);
		\draw [style=dash edge] (22.center) to (21.center);
		\draw [style=dash edge] (23.center) to (26.center);
		\draw [style=dash edge] (25.center) to (26.center);
		\draw [style=dash edge] (25.center) to (24.center);
		\draw [style=dash edge] (23.center) to (24.center);
		\draw [style=hadamard edge] (2) to (5);
		\draw [style=hadamard edge] (2) to (37.center);
		\draw [style=hadamard edge] (5) to (30.center);
		\draw [style=hadamard edge] (1) to (36.center);
		\draw [style=hadamard edge] (0) to (35.center);
		\draw [style=hadamard edge] (4) to (29.center);
		\draw [style=hadamard edge] (1) to (4);
		\draw [style=hadamard edge] (0) to (4);
		\draw [style=hadamard edge, bend left=15] (4) to (5);
		\draw [style=hadamard edge] (4) to (32.center);
		\draw [style=hadamard edge] (5) to (31.center);
		\draw [style=hadamard edge] (7) to (33.center);
		\draw [style=hadamard edge] (8) to (34.center);
		\draw [style=hadamard edge] (0) to (1);
		\draw [style=hadamard edge, bend left] (7) to (8);
		\draw [style=hadamard edge] (40) to (43.center);
		\draw [style=dash edge] (46.center) to (44.center);
		\draw [style=dash edge] (44.center) to (45.center);
		\draw [style=dash edge] (45.center) to (47.center);
		\draw [style=dash edge] (47.center) to (46.center);
		\draw [style=hadamard edge] (39) to (49.center);
		\draw [style=hadamard edge] (38) to (50.center);
		\draw [style=hadamard edge, bend left] (39) to (40);
		\draw [style=hadamard edge] (5) to (40);
		\draw [style=hadamard edge] (4) to (38);
		\draw [style=hadamard edge] (40) to (8);
		\draw [style=hadamard edge] (39) to (42.center);
		\draw [style=hadamard edge] (38) to (7);
		\draw [style=hadamard edge] (39) to (6);
		\draw [style=diredge] (110.center) to (111.center);
		\draw (113.center) to (54);
		\draw (114.center) to (1);
		\draw (115.center) to (2);
		\draw (116.center) to (4);
		\draw (117.center) to (5);
		\draw (118.center) to (38);
		\draw (119.center) to (39);
		\draw (120.center) to (40);
		\draw (123.center) to (6);
		\draw (122.center) to (7);
		\draw (121.center) to (8);
		\draw [style=diredge] (190.center) to (191.center);
		\draw [style=hadamard edge] (271) to (277.center);
		\draw [style=hadamard edge] (272) to (278.center);
		\draw [style=hadamard edge] (273) to (279.center);
		\draw [style=dash edge] (280.center) to (282.center);
		\draw [style=dash edge] (282.center) to (283.center);
		\draw [style=dash edge] (283.center) to (281.center);
		\draw [style=dash edge] (281.center) to (280.center);
		\draw [style=dash edge] (286.center) to (284.center);
		\draw [style=dash edge] (284.center) to (285.center);
		\draw [style=dash edge] (285.center) to (287.center);
		\draw [style=dash edge] (287.center) to (286.center);
		\draw [style=dash edge] (288.center) to (291.center);
		\draw [style=dash edge] (290.center) to (291.center);
		\draw [style=dash edge] (290.center) to (289.center);
		\draw [style=dash edge] (288.center) to (289.center);
		\draw [style=hadamard edge] (267) to (270);
		\draw [style=hadamard edge] (267) to (302.center);
		\draw [style=hadamard edge] (270) to (295.center);
		\draw [style=hadamard edge] (266) to (301.center);
		\draw [style=hadamard edge] (265) to (300.center);
		\draw [style=hadamard edge] (269) to (294.center);
		\draw [style=hadamard edge] (266) to (269);
		\draw [style=hadamard edge] (265) to (269);
		\draw [style=hadamard edge, bend left=15] (269) to (270);
		\draw [style=hadamard edge] (269) to (297.center);
		\draw [style=hadamard edge] (270) to (296.center);
		\draw [style=hadamard edge] (272) to (298.center);
		\draw [style=hadamard edge] (273) to (299.center);
		\draw [style=hadamard edge] (265) to (266);
		\draw [style=hadamard edge, bend left] (272) to (273);
		\draw [style=hadamard edge] (305) to (308.center);
		\draw [style=dash edge] (311.center) to (309.center);
		\draw [style=dash edge] (309.center) to (310.center);
		\draw [style=dash edge] (310.center) to (312.center);
		\draw [style=dash edge] (312.center) to (311.center);
		\draw [style=hadamard edge] (304) to (314.center);
		\draw [style=hadamard edge] (303) to (315.center);
		\draw [style=hadamard edge, bend left] (304) to (305);
		\draw [style=hadamard edge] (270) to (305);
		\draw [style=hadamard edge] (269) to (303);
		\draw [style=hadamard edge] (305) to (273);
		\draw [style=hadamard edge] (304) to (307.center);
		\draw [style=hadamard edge] (303) to (272);
		\draw [style=hadamard edge] (304) to (271);
		\draw (317.center) to (316);
		\draw (319.center) to (267);
		\draw (320.center) to (269);
		\draw (321.center) to (270);
		\draw (322.center) to (303);
		\draw (323.center) to (304);
		\draw (324.center) to (305);
		\draw (327.center) to (271);
		\draw (326.center) to (272);
		\draw (325.center) to (273);
		\draw (328) to (266);
		\draw [style=hadamard edge] (338) to (344.center);
		\draw [style=hadamard edge] (339) to (345.center);
		\draw [style=hadamard edge] (340) to (346.center);
		\draw [style=dash edge] (347.center) to (349.center);
		\draw [style=dash edge] (349.center) to (350.center);
		\draw [style=dash edge] (350.center) to (348.center);
		\draw [style=dash edge] (348.center) to (347.center);
		\draw [style=dash edge] (353.center) to (351.center);
		\draw [style=dash edge] (351.center) to (352.center);
		\draw [style=dash edge] (352.center) to (354.center);
		\draw [style=dash edge] (354.center) to (353.center);
		\draw [style=dash edge] (355.center) to (358.center);
		\draw [style=dash edge] (357.center) to (358.center);
		\draw [style=dash edge] (357.center) to (356.center);
		\draw [style=dash edge] (355.center) to (356.center);
		\draw [style=hadamard edge] (334) to (337);
		\draw [style=hadamard edge] (334) to (369.center);
		\draw [style=hadamard edge] (337) to (362.center);
		\draw [style=hadamard edge] (333) to (368.center);
		\draw [style=hadamard edge] (332) to (367.center);
		\draw [style=hadamard edge] (336) to (361.center);
		\draw [style=hadamard edge] (333) to (336);
		\draw [style=hadamard edge] (332) to (336);
		\draw [style=hadamard edge, bend left=15] (336) to (337);
		\draw [style=hadamard edge] (336) to (364.center);
		\draw [style=hadamard edge] (337) to (363.center);
		\draw [style=hadamard edge] (339) to (365.center);
		\draw [style=hadamard edge] (340) to (366.center);
		\draw [style=hadamard edge] (332) to (333);
		\draw [style=hadamard edge, bend left] (339) to (340);
		\draw [style=hadamard edge] (372) to (375.center);
		\draw [style=dash edge] (378.center) to (376.center);
		\draw [style=dash edge] (376.center) to (377.center);
		\draw [style=dash edge] (377.center) to (379.center);
		\draw [style=dash edge] (379.center) to (378.center);
		\draw [style=hadamard edge] (371) to (381.center);
		\draw [style=hadamard edge] (370) to (382.center);
		\draw [style=hadamard edge, bend left] (371) to (372);
		\draw [style=hadamard edge] (337) to (372);
		\draw [style=hadamard edge] (336) to (370);
		\draw [style=hadamard edge] (372) to (340);
		\draw [style=hadamard edge] (371) to (374.center);
		\draw [style=hadamard edge] (370) to (339);
		\draw [style=hadamard edge] (371) to (338);
		\draw (384.center) to (383);
		\draw (385.center) to (334);
		\draw (386.center) to (336);
		\draw (387.center) to (337);
		\draw (388.center) to (370);
		\draw (389.center) to (371);
		\draw (390.center) to (372);
		\draw (393.center) to (338);
		\draw (392.center) to (339);
		\draw (391.center) to (340);
		\draw (394) to (333);
	\end{pgfonlayer}
\end{tikzpicture}

\]
\end{example}

Pushing the $\pi$ phase into unmeasured spiders is a simple correction strategy; however, it does not work for arbitrary graphs. For example, if a non-output qubit is being measured with an unexpected result, and all its neighbouring qubit has already been measured, then there is no qubit the $\pi$ phase can be pushed to. For a diagram that can utilize this single qubit correction strategy, the diagram must admit a \textit{causal flow}~\cite{Danos2005DeterminismModel}. 

\begin{definition}[Causal flow~\cite{Danos2005DeterminismModel,Browne2007GeneralizedComputation}]
A \textit{Causal flow} is a pair $(f,\prec)$ with $\prec$ a partial order and $f$ a function $f :O_c\rightarrow I_c$ on open graph state $(G,I,O)$ which associates with every non-output vertices a set of non-input vertices such that 

\begin{itemize}
    \item \SXE{$u\in N(f(u))$.}
    \item $u\prec f(u)$
    \item \SXE{$u\prec v$ for all $v\neq u$, $v\in N(f(u))$.}
\end{itemize}
\end{definition}
\SXE{where $N(K)$ denote the neighbor vertices of $K$.}

Consider we are looking for a strategy to execute the graph with the MBQC method, which consists of measuring the qubits and modifying one qubit $f(u)$ after each measurement. The partial order of the causal flow describes a possible order to execute the measurement. $f(u)$ qubits must be measured after $u$, which gives the second rule. Also, applying the correction would affect not only $u$, but all the neighbours of $f(u)$. Therefore these neighbors must be measured after $u$.

However, admitting a causal flow is unnecessary for a graph to be executed with MBQC~\cite{Danos2005DeterminismModel,Duncan2010RewritingFlow,Duncan2012AComputing,Browne2007GeneralizedComputation}. Recall that modifying one qubit can correct the unexpected measurement outcome from a graph with the causal flow. There are at least two improvements to the flow mechanism that can be applied.

First, instead of correcting the state by modifying one qubit, the idea of graph stabilizers can be used to obtain a set of qubits and corresponding operations to correct the state. A stabilizer of a graph is a set of operations that can be applied to the state while keeping it identical. To correct the state, we could consider the unexpected gate as part of a stabilizer, which would keep the state unchanged if we complete it. Such stabilizers can be found intuitively with ZX-calculus by pushing the unexpected $\pi$ phase around the phase-free graph~\cite{Backens2013TheMechanics,Backens2016CompletenessZX-calculus}. We then could relax $f(u)$ to be multiple qubits called the \textit{correction set}.

Second, for the qubits measured on the Pauli basis, some correction does not need to be applied to the qubit physically~\cite{DuncanRewritingFlow,Browne2007GeneralizedComputation}. For example, pushing a $\pi$ spider through another $\pi$ phase spider with a different colour does not need to apply physical corrections because $-\pi$ and $\pi$ phase are equivalent. 

The above two modifications lead to the definition of \textit{Pauli flow}.

\begin{definition}[Pauli flow~\cite{DuncanRewritingFlow}]\label{def:pauli_flow}

An open graph state $(G,I,O)$ has Pauli flow if there exists a map $f:O^c\rightarrow F(I^c)$ and a partial order $\prec$ over $V$ such that for all $u\in O^c$

\begin{enumerate}

    \item if $v \in f(u)$, and $\lambda(v) \notin {X, Y}$ then $u\prec v$,
    \item if $v \neq u$, and $\lambda(v) \notin {Y, Z}$ then $v \notin Odd(f(u))$,
    \item if $v \preceq u, v \in f(u)$ and $\lambda(v) = Y$ then $v \in Odd(f(u))$,
    \item if $\lambda(u) = XY$ then $u \notin f(u)$ and $u \in Odd(f(u))$,
    \item if $\lambda(u) = XZ$ then $u \in f(u)$ and $u \in Odd(f(u))$,
    \item if $\lambda(u) = YZ$ then $u \in f(u)$ and $u \notin Odd(f(u))$,
    \item if $\lambda(u) = X$ then $u \in Odd(f(u))$,
    \item if $\lambda(u) = Z$ then $u \in f(u)$,
    \item if $\lambda(u) = Y$ then either: $u \notin f(u)$ and $u \in Odd(f(u))$ or $u \in f(u)$ and $u \notin Odd(f(u))$.
\end{enumerate}
\end{definition}

Where $Odd(K) = \{u,|N(u)\cap K|=1~mod~2\}$ is the odd neighbour of $K$, i.e. the set of vertices which have an odd number of neighbours in $K$. $N(K)$ denote the neighbor vertices of $K$. $|K|$ denote the number of vertices in $K$. $\lambda(u)$ denote the measurement plane of $u$, for green spiders with $0$ or $\pi$ phase, the measurement plane is $X$. The measurement plane is $XY$ for other arbitrary phases.

In the following sections, we show that all the rewrite rules used to implement the protocol would at least preserve the Pauli flow of the graph. This guarantees the transformed ZX-diagrams can be executed on MBQC hardware.

\section{BQC from ZX-calculus}\label{sec:zx_bqc}

The previous section shows that a graph-like ZX-Diagram can fully describe the information needed to execute a quantum algorithm. This information includes each spider's phase, the connectivity configuration, and the number of spiders used in the graph. The outcome of each measurement may also contain information about the result of the algorithm. 

The UBQC protocol splits the initial phase into two parts to obfuscate this information. Each is independent of the initial phase; however, the execution would yield the same result when combined. The measurement results are obfuscated by randomly flipping the qubits before measurement and classically restoring them by the client after measurement. For obfuscation of the connectivity, the UBQC protocol utilises a universal cluster state; therefore, any algorithm would have an identical entanglement structure. Using the universal cluster state forces all the information of the algorithm to be stored in the phases. Limiting the ability to represent information with the layout of the cluster state requires extra resources. Our proposal obfuscates the connectivity by making those connectivities carry information about the algorithm entanglement structure and become connectivities between two different agents. Because each agent only possesses the fragment of the diagram that executes on itself, it loses track of the entanglement structure of the algorithm. Because our protocol does not require the universal cluster state, and encoding a considerable portion of the algorithm into the connectivity between the spiders, it requires fewer resources than the UBQC protocol.

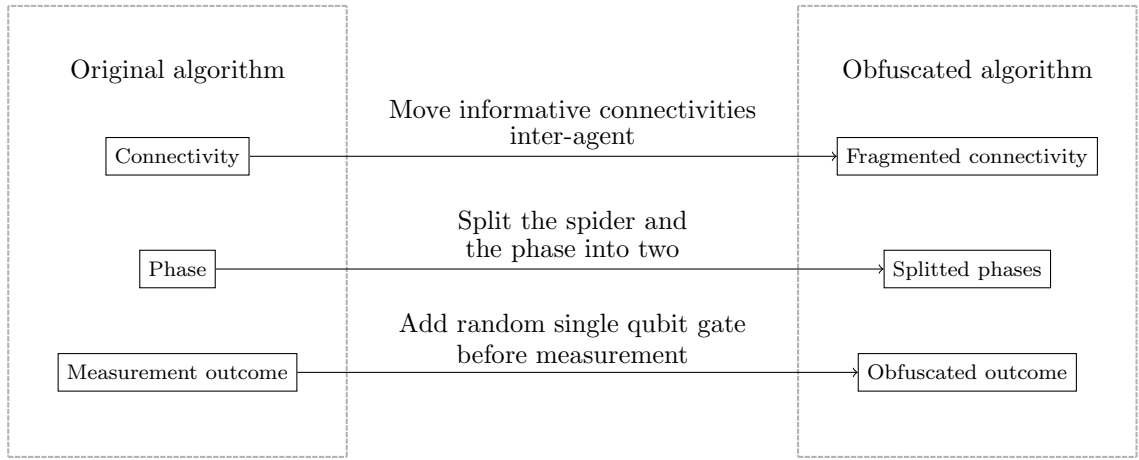
\[\centering
  \begin{tikzpicture}
	\begin{pgfonlayer}{nodelayer}
		\node [style=none] (0) at (-15, 5) {};
		\node [style=none] (1) at (-15, -7) {};
		\node [style=none] (2) at (-6, -7) {};
		\node [style=none] (3) at (-6, 5) {};
		\node [style=none] (4) at (-10.5, 3.25) {Original algorithm};
		\node [style=box] (5) at (-10.5, 1) {Connectivity};
		\node [style=box] (6) at (-10.5, -2) {Phase};
		\node [style=box] (7) at (-10.5, -4.75) {Measurement outcome};
		\node [style=box] (8) at (10.5, 1) {Fragmented connectivity};
		\node [style=box] (9) at (10.5, -2) {Splitted phases};
		\node [style=box] (10) at (10.5, -4.75) {Obfuscated outcome};
		\node [style=none] (11) at (6, 5) {};
		\node [style=none] (12) at (6, -7) {};
		\node [style=none] (13) at (15, -7) {};
		\node [style=none] (14) at (15, 5) {};
		\node [style=none] (15) at (10.5, 3.25) {Obfuscated algorithm};
		\node [style=none] (16) at (0, 2.25) {Move informative connectivities};
		\node [style=none] (17) at (0, -0.75) {Split the spider and };
		\node [style=none] (18) at (0, -4.25) { before measurement};
		\node [style=none] (19) at (0, 1.5) {inter-agent};
		\node [style=none] (20) at (0, -1.5) {the phase into two};
		\node [style=none] (21) at (0, -3.5) {Add random single qubit gate};
	\end{pgfonlayer}
	\begin{pgfonlayer}{edgelayer}
		\draw [style=dash edge] (0.center) to (1.center);
		\draw [style=dash edge] (1.center) to (2.center);
		\draw [style=dash edge] (2.center) to (3.center);
		\draw [style=dash edge] (3.center) to (0.center);
		\draw [style=dash edge] (11.center) to (12.center);
		\draw [style=dash edge] (12.center) to (13.center);
		\draw [style=dash edge] (13.center) to (14.center);
		\draw [style=dash edge] (14.center) to (11.center);
		\draw [style=diredge] (7) to (10);
		\draw [style=diredge] (6) to (9);
		\draw [style=diredge] (5) to (8);
	\end{pgfonlayer}
\end{tikzpicture}
\]

In this section, we show how to implement our protocol with ZX-Calculus. The phases and measurement results are obfuscated with similar approaches to the UBQC protocol. With proper manipulation of the ZX-diagram, the connectivity information can all be hidden by ensuring that each agent only possesses one end of the entanglement that holds the information about the entanglement structure.

\subsection{Defining blocks\label{sec:define_blocks}}

Here we introduce the concept of \textit{Blocks}.  Blocks $B_i$ are a set of spiders that are hosted by the same agent. $B({V}) = B_k$ denote the block spider $V$ belongs to block $B_k$.  For simplicity, we define $B_i-1$ = $B_{i-1}$.

\begin{definition}[Spider depth]
For a given quantum algorithm represented in ZX-Diagram $G(E_N,E_H,V)$, let $d_G(V_1,V_2)$ denote the distance of $V_1$ and $V_2$ in graph $G$, $V_o$ denotes all output spiders. Define \text{depth} of the spider $V_i$ in the graph $G$ as
\begin{equation}
    D(V_i) = min(d(V_i,V_j)), \forall V_j \in V_O
\end{equation}
\end{definition}

\begin{definition}[Blocks initialization]
 Define the block as a set of spiders, and spider $V_i$ belongs to block $B(V_i)$, given by 
 
\begin{equation}
    B(V_i) = B_k  
\end{equation} 

where 
 
\begin{equation}
    k = D(V_i) 
\end{equation}
 
\end{definition}

To illustrate this partition, consider a quantum circuit directly transformed into a ZX-diagram. This partition simply categorises each layer of the quantum circuit into an individual block.

\subsection{Phase obfuscation}

Before we move into the method, let's start with a few definitions.

\begin{definition}[Semi-graph-like diagram]\label{definition:semi-graph-like-diagram}
A diagram is called semi-graph-like if 
\begin{enumerate}
    \item All spiders are Z-spiders.
    \item There are no parallel Hadamard wires or self-loops.
    \item Every input or output is connected to a Z-spider.
    \item Every Z-spider is connected to at most one input or output.
\end{enumerate}
\end{definition}

The difference between a semi-graph-like diagram and a graph-like diagram is that it allows regular edges to be present in the graph.

\begin{definition}[Reduced graph-like diagram]\label{definition:semi-graph-like-diagram-reduce}
A graph-like diagram $G_g$ is the \textit{reduced graph-like diagram} of a semi-graph-like diagram $G_{sg}$ if $G_{sg}$ can be transformed into $G_g$ with only rule $\spiderrule$.
\end{definition}

Now we consider the obfuscation process of the phase of a spider. The goal of this obfuscation step is to rewrite the graph so that the individual phase value in the new graph is independent of the phase values in the original graph. Such rewrite can be implemented by applying the $\spiderrule$ to make a single spider become multiple spiders connected with regular edges, see example~\ref{fig:angle_complication}. Suppose the original spider has phase $\alpha$. The new phases for new spiders are $\alpha_i$. Rule $\spiderrule$ shows that the rewrite graph is equivalent to the original graph if phase $\alpha_i$ is chosen to satisfy

\begin{equation}
    \alpha = \sum_i \alpha_i
\end{equation}

When the operation to implement phase $\alpha$ is split into multiple operations across different agents, each single agent would not be able to find the original phase $\alpha$. 

\begin{example}\label{fig:angle_complication}
Phase obfuscation with rule $\spiderrule$. (a) is the original spider with multiple inputs and outputs. (b) is equivalent to (a), while the phase has been split into two spiders connected with a regular wire. $\alpha_1$ and $\alpha_2$ can be chosen randomly with the restriction $\alpha = \alpha_1 + \alpha_2$. Adding one extra spider gives minimum protection to hide the rotation phase from the agent. (c) depict a more general form where the spider can be split into $n$ spiders. 

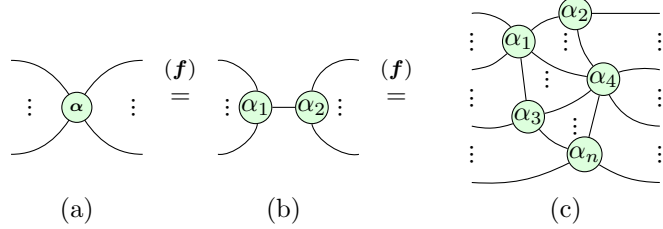
\centering
\begin{tikzpicture}
	\begin{pgfonlayer}{nodelayer}
		\node [style=none] (18) at (-5.75, 0) {$=$};
		\node [style=none] (20) at (-10.25, 1) {};
		\node [style=none, rotate=90] (21) at (-9.75, -0.25) {...};
		\node [style=none] (22) at (-6.75, -1.5) {};
		\node [style=none, rotate=90] (23) at (-7, -0.25) {...};
		\node [style=Z phase dot] (24) at (-8.5, -0.25) {$\alpha$};
		\node [style=none] (25) at (-6.75, 1) {};
		\node [style=none] (26) at (-10.25, -1.5) {};
		\node [style=none] (27) at (-5.75, 0.75) {\spiderrule};
		\node [style=Z dot] (28) at (-3.75, -0.25) {$\alpha_1$};
		\node [style=Z dot] (30) at (-2.25, -0.25) {$\alpha_2$};
		\node [style=none] (31) at (0, 0) {$=$};
		\node [style=none] (32) at (0, 0.75) {\spiderrule};
		\node [style=none] (33) at (-4.75, 1) {};
		\node [style=none] (34) at (-4.75, -1.5) {};
		\node [style=none] (35) at (-1, 1) {};
		\node [style=none] (36) at (-1, -1.5) {};
		\node [style=none] (37) at (2, 2.25) {};
		\node [style=none] (38) at (2, 0.75) {};
		\node [style=none] (39) at (2, -0.75) {};
		\node [style=none] (40) at (2, -2.25) {};
		\node [style=Z dot] (41) at (3.25, 1.5) {$\alpha_1$};
		\node [style=Z dot] (42) at (3.5, -0.5) {$\alpha_3$};
		\node [style=Z dot] (43) at (4.75, 2.25) {$\alpha_2$};
		\node [style=Z dot] (44) at (5, -1.5) {$\alpha_n$};
		\node [style=Z dot] (45) at (5.5, 0.5) {$\alpha_4$};
		\node [style=none] (46) at (7, 2.25) {};
		\node [style=none] (47) at (7, 0.75) {};
		\node [style=none] (48) at (7, -0.75) {};
		\node [style=none] (49) at (7, -2.25) {};
		\node [style=none, rotate=90] (50) at (2, 1.5) {...};
		\node [style=none, rotate=90] (51) at (2, 0) {...};
		\node [style=none, rotate=90] (54) at (2, -1.5) {...};
		\node [style=none, rotate=90] (55) at (7, 0) {...};
		\node [style=none, rotate=90] (56) at (7, 1.5) {...};
		\node [style=none, rotate=90] (57) at (7, -1.5) {...};
		\node [style=none, rotate=90] (58) at (4.75, -0.75) {...};
		\node [style=none, rotate=90] (59) at (4, 0.5) {...};
		\node [style=none, rotate=90] (60) at (4.5, 1.5) {...};
		\node [style=none] (61) at (-8.5, -3) {(a)};
		\node [style=none] (62) at (-3, -3) {(b)};
		\node [style=none] (63) at (4.5, -3) {(c)};
		\node [style=none, rotate=90] (64) at (-4.5, -0.25) {...};
		\node [style=none, rotate=90] (65) at (-1.5, -0.25) {...};
	\end{pgfonlayer}
	\begin{pgfonlayer}{edgelayer}
		\draw [in=-120, out=0] (26.center) to (24);
		\draw [in=180, out=-60] (24) to (22.center);
		\draw [in=180, out=60] (24) to (25.center);
		\draw [in=120, out=0] (20.center) to (24);
		\draw [bend left=45, looseness=0.75] (33.center) to (28);
		\draw [bend right=45, looseness=0.75] (34.center) to (28);
		\draw [bend left=45] (30) to (35.center);
		\draw [bend right=45, looseness=0.75] (30) to (36.center);
		\draw (28) to (30);
		\draw [bend left, looseness=0.75] (37.center) to (41);
		\draw [bend right=15, looseness=1.25] (38.center) to (41);
		\draw [bend right=15] (39.center) to (42);
		\draw [bend right=15] (40.center) to (44);
		\draw [bend right=15] (44) to (49.center);
		\draw [bend right] (45) to (48.center);
		\draw [bend left=15, looseness=0.75] (45) to (47.center);
		\draw (43) to (46.center);
		\draw [bend right=15] (41) to (45);
		\draw [bend right=15] (43) to (41);
		\draw (41) to (42);
		\draw [bend right=15, looseness=0.75] (42) to (45);
		\draw (45) to (44);
		\draw [bend right=15] (42) to (44);
		\draw [bend right=15] (43) to (45);
	\end{pgfonlayer}
\end{tikzpicture}

\end{example}

We have shown the obfuscated diagram is equivalent to the original diagram, and next, we show the obfuscated diagram can also be executed on the physical hardware in an MBQC manner. 

\begin{theorem}[Spider split flow preservation]\label{fig:execution_strategy}
Given a graph state $(G,I,O)$, where $G$ is a graph-like diagram $G(E_N,E_H,V)$. \textit{Split} splider $V_i\in V$ into N spiders, $\mathbf{\Tilde{V}_i} = \{\Tilde{V}_i^{(0)}...\Tilde{V}_i^{(N)}\}$, that is construct a new graph $\Tilde{G}{(\Tilde{E}_N,\Tilde{E}_H),\Tilde{V}_i}$, where $\Tilde{V}=V$ except $V_i$ is replaced with a set of spiders $V_i^{k}$. $V_i^{k}$ are connected through regular edges. If $(G,I,O)$ admit a Pauli flow $(f,\prec)$, the new graph $\Tilde{G}$ also admit a Pauli flow.

\end{theorem}

\begin{proof}
 By measuring all the splited spider $\Tilde{V}_i^{(n)}$, we could obtain all $\phi_{\Tilde{\phi}^i}$. Using $\spiderrule$ to merge the split spider back to one spider, we can obtain the effective measurement outcome for $V_i$ as $\phi_{V_i} = \sum_N \phi_{\Tilde{V}^i}$.
Therefore the split spiders $\mathbf{\Tilde{V}_i}$ has the same predecessors and successors as $V_i$ in partial order $\prec$. The partial order $\Tilde{\prec}$ for $\Tilde{G}$ can also be constructed as follows.

\begin{equation}
\begin{aligned}
  \Tilde{\prec} &= \bigcup\{(V_m,V_n)\}\cup \bigcup\{(V_i^\prime,V_k)\}\cup \bigcup\{(V_j,V_i^\prime)\},\\ &\forall (V_m,V_n) \in \prec ,V_m \neq V_i, V_n \neq V_i,\\ &\forall (V_j,V_i) \in \prec \mathrm{and} ~ \forall (V_i,V_k) \in \prec   
\end{aligned}
\end{equation}
\end{proof}

Since $V_i$ satisfies all the requirements from definition of Pauli flow \ref{def:pauli_flow}, each node $\{\Tilde{V}_i^{(0)}...\Tilde{V}_i^{(N)}\}$ also satisfies all the requirements. Therefore new graph $\Tilde{G}$ also admits a Pauli flow.

\begin{example}
 
Execution strategy on a ZX-diagram with regular edges. For a single spider in the original ZX-daigram (a), it is split into multiple spiders (b) by rule $\spiderrule$ . All split spiders were measured when executing a measurement step in the original diagram (d). Inversely apply rule $\spiderrule$ gives an equivalent effect of measuring the single spider in the original diagram (c). 

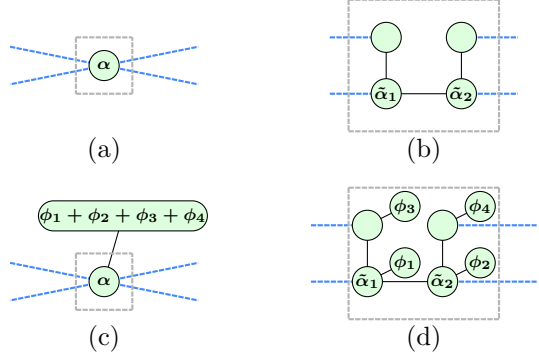
\centering
    \begin{tikzpicture}
	\begin{pgfonlayer}{nodelayer}
		\node [style=Z phase dot] (0) at (-6, 1.75) {$\alpha$};
		\node [style=none] (4) at (-8.5, 2.25) {};
		\node [style=none] (5) at (-3.5, 2.25) {};
		\node [style=Z phase dot] (6) at (1.5, 1) {$\Tilde{\alpha}_1$};
		\node [style=none] (7) at (0, 1) {};
		\node [style=Z phase dot] (8) at (3.5, 1) {$\Tilde{\alpha}_2$};
		\node [style=none] (9) at (5, 1) {};
		\node [style=none] (10) at (-8.5, 1.25) {};
		\node [style=none] (11) at (-3.5, 1.25) {};
		\node [style=Z phase dot] (12) at (1.5, 2.5) {};
		\node [style=Z phase dot] (13) at (3.5, 2.5) {};
		\node [style=none] (14) at (0, 2.5) {};
		\node [style=none] (15) at (5, 2.5) {};
		\node [style=Z phase dot] (16) at (-6, -4) {$\alpha$};
		\node [style=none] (17) at (-8.5, -3.5) {};
		\node [style=none] (18) at (-3.5, -3.5) {};
		\node [style=Z phase dot] (19) at (1, -4) {$\Tilde{\alpha}_1$};
		\node [style=none] (20) at (-0.5, -4) {};
		\node [style=Z phase dot] (21) at (3, -4) {$\Tilde{\alpha}_2$};
		\node [style=none] (22) at (5.5, -4) {};
		\node [style=none] (23) at (-8.5, -4.5) {};
		\node [style=none] (24) at (-3.5, -4.5) {};
		\node [style=Z phase dot] (25) at (1, -2.5) {};
		\node [style=Z phase dot] (26) at (3, -2.5) {};
		\node [style=none] (27) at (-0.5, -2.5) {};
		\node [style=none] (28) at (5.5, -2.5) {};
		\node [style=Z phase dot] (29) at (2, -3.5) {$\phi_1$};
		\node [style=Z phase dot] (30) at (4, -3.5) {$\phi_2$};
		\node [style=Z phase dot] (31) at (2, -2) {$\phi_3$};
		\node [style=Z phase dot] (32) at (4, -2) {$\phi_4$};
		\node [style=Z phase dot] (33) at (-5.5, -2.25) {$\phi_1+\phi_2+\phi_3+\phi_4$};
		\node [style=none] (34) at (-6, -0.5) {(a)};
		\node [style=none] (35) at (2.5, -0.5) {(b)};
		\node [style=none] (36) at (-6, -5.5) {(c)};
		\node [style=none] (37) at (2.5, -5.5) {(d)};
		\node [style=none] (38) at (-6.75, 2.5) {};
		\node [style=none] (39) at (-6.75, 1) {};
		\node [style=none] (40) at (-5.25, 1) {};
		\node [style=none] (41) at (-5.25, 2.5) {};
		\node [style=none] (42) at (-5.25, 2.5) {};
		\node [style=none] (43) at (0.5, 3.5) {};
		\node [style=none] (44) at (0.5, 0) {};
		\node [style=none] (45) at (4.5, 0) {};
		\node [style=none] (46) at (4.5, 3.5) {};
		\node [style=none] (47) at (0.5, -1.5) {};
		\node [style=none] (48) at (0.5, -5) {};
		\node [style=none] (49) at (4.5, -5) {};
		\node [style=none] (50) at (4.5, -1.5) {};
		\node [style=none] (51) at (-6.75, -3.25) {};
		\node [style=none] (52) at (-6.75, -4.75) {};
		\node [style=none] (53) at (-5.25, -4.75) {};
		\node [style=none] (54) at (-5.25, -3.25) {};
		\node [style=none] (55) at (-5.25, -3.25) {};
	\end{pgfonlayer}
	\begin{pgfonlayer}{edgelayer}
		\draw [style=hadamard edge] (4.center) to (0);
		\draw [style=hadamard edge] (0) to (5.center);
		\draw [style=hadamard edge] (7.center) to (6);
		\draw (6) to (8);
		\draw [style=hadamard edge] (8) to (9.center);
		\draw [style=hadamard edge] (0) to (10.center);
		\draw [style=hadamard edge] (0) to (11.center);
		\draw [style=hadamard edge] (14.center) to (12);
		\draw [style=hadamard edge] (13) to (15.center);
		\draw (6) to (12);
		\draw (13) to (8);
		\draw [style=hadamard edge] (17.center) to (16);
		\draw [style=hadamard edge] (16) to (18.center);
		\draw [style=hadamard edge] (20.center) to (19);
		\draw (19) to (21);
		\draw [style=hadamard edge] (21) to (22.center);
		\draw [style=hadamard edge] (16) to (23.center);
		\draw [style=hadamard edge] (16) to (24.center);
		\draw [style=hadamard edge] (27.center) to (25);
		\draw [style=hadamard edge] (26) to (28.center);
		\draw (19) to (25);
		\draw (26) to (21);
		\draw (19) to (29);
		\draw (21) to (30);
		\draw (25) to (31);
		\draw (26) to (32);
		\draw (16) to (33);
		\draw [style=dash edge] (38.center) to (42.center);
		\draw [style=dash edge] (38.center) to (39.center);
		\draw [style=dash edge] (39.center) to (40.center);
		\draw [style=dash edge] (40.center) to (42.center);
		\draw [style=dash edge] (43.center) to (44.center);
		\draw [style=dash edge] (44.center) to (45.center);
		\draw [style=dash edge] (46.center) to (45.center);
		\draw [style=dash edge] (43.center) to (46.center);
		\draw [style=dash edge] (47.center) to (48.center);
		\draw [style=dash edge] (48.center) to (49.center);
		\draw [style=dash edge] (50.center) to (49.center);
		\draw [style=dash edge] (47.center) to (50.center);
		\draw [style=dash edge] (51.center) to (55.center);
		\draw [style=dash edge] (51.center) to (52.center);
		\draw [style=dash edge] (52.center) to (53.center);
		\draw [style=dash edge] (53.center) to (55.center);
	\end{pgfonlayer}
\end{tikzpicture}
\end{example}

From the theorem \ref{fig:execution_strategy}, we define the Pauli flow for a semi-graph-like diagram.

\begin{definition}[Pauli flow on semi-graph-like diagram]\label{def:sgl-diagram}
A semi-graph-like diagram admits a Pauli flow if its reduced graph-like diagram admits a Pauli flow. 
\end{definition}

\begin{lemma}[Spider rule flow preservation]\label{lemma:spider_rule_flow}
Rewrite rule $\spiderrule$ preserves Pauli flow on semi-graph-like diagrams.
\end{lemma}

\begin{proof}
Because applying $\spiderrule$ to a semi-graph-like diagram will not change its reduced graph-like diagram. From definition \ref{def:sgl-diagram}, $\spiderrule$ will not affect the flow property of the diagram. 
\end{proof}

\begin{theorem}[Phase obfuscation]
\label{fig:phase_splitting}
Given a quantum algorithm represented in graph-like ZX diagram $G=(\emptyset,E_H,V)$ with $n$ spiders, where $\emptyset$ denote the null set and $E_H$ denote the Hadamard edges. Each spider $V_i$ has phase $\alpha_i$. A graph $\Tilde{G}=(\Tilde{E}_N,\Tilde{E}_H,\Tilde{V})$ equivalent to $G$ can always be found, preserves the Pauli flow of $G$, and each individual phase $\Tilde{\alpha}_i$ is independent to $G$.
\end{theorem}

\begin{proof}
Construct $\Tilde{G}=(\Tilde{E}_N,\Tilde{E}_H,\Tilde{V})$, with $2n$ spiders. The Hadamard edge $\Tilde{E}_H$ and regular edge $\Tilde{E}_N$ is given by  
 
\begin{equation}
    \Tilde{E}_{H} = \bigcup \{(\Tilde{V}_{2i+1},\Tilde{V}_{2j})\},~ \forall (V_i,V_j) \in E 
\end{equation}

\begin{equation}
    \Tilde{E}_{N} = \bigcup \{\Tilde{V}_{2i},\Tilde{V}_{2i+1}\},~\forall V_i\in V
\end{equation}

And the new phase $\Tilde{\alpha}_i$ is given by: 

\begin{equation}
    \Tilde{\alpha}_{2i}=\alpha_i - \beta_i 
\end{equation}

\begin{equation}
    \Tilde{\alpha}_{2i+1}=\beta_i 
\end{equation}

where $\beta_i$ is a random phase value. $\Tilde{G}$ can be rewrite to $G$ by applying the $\spiderrule$ to merge spider $\Tilde{V}_{2i}$ and $\Tilde{V}_{2i+1}$. Therefore $\Tilde{G}$ and $G$ are equivalent. Since $\Tilde{\alpha}_{2i} = \alpha_i - \beta_i + b_{2i}\pi$, $\Tilde{\alpha}_{2i+1} = \beta_i + b_{2i+1}\pi$, when $\beta_i$ is chosen uniformly random, $\Tilde{\alpha}_{2i}$ or $\Tilde{\alpha}_{2i+1}$ is independent from $\alpha_i$. Since $\Tilde{\alpha}_{2i}$ and $\Tilde{\alpha}_{2i+1}$ only dependent to $\alpha_i$ and $\beta_i$, each single of them is independent to $G$. This rewrite only uses $\spiderrule$,from lemma \ref{lemma:spider_rule_flow}, it preserves the Pauli flow.

\end{proof}

\begin{example}
To illustrate the phase obfuscation strategy, consider (a) the original graph-like ZX diagram describing the original algorithm. Each spider $V_i$ with phase $\alpha_i$ is split into two spiders connected with a regular wire. A random phase value $\beta_i$ is generated, and the phase for two new spiders $\Tilde{\alpha}_{2i} = \alpha_i - \beta_i$, $\Tilde{\alpha}_{2i+1}=\beta_i$, results in (b).

\vskip 1em
\centering

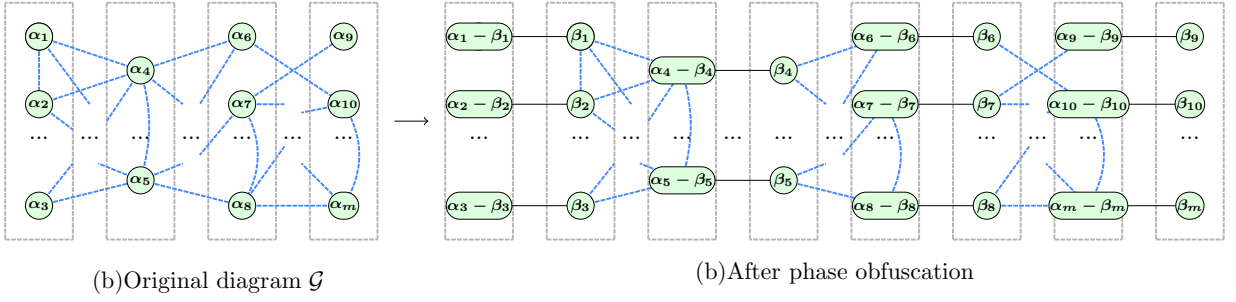
\[\scalebox{0.9}{\begin{tikzpicture}
	\begin{pgfonlayer}{nodelayer}
		\node [style=Z phase dot] (0) at (-12, 1.5) {};
		\node [style=Z phase dot] (1) at (-12, -0.5) {$\alpha_2$};
		\node [style=Z phase dot] (3) at (-12, -3.5) {$\alpha_3$};
		\node [style=none] (4) at (-12, -1.5) {...};
		\node [style=Z phase dot] (11) at (-9, 0.5) {$\alpha_4$};
		\node [style=Z phase dot] (12) at (-9, -2.75) {$\alpha_5$};
		\node [style=Z phase dot] (23) at (-3, 1.5) {$\alpha_{9}$};
		\node [style=Z phase dot] (24) at (-3, -0.5) {$\alpha_{10}$};
		\node [style=Z phase dot] (25) at (-3, -3.5) {$\alpha_{m}$};
		\node [style=none] (26) at (-3, -1.5) {...};
		\node [style=none] (27) at (-9, -1.25) {};
		\node [style=none] (28) at (-9, -1.5) {...};
		\node [style=none] (55) at (-13, 2.5) {};
		\node [style=none] (56) at (-11, 2.5) {};
		\node [style=none] (57) at (-13, -4.5) {};
		\node [style=none] (58) at (-11, -4.5) {};
		\node [style=none] (59) at (-4, -4.5) {};
		\node [style=none] (60) at (-2, -4.5) {};
		\node [style=none] (61) at (-4, 2.5) {};
		\node [style=none] (62) at (-2, 2.5) {};
		\node [style=none] (213) at (-10, 2.5) {};
		\node [style=none] (214) at (-8, 2.5) {};
		\node [style=none] (215) at (-8, -4.5) {};
		\node [style=none] (216) at (-10, -4.5) {};
		\node [style=none] (217) at (-10.5, -1.5) {...};
		\node [style=none] (218) at (-4.5, -1.5) {...};
		\node [style=none] (274) at (-10, -1) {};
		\node [style=none] (275) at (-10.25, -2.25) {};
		\node [style=none] (276) at (-7.75, -2.25) {};
		\node [style=none] (277) at (-7.75, -0.5) {};
		\node [style=none] (278) at (-4.25, -0.75) {};
		\node [style=none] (279) at (-4.25, -2.25) {};
		\node [style=none] (281) at (-10.5, -0.5) {};
		\node [style=none] (282) at (-11, -1.25) {};
		\node [style=none] (283) at (-11, -2.25) {};
		\node [style=none] (288) at (-7, -5.75) {(b)Original diagram $\mathcal{G}$};
		\node [style=Z phase dot] (289) at (-6, 1.5) {$\alpha_6$};
		\node [style=Z phase dot] (290) at (-6, -0.5) {$\alpha_7$};
		\node [style=Z phase dot] (291) at (-6, -3.5) {$\alpha_8$};
		\node [style=none] (292) at (-6, -1.5) {...};
		\node [style=none] (293) at (-4.75, -0.5) {};
		\node [style=none] (295) at (-4.75, -1.75) {};
		\node [style=none] (296) at (-7, -4.5) {};
		\node [style=none] (297) at (-5, -4.5) {};
		\node [style=none] (298) at (-7, 2.5) {};
		\node [style=none] (299) at (-5, 2.5) {};
		\node [style=none] (300) at (-7.5, -1.5) {...};
		\node [style=none] (301) at (-7.25, -2) {};
		\node [style=none] (303) at (-7.25, -0.5) {};
		\node [style=Z phase dot] (358) at (-12, 1.5) {$\alpha_1$};
		\node [style=Z phase dot] (359) at (4, 1.5) {};
		\node [style=Z phase dot] (360) at (4, -0.5) {$\beta_2$};
		\node [style=Z phase dot] (361) at (4, -3.5) {$\beta_3$};
		\node [style=none] (362) at (4, -1.5) {...};
		\node [style=Z phase dot] (363) at (7, 0.5) {$\alpha_4-\beta_4$};
		\node [style=Z phase dot] (364) at (7, -2.75) {$\alpha_5-\beta_5$};
		\node [style=Z phase dot] (365) at (19, 1.5) {$\alpha_9-\beta_9$};
		\node [style=Z phase dot] (366) at (19, -0.5) {$\alpha_{10}-\beta_{10}$};
		\node [style=Z phase dot] (367) at (19, -3.5) {$\alpha_m-\beta_m$};
		\node [style=none] (368) at (19, -1.5) {...};
		\node [style=none] (369) at (7, -1.25) {};
		\node [style=none] (370) at (7, -1.5) {...};
		\node [style=none] (371) at (3, 2.5) {};
		\node [style=none] (372) at (5, 2.5) {};
		\node [style=none] (373) at (3, -4.5) {};
		\node [style=none] (374) at (5, -4.5) {};
		\node [style=none] (375) at (18, -4.5) {};
		\node [style=none] (376) at (20, -4.5) {};
		\node [style=none] (377) at (18, 2.5) {};
		\node [style=none] (378) at (20, 2.5) {};
		\node [style=none] (379) at (6, 2.5) {};
		\node [style=none] (380) at (8, 2.5) {};
		\node [style=none] (381) at (8, -4.5) {};
		\node [style=none] (382) at (6, -4.5) {};
		\node [style=none] (383) at (5.5, -1.5) {...};
		\node [style=none] (384) at (17.5, -1.5) {...};
		\node [style=none] (385) at (6, -1) {};
		\node [style=none] (386) at (5.75, -2.25) {};
		\node [style=none] (387) at (8.25, -2.25) {};
		\node [style=none] (388) at (8.25, -0.5) {};
		\node [style=none] (389) at (17.75, -0.75) {};
		\node [style=none] (390) at (17.75, -2.25) {};
		\node [style=none] (391) at (5.5, -0.5) {};
		\node [style=none] (392) at (5, -1.25) {};
		\node [style=none] (393) at (5, -2.25) {};
		\node [style=Z phase dot] (394) at (16, 1.5) {$\beta_6$};
		\node [style=Z phase dot] (395) at (16, -0.5) {$\beta_7$};
		\node [style=Z phase dot] (396) at (16, -3.5) {$\beta_8$};
		\node [style=none] (397) at (16, -1.5) {...};
		\node [style=none] (398) at (17.25, -0.5) {};
		\node [style=none] (399) at (17.25, -1.75) {};
		\node [style=none] (400) at (15, -4.5) {};
		\node [style=none] (401) at (17, -4.5) {};
		\node [style=none] (402) at (15, 2.5) {};
		\node [style=none] (403) at (17, 2.5) {};
		\node [style=none] (404) at (8.5, -1.5) {...};
		\node [style=Z phase dot] (407) at (4, 1.5) {$\beta_1$};
		\node [style=Z phase dot] (408) at (1, 1.5) {};
		\node [style=Z phase dot] (409) at (1, -0.5) {$\alpha_2-\beta_2$};
		\node [style=Z phase dot] (410) at (1, -3.5) {$\alpha_3-\beta_3$};
		\node [style=none] (411) at (1, -1.5) {...};
		\node [style=Z phase dot] (412) at (10, 0.5) {$\beta_4$};
		\node [style=Z phase dot] (413) at (10, -2.75) {$\beta_5$};
		\node [style=Z phase dot] (414) at (22, 1.5) {$\beta_9$};
		\node [style=Z phase dot] (415) at (22, -0.5) {$\beta_{10}$};
		\node [style=Z phase dot] (416) at (22, -3.5) {$\beta_m$};
		\node [style=none] (417) at (22, -1.5) {...};
		\node [style=none] (418) at (10, -1.25) {};
		\node [style=none] (419) at (10, -1.5) {...};
		\node [style=none] (420) at (0, 2.5) {};
		\node [style=none] (421) at (2, 2.5) {};
		\node [style=none] (422) at (0, -4.5) {};
		\node [style=none] (423) at (2, -4.5) {};
		\node [style=none] (424) at (21, -4.5) {};
		\node [style=none] (425) at (23, -4.5) {};
		\node [style=none] (426) at (21, 2.5) {};
		\node [style=none] (427) at (23, 2.5) {};
		\node [style=none] (428) at (9, 2.5) {};
		\node [style=none] (429) at (11, 2.5) {};
		\node [style=none] (430) at (11, -4.5) {};
		\node [style=none] (431) at (9, -4.5) {};
		\node [style=none] (436) at (11.25, -2.25) {};
		\node [style=none] (437) at (11.25, -0.5) {};
		\node [style=Z phase dot] (443) at (13, 1.5) {$\alpha_6-\beta_6$};
		\node [style=Z phase dot] (444) at (13, -0.5) {$\alpha_7-\beta_7$};
		\node [style=Z phase dot] (445) at (13, -3.5) {$\alpha_8-\beta_8$};
		\node [style=none] (446) at (13, -1.5) {...};
		\node [style=none] (449) at (12, -4.5) {};
		\node [style=none] (450) at (14, -4.5) {};
		\node [style=none] (451) at (12, 2.5) {};
		\node [style=none] (452) at (14, 2.5) {};
		\node [style=none] (453) at (11.5, -1.5) {...};
		\node [style=none] (454) at (11.75, -2) {};
		\node [style=none] (455) at (11.75, -0.5) {};
		\node [style=Z phase dot] (456) at (1, 1.5) {$\alpha_1-\beta_1$};
		\node [style=none] (457) at (-1.5, -1) {};
		\node [style=none] (458) at (-0.5, -1) {};
		\node [style=none] (459) at (-0.5, -1) {};
		\node [style=none] (460) at (11.5, -5.5) {(b)After phase obfuscation};
	\end{pgfonlayer}
	\begin{pgfonlayer}{edgelayer}
		\draw [style=dash edge] (55.center) to (57.center);
		\draw [style=dash edge] (57.center) to (58.center);
		\draw [style=dash edge] (58.center) to (56.center);
		\draw [style=dash edge] (56.center) to (55.center);
		\draw [style=dash edge] (61.center) to (59.center);
		\draw [style=dash edge] (59.center) to (60.center);
		\draw [style=dash edge] (60.center) to (62.center);
		\draw [style=dash edge] (62.center) to (61.center);
		\draw [style=dash edge] (213.center) to (216.center);
		\draw [style=dash edge] (215.center) to (216.center);
		\draw [style=dash edge] (215.center) to (214.center);
		\draw [style=dash edge] (213.center) to (214.center);
		\draw [style=hadamard edge] (3) to (12);
		\draw [style=hadamard edge] (3) to (283.center);
		\draw [style=hadamard edge] (12) to (275.center);
		\draw [style=hadamard edge] (1) to (282.center);
		\draw [style=hadamard edge] (0) to (281.center);
		\draw [style=hadamard edge] (11) to (274.center);
		\draw [style=hadamard edge] (1) to (11);
		\draw [style=hadamard edge] (0) to (11);
		\draw [style=hadamard edge, bend left=15] (11) to (12);
		\draw [style=hadamard edge] (11) to (277.center);
		\draw [style=hadamard edge] (12) to (276.center);
		\draw [style=hadamard edge] (24) to (278.center);
		\draw [style=hadamard edge] (25) to (279.center);
		\draw [style=hadamard edge] (0) to (1);
		\draw [style=hadamard edge, bend left] (24) to (25);
		\draw [style=hadamard edge] (291) to (295.center);
		\draw [style=dash edge] (298.center) to (296.center);
		\draw [style=dash edge] (296.center) to (297.center);
		\draw [style=dash edge] (297.center) to (299.center);
		\draw [style=dash edge] (299.center) to (298.center);
		\draw [style=hadamard edge] (290) to (301.center);
		\draw [style=hadamard edge] (289) to (303.center);
		\draw [style=hadamard edge, bend left] (290) to (291);
		\draw [style=hadamard edge] (12) to (291);
		\draw [style=hadamard edge] (11) to (289);
		\draw [style=hadamard edge] (291) to (25);
		\draw [style=hadamard edge] (290) to (293.center);
		\draw [style=hadamard edge] (289) to (24);
		\draw [style=hadamard edge] (290) to (23);
		\draw [style=dash edge] (371.center) to (373.center);
		\draw [style=dash edge] (373.center) to (374.center);
		\draw [style=dash edge] (374.center) to (372.center);
		\draw [style=dash edge] (372.center) to (371.center);
		\draw [style=dash edge] (377.center) to (375.center);
		\draw [style=dash edge] (375.center) to (376.center);
		\draw [style=dash edge] (376.center) to (378.center);
		\draw [style=dash edge] (378.center) to (377.center);
		\draw [style=dash edge] (379.center) to (382.center);
		\draw [style=dash edge] (381.center) to (382.center);
		\draw [style=dash edge] (381.center) to (380.center);
		\draw [style=dash edge] (379.center) to (380.center);
		\draw [style=hadamard edge] (361) to (364);
		\draw [style=hadamard edge] (361) to (393.center);
		\draw [style=hadamard edge] (364) to (386.center);
		\draw [style=hadamard edge] (360) to (392.center);
		\draw [style=hadamard edge] (359) to (391.center);
		\draw [style=hadamard edge] (363) to (385.center);
		\draw [style=hadamard edge] (360) to (363);
		\draw [style=hadamard edge] (359) to (363);
		\draw [style=hadamard edge, bend left=15] (363) to (364);
		\draw [style=hadamard edge] (366) to (389.center);
		\draw [style=hadamard edge] (367) to (390.center);
		\draw [style=hadamard edge] (359) to (360);
		\draw [style=hadamard edge, bend left] (366) to (367);
		\draw [style=hadamard edge] (396) to (399.center);
		\draw [style=dash edge] (402.center) to (400.center);
		\draw [style=dash edge] (400.center) to (401.center);
		\draw [style=dash edge] (401.center) to (403.center);
		\draw [style=dash edge] (403.center) to (402.center);
		\draw [style=hadamard edge] (396) to (367);
		\draw [style=hadamard edge] (395) to (398.center);
		\draw [style=hadamard edge] (394) to (366);
		\draw [style=hadamard edge] (395) to (365);
		\draw [style=dash edge] (420.center) to (422.center);
		\draw [style=dash edge] (422.center) to (423.center);
		\draw [style=dash edge] (423.center) to (421.center);
		\draw [style=dash edge] (421.center) to (420.center);
		\draw [style=dash edge] (426.center) to (424.center);
		\draw [style=dash edge] (424.center) to (425.center);
		\draw [style=dash edge] (425.center) to (427.center);
		\draw [style=dash edge] (427.center) to (426.center);
		\draw [style=dash edge] (428.center) to (431.center);
		\draw [style=dash edge] (430.center) to (431.center);
		\draw [style=dash edge] (430.center) to (429.center);
		\draw [style=dash edge] (428.center) to (429.center);
		\draw [style=hadamard edge] (412) to (437.center);
		\draw [style=hadamard edge] (413) to (436.center);
		\draw [style=dash edge] (451.center) to (449.center);
		\draw [style=dash edge] (449.center) to (450.center);
		\draw [style=dash edge] (450.center) to (452.center);
		\draw [style=dash edge] (452.center) to (451.center);
		\draw [style=hadamard edge] (444) to (454.center);
		\draw [style=hadamard edge] (443) to (455.center);
		\draw [style=hadamard edge, bend left] (444) to (445);
		\draw [style=hadamard edge] (413) to (445);
		\draw [style=hadamard edge] (412) to (443);
		\draw (456) to (407);
		\draw (409) to (360);
		\draw (410) to (361);
		\draw (363) to (412);
		\draw (364) to (413);
		\draw (443) to (394);
		\draw (444) to (395);
		\draw (445) to (396);
		\draw (365) to (414);
		\draw (366) to (415);
		\draw (367) to (416);
		\draw [style=diredge] (457.center) to (459.center);
	\end{pgfonlayer}
\end{tikzpicture}
}\]

\end{example}

After this construction, we could show that $D(\Tilde{V}_{2i+1}) < D(\Tilde{V}_{2i})$

\subsection{Connectivity obfuscation}

Recall that in the phase obfuscation step, we have turned each spider into a pair of spiders connected with a regular wire, and two spiders in the pair now belong to different blocks. The wires within each block can be rewritten into a wire between the adjacent block. Such wire can be rewritten by disconnecting it from one spider and connecting it to the spider in an adjacent block which has a regular wire connected to the just disconnected from. 

\begin{theorem}[Internal connectivity to external connectivity]
    \label{fig:internal_wires}
 Given graph $G=(E_N,E_H,V)$ where $E_{Hi,j}=\{\Tilde{V}_{2i},\Tilde{V}_{2j}\}$ is an wire connect two spiders $\Tilde{V}_{2i}$ and $\Tilde{V}_{2j}$ in the same block. $E_{Hi,j}$ can be replaced with $\{\Tilde{V}_{2i+1},\Tilde{V}_{2j}\}$ , where $\Tilde{V}_{2i}$ and $\Tilde{V}_{2i+1}$ are connected with an regular wire. The rewrite rule preserves the Pauli flow of $G$. 
\end{theorem}

\begin{proof}
The graph can be constructed with the following rewrite. The rewrite uses only $\spiderrule$, from lemma \ref{lemma:spider_rule_flow} it preserves Pauli flow.
    \centering{
    \begin{tikzpicture}
	\begin{pgfonlayer}{nodelayer}
		\node [style=none] (384) at (-4.5, -2) {...};
		\node [style=Z dot] (395) at (-6, -1) {};
		\node [style=Z dot] (396) at (-6, -3) {};
		\node [style=none] (397) at (-6, -2) {...};
		\node [style=none] (398) at (-4, -1) {};
		\node [style=none] (400) at (-6.5, -3.75) {};
		\node [style=none] (401) at (-5.5, -3.75) {};
		\node [style=none] (402) at (-6.5, -0.25) {};
		\node [style=none] (403) at (-5.5, -0.25) {};
		\node [style=none] (436) at (-9.5, -3) {};
		\node [style=Z dot] (444) at (-8, -1) {};
		\node [style=Z dot] (445) at (-8, -3) {};
		\node [style=none] (446) at (-8, -2) {...};
		\node [style=none] (449) at (-8.5, -3.75) {};
		\node [style=none] (450) at (-7.5, -3.75) {};
		\node [style=none] (451) at (-8.5, -0.25) {};
		\node [style=none] (452) at (-7.5, -0.25) {};
		\node [style=none] (453) at (-9, -2) {...};
		\node [style=none] (454) at (-9.5, -1) {};
		\node [style=none] (457) at (-3, -2) {};
		\node [style=none] (458) at (-2, -2) {};
		\node [style=none] (459) at (-2, -2) {};
		\node [style=none] (462) at (-4, -3) {};
		\node [style=none] (463) at (3.5, -2) {...};
		\node [style=Z dot] (465) at (2.5, -1) {};
		\node [style=none] (467) at (2.5, -2) {...};
		\node [style=none] (468) at (4, -1) {};
		\node [style=none] (470) at (2, -3.75) {};
		\node [style=none] (471) at (3, -3.75) {};
		\node [style=none] (472) at (2, -0.25) {};
		\node [style=none] (473) at (3, -0.25) {};
		\node [style=none] (474) at (-1, -3) {};
		\node [style=Z dot] (477) at (0.5, -1) {};
		\node [style=Z dot] (478) at (1.5, -3) {};
		\node [style=none] (479) at (0.5, -2) {...};
		\node [style=none] (480) at (0, -3.75) {};
		\node [style=none] (481) at (1, -3.75) {};
		\node [style=none] (482) at (0, -0.25) {};
		\node [style=none] (483) at (1, -0.25) {};
		\node [style=none] (484) at (-0.5, -2) {...};
		\node [style=none] (485) at (-1, -1) {};
		\node [style=none] (488) at (4, -3) {};
		\node [style=none] (489) at (-2.5, -1) {$\spiderrule$};
		\node [style=none] (490) at (5, -2) {};
		\node [style=none] (491) at (6, -2) {};
		\node [style=none] (492) at (6, -2) {};
		\node [style=none] (493) at (11.5, -2) {...};
		\node [style=Z dot] (494) at (10.5, -1) {};
		\node [style=Z dot] (495) at (10.5, -3) {};
		\node [style=none] (496) at (10.5, -2) {...};
		\node [style=none] (497) at (12, -1) {};
		\node [style=none] (498) at (10, -3.75) {};
		\node [style=none] (499) at (11, -3.75) {};
		\node [style=none] (500) at (10, -0.25) {};
		\node [style=none] (501) at (11, -0.25) {};
		\node [style=none] (502) at (7, -3) {};
		\node [style=Z dot] (503) at (8.5, -1) {};
		\node [style=Z dot] (504) at (8.5, -3) {};
		\node [style=none] (505) at (8.5, -2) {...};
		\node [style=none] (506) at (8, -3.75) {};
		\node [style=none] (507) at (9, -3.75) {};
		\node [style=none] (508) at (8, -0.25) {};
		\node [style=none] (509) at (9, -0.25) {};
		\node [style=none] (510) at (7.5, -2) {...};
		\node [style=none] (511) at (7, -1) {};
		\node [style=none] (512) at (12, -3) {};
		\node [style=none] (513) at (5.5, -1) {$\spiderrule$};
	\end{pgfonlayer}
	\begin{pgfonlayer}{edgelayer}
		\draw [style=dash edge] (402.center) to (400.center);
		\draw [style=dash edge] (400.center) to (401.center);
		\draw [style=dash edge] (401.center) to (403.center);
		\draw [style=dash edge] (403.center) to (402.center);
		\draw [style=hadamard edge] (395) to (398.center);
		\draw [style=dash edge] (451.center) to (449.center);
		\draw [style=dash edge] (449.center) to (450.center);
		\draw [style=dash edge] (450.center) to (452.center);
		\draw [style=dash edge] (452.center) to (451.center);
		\draw [style=hadamard edge] (444) to (454.center);
		\draw [style=hadamard edge, bend left] (444) to (445);
		\draw (444) to (395);
		\draw (445) to (396);
		\draw [style=diredge] (457.center) to (459.center);
		\draw [style=hadamard edge] (436.center) to (445);
		\draw [style=hadamard edge] (396) to (462.center);
		\draw [style=dash edge] (472.center) to (470.center);
		\draw [style=dash edge] (470.center) to (471.center);
		\draw [style=dash edge] (471.center) to (473.center);
		\draw [style=dash edge] (473.center) to (472.center);
		\draw [style=hadamard edge] (465) to (468.center);
		\draw [style=dash edge] (482.center) to (480.center);
		\draw [style=dash edge] (480.center) to (481.center);
		\draw [style=dash edge] (481.center) to (483.center);
		\draw [style=dash edge] (483.center) to (482.center);
		\draw [style=hadamard edge] (477) to (485.center);
		\draw (477) to (465);
		\draw [style=hadamard edge] (474.center) to (478);
		\draw [style=hadamard edge] (478) to (465);
		\draw [style=diredge] (490.center) to (492.center);
		\draw [style=dash edge] (500.center) to (498.center);
		\draw [style=dash edge] (498.center) to (499.center);
		\draw [style=dash edge] (499.center) to (501.center);
		\draw [style=dash edge] (501.center) to (500.center);
		\draw [style=hadamard edge] (494) to (497.center);
		\draw [style=dash edge] (508.center) to (506.center);
		\draw [style=dash edge] (506.center) to (507.center);
		\draw [style=dash edge] (507.center) to (509.center);
		\draw [style=dash edge] (509.center) to (508.center);
		\draw [style=hadamard edge] (503) to (511.center);
		\draw (503) to (494);
		\draw (504) to (495);
		\draw [style=hadamard edge] (502.center) to (504);
		\draw [style=hadamard edge] (495) to (512.center);
		\draw [style=hadamard edge] (504) to (494);
		\draw [style=hadamard edge] (478) to (488.center);
	\end{pgfonlayer}
\end{tikzpicture}}
\end{proof}

\begin{theorem}\label{fig:double_extension}\cite{arxiv_2205.02009,arxiv_2304.08166}
Let $G = (V, E)$ be a graph with vertices $V$ and edges $E$. Suppose the labelled open graph $(G, I, O)$, and $\lambda(u) \in \{XY, X\}$ for all $u \in O^c$
, has Pauli flow. Pick an edge ${v, w} \in E$ and subdivide it twice, i.e. let $G^\prime := (V^\prime,E^\prime)$ where $V^\prime := V \cup {v^\prime,w^\prime}$ contains two new vertices $v^\prime$, $w^\prime$, and

\begin{equation}
    E^\prime = (E\slash \{\{v,w\}\})\cup \{\{v,w^\prime\},\{v^\prime,w^\prime\},\{v^\prime,w\}\}.
\end{equation}

Then $(G^\prime,I,O,\lambda^\prime)$ has Pauli flow, where

\begin{equation}
    \lambda^\prime(u) := 
\begin{cases}
    \lambda(u),& \text{if } u\in V \slash O\\
    X,         & \text{if } u\in \{v^\prime,w^\prime\}
\end{cases}
\end{equation}
    
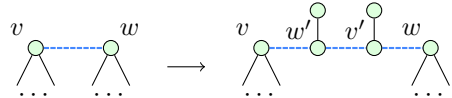
\centering
    \begin{tikzpicture}
	\begin{pgfonlayer}{nodelayer}
		\node [style=Z dot] (0) at (-4, 1) {};
		\node [style=Z dot] (1) at (-2, 1) {};
		\node [style=none] (2) at (-4.5, 0) {};
		\node [style=none] (3) at (-3.5, 0) {};
		\node [style=none] (4) at (-2.5, 0) {};
		\node [style=none] (5) at (-1.5, 0) {};
		\node [style=Z dot] (6) at (2, 1) {};
		\node [style=Z dot] (7) at (6.5, 1) {};
		\node [style=none] (8) at (1.5, 0) {};
		\node [style=none] (9) at (2.5, 0) {};
		\node [style=none] (10) at (6, 0) {};
		\node [style=none] (11) at (7, 0) {};
		\node [style=none] (12) at (-4, -0.25) {$\cdots$};
		\node [style=none] (13) at (-2, -0.25) {$\cdots$};
		\node [style=none] (14) at (2, -0.25) {$\cdots$};
		\node [style=none] (15) at (6.5, -0.25) {$\cdots$};
		\node [style=Z dot] (16) at (3.5, 1) {};
		\node [style=Z dot] (17) at (5, 1) {};
		\node [style=Z dot] (18) at (3.5, 2) {};
		\node [style=Z dot] (19) at (5, 2) {};
		\node [style=none] (20) at (1.5, 1.5) {$v$};
		\node [style=none] (21) at (3, 1.5) {$w^\prime$};
		\node [style=none] (22) at (4.5, 1.5) {$v^\prime$};
		\node [style=none] (23) at (6, 1.5) {$w$};
		\node [style=none] (24) at (-4.5, 1.5) {$v$};
		\node [style=none] (25) at (-1.5, 1.5) {$w$};
		\node [style=none] (26) at (-0.5, 0.5) {};
		\node [style=none] (27) at (0.5, 0.5) {};
	\end{pgfonlayer}
	\begin{pgfonlayer}{edgelayer}
		\draw [style=hadamard edge] (0) to (1);
		\draw (0) to (2.center);
		\draw (0) to (3.center);
		\draw (1) to (4.center);
		\draw (1) to (5.center);
		\draw (6) to (8.center);
		\draw (6) to (9.center);
		\draw (7) to (10.center);
		\draw (7) to (11.center);
		\draw [style=hadamard edge] (6) to (16);
		\draw [style=hadamard edge] (16) to (17);
		\draw [style=hadamard edge] (17) to (7);
		\draw (18) to (16);
		\draw (19) to (17);
		\draw [style=diredge] (26.center) to (27.center);
	\end{pgfonlayer}
\end{tikzpicture}
 
\end{theorem}

\begin{theorem}[Obfuscate the connectivity between blocks]\label{fig:seperate_wires}
Given a semi-graph ZX diagram $G(E_N, E_H, V)$ generated from phase obfuscation, an equivalent, Pauli flow preserving semi-graph ZX-daigram $\Tilde{G}(\Tilde{E}_N,\Tilde{E}_H,\Tilde{V})$ and a block partition $B$ can be found, such that for edges $\Tilde{E}=(\Tilde{V}_i,\Tilde{V}_j)$ within the same partition $\Tilde{V}_i\in B_k$,$\Tilde{V}_j \in B_k$, all edges depends only on $deg(V_i)$, $V_i\in V$.
\end{theorem}

\begin{proof}
To find $\Tilde{G}$, we first apply the rewrite rule from theorem \ref{fig:internal_wires} and remove wires within the blocks. Then for all $E_{i,j}=\{V_i,V_j\} \in E_H$, we apply rewrite rules from \ref{fig:double_extension}. Denote the newly added spider for each edge $E_{i,j}$ as $W_{i,j}$ and $W^\prime_{i,j}$. Assign $W_{i,j}$ to block $B(V_i)$ and $W^\prime_{i,j}$ to $B(V_j)$.

After this rewrite, all the edges within the blocks connect to an extra spider $W_{i,j}$ and $W^\prime_{i,j}$. For any two spiders $W_A$ and $W_B$ connect to $V_i$, they can be only distinguished by the external connection. Therefore all the edges $(V_i,W_{i,j})$ and $(V_j,W^\prime_{i,j})$ within the same block doesn't depends on $E_H$. However since multiple $W$ spiders may still connect to a single $V$ spider, the internal wires depend on the degree of spider $V_i$ and $V_j$. 

Because the rewrite rule used in obfuscation is from theorem \ref{fig:internal_wires} and  \ref{fig:double_extension} and they both preserve Pauli flow, the rewrite for obfuscation also preserves Pauli flow.
\end{proof}

\begin{example}\label{example:pivot_grow_spider}
The process of connectivity obfuscation between blocks can be illustrated as follows. Suppose a fragment of the quantum algorithm looks like (a). Here all the phases in the spiders are not included in the diagram. First extra spiders are created with rule $\spiderrule$ to ensure no two inter-block wires are connected to the same spider, see (b). Then extra dummy spiders are created to obfuscate the spiders' degree (wire connected to the same spider). Each dummy spider has only two wires and would connect another dummy spider or a hub spider. See (c). Now for each inter-block wire, a random id $s_{i,k}$ is generated for its identification, see (d). The spiders' order can be randomly shuffled in each block, as long as the connectivity remains the same. See (e). Two adjacent blocks are assigned to a different quantum agent. Each agent only needs the wire identity $s_{i,k}$ to establish the correct entanglement. 
\vskip 1em
\centering
\input{abstract/seperate_edges.tikz}
\end{example}
  
The connectivity obfuscation restricts each agent to have only the label of edges $s_{i,j}$ instead of the actual qubit connected in the adjacent agents. Such obfuscation prevents the connectivity configuration of the ZX\-diagram from being reconstructed. After the obfuscation, the leg connects to hub spiders are all in an equivalent position; therefore agent cannot distinguish the direction of information flow during the execution. In practice, the $s_{i,j}$ can be used to identify the pre-shared bell pairs between the agents. Each agent would not be able to know the entanglement structure of other agents.

The connectivity obfuscation step hides the agent's other end of the entanglement. However, the actual required entanglement can be estimated by the agent by counting the number of entanglements within its block, connected to some \textit{hub spiders}. To further obfuscate the resource requirement of the quantum algorithm, we need to modify the number of connectivity of these hub spiders. We could add dummy qubit resources to obfuscate the exact resource requirement of the quantum algorithm. This can be done by attaching two phase-free spiders to the existing graph and connecting them with Hadamard edges.

\begin{theorem}\label{fig:appendage_extension}\cite{arxiv_2205.02009,arxiv_2304.08166}
Let $G = (V, E)$ be a graph with vertices $V$ and edges $E$. Suppose the labelled open graph $(G, I, O)$, and $\lambda(u) \in \{XY, X\}$ for all $u \in O^c$
, has Pauli flow. Pick a node ${u} \in E$ and append two new vertices connected by a Hadamad edge, i.e. let $G^\prime := (V^\prime, E^\prime)$ where $V^\prime := V \cup {v,w}$ contains two new vertices $v^\prime$, $w^\prime$, and

\begin{equation}
    E^\prime = E\cup \{\{v,w\},\{w,w^\prime\}\}.
\end{equation}

Then $(G^\prime,I,O,\lambda^\prime)$ has Pauli flow, where

\begin{equation}
    \lambda^\prime(u) := 
\begin{cases}
    \lambda(u),& \text{if } u\in V \slash O\\
    X,         & \text{if } u\in \{w,w^\prime\}
\end{cases}
\end{equation}
    
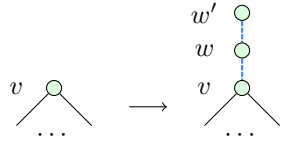
\centering
    \begin{tikzpicture}
	\begin{pgfonlayer}{nodelayer}
		\node [style=Z dot] (6) at (6, 1) {};
		\node [style=none] (8) at (5, 0) {};
		\node [style=none] (9) at (7, 0) {};
		\node [style=none] (14) at (6, -0.25) {$\cdots$};
		\node [style=none] (20) at (5, 1) {$v$};
		\node [style=none] (26) at (3, 0.5) {};
		\node [style=none] (27) at (4, 0.5) {};
		\node [style=Z dot] (28) at (1, 1) {};
		\node [style=none] (30) at (0, 0) {};
		\node [style=none] (31) at (2, 0) {};
		\node [style=none] (34) at (1, -0.25) {$\cdots$};
		\node [style=none] (38) at (0, 1) {$v$};
		\node [style=Z dot] (39) at (6, 2) {};
		\node [style=Z dot] (40) at (6, 3) {};
		\node [style=none] (41) at (5, 2) {$w$};
		\node [style=none] (42) at (5, 3) {$w^\prime$};
	\end{pgfonlayer}
	\begin{pgfonlayer}{edgelayer}
		\draw (6) to (8.center);
		\draw (6) to (9.center);
		\draw [style=diredge] (26.center) to (27.center);
		\draw (28) to (30.center);
		\draw (28) to (31.center);
		\draw [style=hadamard edge] (40) to (39);
		\draw [style=hadamard edge] (39) to (6);
	\end{pgfonlayer}
\end{tikzpicture}
 
\end{theorem}

\begin{theorem}[Dummy resources]\label{fig:extra_resources}

Given a ZX-diagram $G(E_N, E_H, V)$ and partition $B$ generated from connectivity obfuscation. An equivalent, Pauli flow preserved graph $\Tilde{G}(\Tilde{E}_N,\Tilde{E}_H,\Tilde{V})$ and a partition $\Tilde{B}$ can be found, such that the hub spider has a larger degree.
\end{theorem}

\begin{proof}
\vskip 1em
Suppose we want to increase the degree of spider $V_i\in B_j$. Apply the rewrite rule from theorem \ref{fig:appendage_extension} to the spider $V_i$, and we have two added spider $w$ and $w^\prime$. Assign $w$ to $B_{j-1}$ and $w^\prime$ to $B_{j-2}$. The rewrite graph is equivalent to the original graph, preserves the Pauli flow, and increases the degree of $V_i$.

\end{proof}

\subsection{Measurement result obfuscation}

\SXE{So far, we have made input information to each agent independent of the algorithm. However, each agent may obtain information from their measurement outcome. Here we show some information may leak out from the measurement distribution of intermediate blocks, if the measurement outcome is not further obfuscated. To illustrate it, first consider a circuit is folded into two piece and each part of the ZX-Diagram is executed on a different agent. See Example \ref{fig:circuit_bending_abstract}.}

\begin{example}\label{fig:circuit_bending_abstract}
\SXE{Fold a quantum circuit. On (a) we rewrite our quantum circuit into a ZX-diagram and then divide them into two quantum circuits noted as $U_1$ and $U_2$. Then we fold the ZX-diagram in (b) and rewrite the folded connection between $U_1$ and $U_2$ in (c). Eventually, we add the initial state of the quantum circuit in (d). The ZX-diagram in (d) is ready to be extracted into a quantum circuit with shallower depth but used twice as the original quantum circuit.}

    \centering
    \input{abstract/circuit_bending.tikz} 
\end{example}

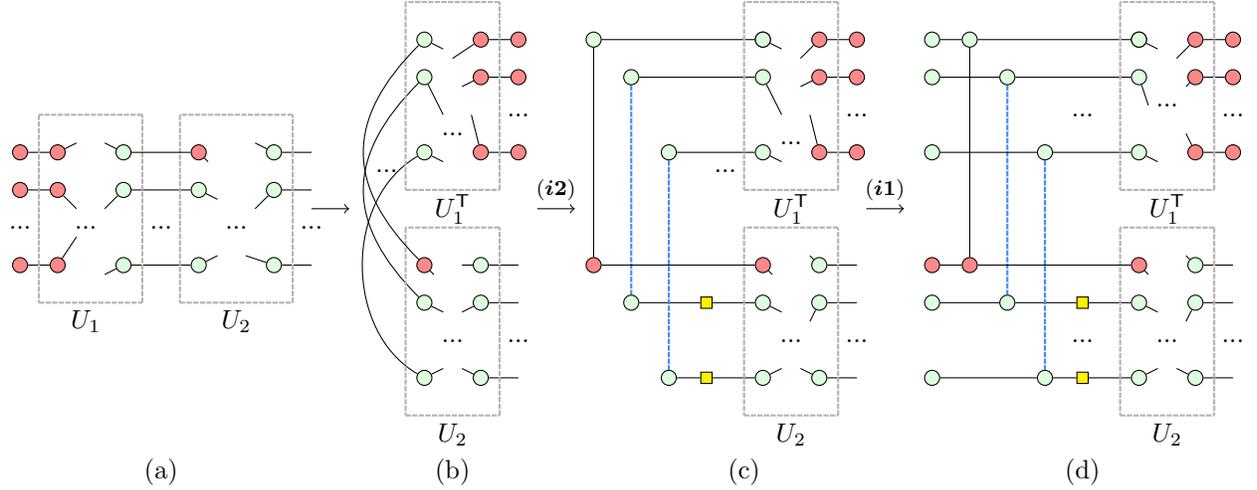
\SXE{Now we would like to understand the measurement result distribution of the agent executing the upper half of the diagram. }

\begin{example}\label{fig:circuit_bending_overhead}
\SXE{The reduced density matrix of a folded circuit. The reduced density matrix can be expressed by adding a dual of the existing quantum circuit and connecting the qubits that need to be traced away. Here we can show that after tracing away the qubits containing the computation result, the ZX-diagram of the reduced density matrix is an identity. This identity indicates that the measurement distribution of these qubits is uniform.}

    
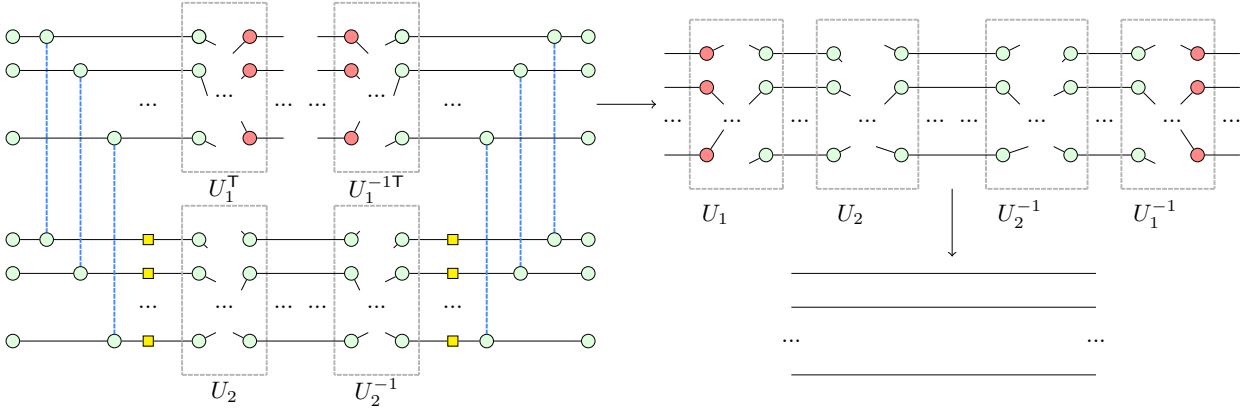
\centering

\[\scalebox{0.9}{\input{abstract/circuit_bending_overhead.tikz}}\]
\end{example}

\SXE{The graph-like ZX-diagram can be considered folding the circuit until it has only one operation before it gets measured. We can apply the same analysis to our protocol.}

\begin{example}\label{fig:recovery_leackage}

\SXE{Information leakage from correction. Here we show that there could be information leakage from the agent's observation when correction is applied. To understand each agent's measurement distribution, we first arrange the spiders into the same column. Then we generate its conjugate diagram next to the existing diagram. The observed distribution described by the reduced density matrix can be obtained by tracing away the unmeasured qubits. However, we always apply the correction to restore the quantum state for measured qubits. Therefore it is equivalent to measuring the zero-phase spider on these qubits. The reduced density matrix is shown in (a). Now we try to simplify the quantum circuit; most of the unmeasured qubits can be traced away. However, it still leaves a graph that is not necessarily identity—shown in (b). The non-identity graph indicates the agent can observe a non-uniformed distribution, which may carry useful information.}

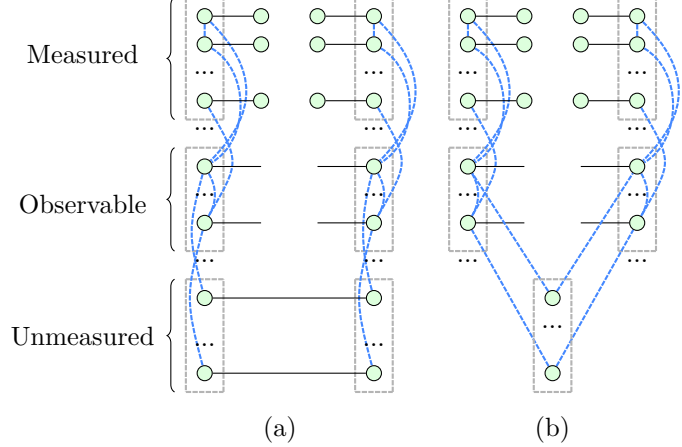
\centering
    \input{abstract/recovery_leakage.tikz}
\end{example}

\SXE{The observed distribution of each agent is characterized by the reduced density matrix. A non-uniform distribution indicates a potential information leakage. Now we introduce the measurement result obfuscation approach to resolve the information leakage from the non-uniform measurement outcome distribution.}

\begin{theorem}[Measurement result obfuscation]\label{fig:measurement_obfuscation}
 Given ZX-Diagram $G(E_N, E_H, V)$ executes with MBQC methods. For each spider $V_i\in V$ with phase $\alpha_i$, generate a random bit $b_i$ and define $\Tilde{\alpha}_i = \alpha_i + b_i \pi$, the measured distribution is independent of the diagram and the distribution for the calculation result can be reconstructed classically by the client.
\end{theorem}

\begin{proof}
 Suppose for each shot a new diagram $\Tilde{G}(\Tilde{E}_N,\Tilde{E}_H,\Tilde{V})$ is constructed at the execution time. Consider the measurement outcome for measuring spider $\Tilde{V}_i$ is $\Tilde{r}_i$. The result of executing $\Tilde{G}$  is equivalent to $G$ when we consider $r_i = \Tilde{r}_i \oplus b_i$. Since $b_i$ is chosen randomly, the distribution of $\Tilde{r}_i$ is random and independent to the diagram $G$. 
\end{proof}

\begin{example}
 To illustrate the obfuscation of the readout distribution for each qubit, we introduce a bit string $b_{i}$. Suppose a fragment of the ZX diagram is shown as (a). For each $V_i$, we add two connected spiders with phase $b_i\pi$. This is equivalent to adding $2\pi b_i$ to each $V_i$ as (b). Then one of the spiders is removed and converted into a classical flip operation. Merge the other spiders, and we have a new diagram as (c). For each $\Tilde{V}_i$, we add a $\pi$ phase if $b_i$ is 1, otherwise, keep it the same.
\vskip 1em

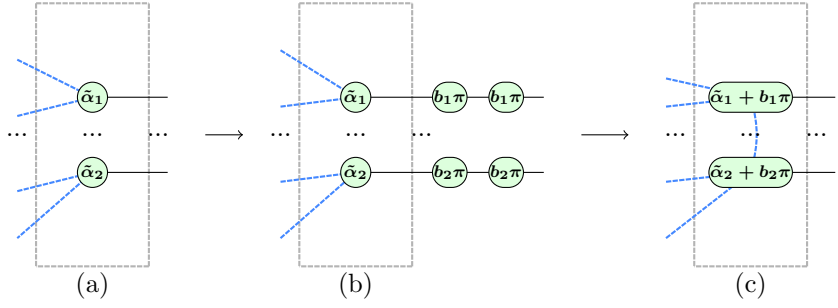
\centering
\begin{tikzpicture}
	\begin{pgfonlayer}{nodelayer}
		\node [style=Z phase dot] (4) at (-9, 1) {$\Tilde{\alpha}_1$};
		\node [style=Z phase dot] (5) at (-9, -1) {$\Tilde{\alpha}_2$};
		\node [style=none] (10) at (-9, 0) {};
		\node [style=none] (11) at (-9, 0) {...};
		\node [style=none] (23) at (-10.5, 3.5) {};
		\node [style=none] (24) at (-7.5, 3.5) {};
		\node [style=none] (25) at (-7.5, -3.5) {};
		\node [style=none] (26) at (-10.5, -3.5) {};
		\node [style=none] (31) at (-11, -2.75) {};
		\node [style=none] (48) at (-7.25, 0) {...};
		\node [style=none] (110) at (4, 0) {};
		\node [style=none] (111) at (5.25, 0) {};
		\node [style=none] (116) at (-7, 1) {};
		\node [style=none] (117) at (-7, -1) {};
		\node [style=none] (190) at (-6, 0) {};
		\node [style=none] (191) at (-5, 0) {};
		\node [style=none] (262) at (-9, -4) {(a)};
		\node [style=none] (395) at (-11, 2) {};
		\node [style=none] (396) at (-11, 0.5) {};
		\node [style=none] (397) at (-11, -1.5) {};
		\node [style=none] (398) at (-11, 0) {...};
		\node [style=Z phase dot] (462) at (-2, 1) {$\Tilde{\alpha}_1$};
		\node [style=Z phase dot] (463) at (-2, -1) {$\Tilde{\alpha}_2$};
		\node [style=none] (464) at (-2, 0) {};
		\node [style=none] (465) at (-2, 0) {...};
		\node [style=none] (466) at (-3.5, 3.5) {};
		\node [style=none] (467) at (-0.5, 3.5) {};
		\node [style=none] (468) at (-0.5, -3.5) {};
		\node [style=none] (469) at (-3.5, -3.5) {};
		\node [style=none] (481) at (-0.25, 0) {...};
		\node [style=none] (489) at (-2, -4) {(b)};
		\node [style=none] (490) at (-4, 2.25) {};
		\node [style=none] (491) at (-4, 0.75) {};
		\node [style=none] (492) at (-4, -1.25) {};
		\node [style=none] (493) at (-4, 0) {...};
		\node [style=Z phase dot] (504) at (0.5, 1) {$b_1\pi$};
		\node [style=Z phase dot] (505) at (2, 1) {$b_1\pi$};
		\node [style=Z phase dot] (506) at (0.5, -1) {$b_2\pi$};
		\node [style=Z phase dot] (507) at (2, -1) {$b_2\pi$};
		\node [style=none] (510) at (3, 1) {};
		\node [style=none] (511) at (3, -1) {};
		\node [style=Z phase dot] (513) at (8.5, 1) {$\Tilde{\alpha}_1+b_1\pi$};
		\node [style=Z phase dot] (514) at (8.5, -1) {$\Tilde{\alpha}_2+b_2\pi$};
		\node [style=none] (516) at (8.5, 0) {...};
		\node [style=none] (517) at (7, 3.5) {};
		\node [style=none] (518) at (10, 3.5) {};
		\node [style=none] (519) at (10, -3.5) {};
		\node [style=none] (520) at (7, -3.5) {};
		\node [style=none] (532) at (10.25, 0) {...};
		\node [style=none] (535) at (10.75, 1) {};
		\node [style=none] (536) at (10.75, -1) {};
		\node [style=none] (540) at (8.5, -4) {(c)};
		\node [style=none] (541) at (6.25, 1.5) {};
		\node [style=none] (542) at (6.25, 0.75) {};
		\node [style=none] (543) at (6.25, -1.25) {};
		\node [style=none] (544) at (6.5, 0) {...};
		\node [style=none] (545) at (-4, -2.75) {};
		\node [style=none] (546) at (6.25, -2.75) {};
		\node [style=none] (547) at (8.5, -3.25) {};
	\end{pgfonlayer}
	\begin{pgfonlayer}{edgelayer}
		\draw [style=dash edge] (23.center) to (26.center);
		\draw [style=dash edge] (25.center) to (26.center);
		\draw [style=dash edge] (25.center) to (24.center);
		\draw [style=dash edge] (23.center) to (24.center);
		\draw [style=hadamard edge] (5) to (31.center);
		\draw [style=diredge] (110.center) to (111.center);
		\draw (116.center) to (4);
		\draw (117.center) to (5);
		\draw [style=diredge] (190.center) to (191.center);
		\draw [style=hadamard edge] (395.center) to (4);
		\draw [style=hadamard edge] (396.center) to (4);
		\draw [style=hadamard edge] (397.center) to (5);
		\draw [style=dash edge] (466.center) to (469.center);
		\draw [style=dash edge] (468.center) to (469.center);
		\draw [style=dash edge] (468.center) to (467.center);
		\draw [style=dash edge] (466.center) to (467.center);
		\draw [style=hadamard edge] (490.center) to (462);
		\draw [style=hadamard edge] (491.center) to (462);
		\draw [style=hadamard edge] (492.center) to (463);
		\draw (504) to (505);
		\draw (506) to (507);
		\draw (504) to (462);
		\draw (506) to (463);
		\draw (510.center) to (505);
		\draw (511.center) to (507);
		\draw [style=dash edge] (517.center) to (520.center);
		\draw [style=dash edge] (519.center) to (520.center);
		\draw [style=dash edge] (519.center) to (518.center);
		\draw [style=dash edge] (517.center) to (518.center);
		\draw [style=hadamard edge, bend left=15] (513) to (514);
		\draw (535.center) to (513);
		\draw (536.center) to (514);
		\draw [style=hadamard edge] (541.center) to (513);
		\draw [style=hadamard edge] (542.center) to (513);
		\draw [style=hadamard edge] (543.center) to (514);
		\draw [style=hadamard edge] (545.center) to (463);
		\draw [style=hadamard edge] (546.center) to (514);
	\end{pgfonlayer}
\end{tikzpicture}
\end{example}

To understand this more easily, consider a quantum circuit that generates a binary distribution. Such distribution can be hidden by randomly applying a $\pi$ rotation to the qubit, swapping the probability of $\ket{0}$ and $\ket{1}$ just before the measurement. Then the original distribution can be restored by classically swapping them back.

Suppose the distribution without measurement obfuscation is $\Tilde{r_j}$, and the distribution after measurement obfuscation is $r_j$. The extra $\pi$ phase swaps the distribution of measuring the qubit with $\pi$ phase or zero phases. Therefore $r_j = \Tilde{r_j} \oplus b_j$, where $\oplus$ denote the bit-wise exclusive or operation. If $b_{i}$ is chosen uniformly random, the measured result would be uniform. The correction process would be intuitive: $\Tilde{r_j} = r_j \oplus b_j$. If we have measured a $\pi$ phase and have already added a $\pi$ phase to the spider, it cancels out if we have measured a zero phase and have added a $\pi$ phase to the spider, it is equivalent to measuring a $\pi$ phase without modifying the phase of the spider. 

\begin{theorem}[Measurement result independence]\label{theorem:measurement_result_independence}
The distribution of measurement results $r_i$ is independent of the executed quantum algorithm.
\end{theorem}
\begin{proof}
The measurement result without measurement obfuscation $\Tilde{r_j}$ can be non-uniform. The measurement result observed by each agent is $r_j = \Tilde{r}_j \oplus b_j$. With the $b_j$ chosen uniformly random, the measurement distribution of $r_j$ would be uniform. 
\end{proof}

\subsection{Circuit extraction}

So far, we have generated a ZX-diagram, which needs to be extracted into physical quantum operations. Note that the graph-like ZX-diagram contains only Hadamard wires. The regular wires between blocks come from the phase obfuscation step when each spider is split into two and connected with regular wires. So each spider is connected to a maximum of one regular wire to spiders at other agents. 

The extraction can be implemented with the following method:

\begin{theorem}\label{fig:wire_extraction}
With a given obfuscated graph $\Tilde{G}(\Tilde{E}_N,\Tilde{E}_H,\Tilde{V})$ where $size(\Tilde{E}_N)=1$ , i.e. each node would connect to multiple Hadamard edges and maximum 1 regular edge. $\Tilde{G}$ can be extracted into physical quantum operations and implemented on a quantum device. The extracted physical operation to implement an edge is independent of the edge type.  
\end{theorem}

\begin{proof}
Each edge in graph $\Tilde{G}$ can be extracted into quantum operations with the following method. Hadamard wires can be extracted into a CZ operation, or applying CZ operation first, then applying Hadamard gate on both qubits. It is shown in (a). Regular wires can be extracted into a CZ operation and then apply a Hadamard gate on only one side. By randomly choosing the method to extract the circuit, whether the Hadamard gate exists is independent of the type of wire being extracted. 

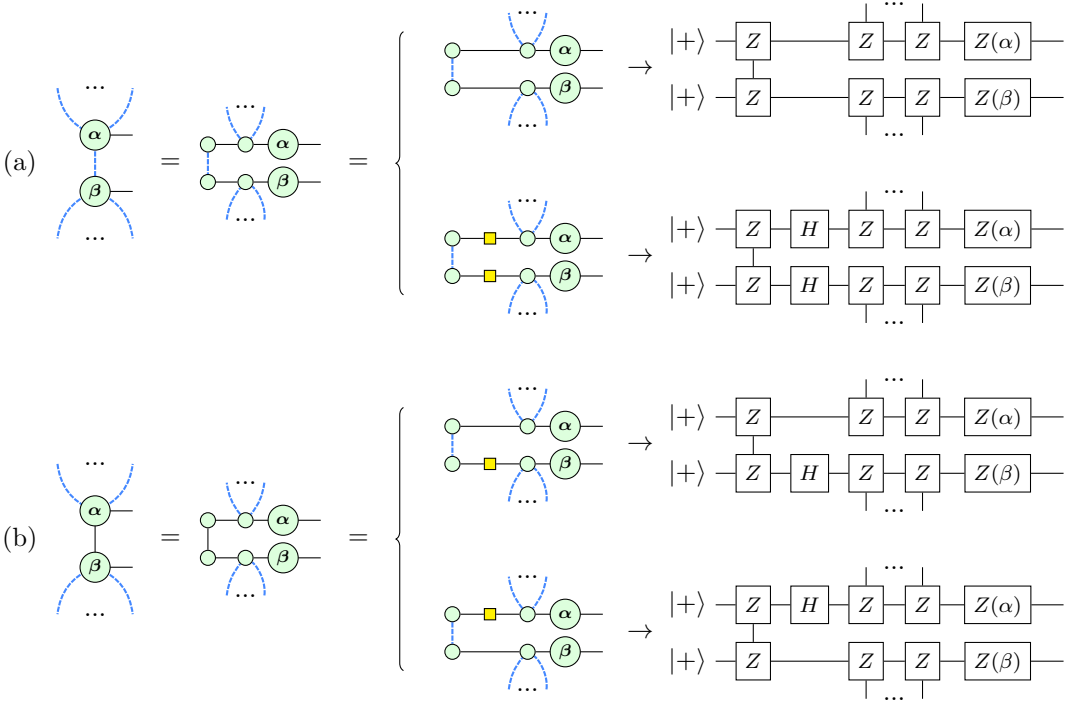
\centering
    \input{abstract/edge_extraction.tikz}
\end{proof}

\subsection{Summarize the protocol}

The information describing the quantum algorithm consists of the phases of each spider and the connectivity between spiders under the perspective of ZX-diagram. Also, the measurement outcome would cause information leakage when the correction process is applied. Our protocol provides a complete solution to obfuscate information from these three aspects. First, our protocol utilizes the same strategy as the UBQC protocol to obfuscate the phase information and the measurement outcome. The phase rotation operation is split into two and performed by different agents. Then the measurement outcome is obfuscated by randomly flipping the quantum distribution with a phase difference of $\pi$. The rotation phase evaluation happens during the execution process to update the correction into the phase in real time. Finally, for the connectivity, our proposal moves all the connectivity information that reveals the algorithm as a wire between two different agents to hide the connectivity of the diagram. Since each agent cannot access the information from its neighbouring agent, it loses track of the information on the other side of the wire. Here we present the formal description of our protocol. It contains two major components: The client's preparation step, which obfuscates and generates proper ZX-diagram blocks for each agent. Then, in the execution step, the client interacts with each agent to implement calculations and retrieve results. 

\SXE{
 Our protocol does not assume the input state is fully classical. For the case that the input state contains quantum data, it can be prepared by teleporting the quantum data to the agents, then make the teleport data into a segment of the ZX-Diagram for computations. The client can remain fully classical to handle the quantum data by relying on a trusted third party to supply the quantum data. 
}
\clearpage

\begin{protocol}{MBQC BQC without universal cluster state. }
\textit{Inputs.} 
\begin{enumerate}[nosep]
    \item $\mathcal{G}(E,V)$ : A graph-like ZX-diagram describes the quantum algorithm $\Lambda$. $E$,$V$ denote the connections and spiders in the graph. For classical data described input state, the preparation circuit is included in $\mathcal{G}$.
    \item $\rho_0$: The input state of the quantum algorithm.
    \item $A=\{A_1...A_m\}, m \geq 2$ : The available quantum agents. \SXE{The number of agents $m$ can be chosen arbitrarily, provided it is greater than or equal to 2.}
\end{enumerate} 
\textit{Definitions.}
\begin{enumerate}[nosep]
    \item $\mathcal{\Tilde{G}}(\Tilde{E},\Tilde{V})$. The processed ZX-diagram for execution.
    \item $V_i$: the i-th spider in $\mathcal{G}$.
    \item $\alpha_i$: the phase of i-th spider in $\mathcal{G}$.
    \item $\Tilde{V_j}$: the j-th spider in $\Tilde{\mathcal{G}}$.
    \item $\Tilde{\alpha_j}$: the phase of j-th spider in $\Tilde{\mathcal{G}}$.
    \item $\beta_j$ : the random value generated for phase obfuscation for $V_i$.
    \item $b_i$: A random bit for measurement obfuscation of i-th qubits.
    \item $B_k = \{\Tilde{V_j}\}$: The k-th block of $\Tilde{\mathcal{G}}$.
    \item $d_j = d(\Tilde{V_j})$: The distance of $V_j$ to its nearest output spider in $\Tilde{\mathcal{G}}$.
    \item $n = max(d_j)$: the number of fragmented blocks in  $\Tilde{\mathcal{G}}$. 
    \item $r_j$: The measurement result of $\Tilde{V_j}$.
    \item $\Tilde{r}_j$: The corresponding measurement result of $\Tilde{V_j}$ without measurement obfuscation.
\end{enumerate}
\textit{Goal.} Retrieve the measurement distribution of $\Lambda(\rho_0)$.
%
  \paragraph{Preparation.}
  \begin{enumerate}
    \item  The client split each $V_i \in V$ spider into two spiders $\Tilde{V}_{2i}$ and $\Tilde{V}_{2i+1}$ with rule $\spiderrule$, each spider has phase $\Tilde{\alpha}_{2i}$ and $\Tilde{\alpha}_{2i+1}$.  Note that $\Tilde{\alpha}_{2i}$ and $\Tilde{\alpha}_{2i+1}$ are symbols for placeholder, the actual value of will be evaluated in the later steps.
    \item The client rearrange the ZX-diagram and group spiders into $n$ blocks $B_k = {\Tilde{V_j}}$ where the distance to output spider $d_j=k$.
     \item
    For each connectivity within the same block, the client uses rule $\spiderrule$ as example \ref{fig:internal_wires} to move it to the adjacent block. 
    \item The client split each Hadamard edge between blocks into two empty spiders and three edges, as shown in example \ref{fig:seperate_wires}. The two spiders are assigned to the block that their neighbour spider belongs to.  
    \item The client analyze $\Tilde{G}$ and find a flow.
  \end{enumerate}

  \paragraph{Execution.}
  \begin{enumerate}
    \item The client assign block $B_j$ to agent $A_i$ when $j~mod~\SXE{m} = i$. \SXE{Fragment block $B_j$ are found by the rules from section \ref{sec:define_blocks}.} 
    \item For each sample
    \begin{enumerate}
      \item  The client generate random phase values $\{\beta_i\}$ and random bit $\{b_j\}$. 
      \item The client assign $\Tilde{\alpha}_{2i}=\alpha_i-\beta_i + b_{2i}\pi$ and $\Tilde{\alpha}_{2i+1}=\beta_i+b_{2i+1}\pi$.
      \item The client randomly assigns spiders to qubits and allocates resources from each agent.
      \item Agents reset all qubits in all the blocks into $\ket{+} state.$ For quantum data, teleport the input state into the input qubits and set all the other qubits into  $\ket{+}$ state.
      \item The client extracts the ZX-diagram into quantum operations with example \ref{example:pivot_grow_spider}. Then request agents to establish shared entanglement between agents based on $\{s_i,j\}$.
      \item Follows the flow of $G$ to execute the diagram. Handle $V_i$ in ascending order of the partial order of the flow. For all spiders $\Tilde{V_j}$ that splits from $V_i$
        \begin{enumerate}
           \item The client sends $\Tilde{\alpha}_j$ to the corresponding agent, requesting the agent to measure qubits in \SXE{XY plane} with angle of $-\Tilde{\alpha}_j$.
           \item The client get results $r_j$, calculate the $\Tilde{r}_j = b_j \oplus r_j$. 
           \item The client calculates the effectively measured phase $r=\sum r_j$ on $V_i$, 
           \item The client make changes to $\alpha_j$ for correction based on $r$ and the ZX-diagram with example \ref{fig:pi_propergation_measurement_sequence}.
       \end{enumerate}
    \end{enumerate}
    \item The client returns the sampled distribution of $\Tilde{r}_o$ where $\Tilde{V_o}$ is a output spider. 
  \end{enumerate}
\end{protocol}

\section{Proof of correctness and blindness\label{sec:proof} }

In this section, we go through the techniques used to protect the information and give proof of the correctness and blindness of our protocol. First, we provide the definition of blindness.

\begin{definition}[Blindness]\label{definition:blindness}
Let P be a quantum delegated computation on input $X$ and let $L(X)$ be any function
of the input. We say that a quantum delegated computation protocol is blind while leaking at most
$L(X)$ if, on client’s input $X$, for any fixed $Y = L(X)$, the following two hold when given $Y$
\begin{enumerate}
    \item The distribution of the classical information obtained by an agent in $P$ is independent of $X$.
    \item Given the distribution of classical information described in 1, the state of the quantum system obtained by an agent in $P$ is fixed and independent of $X$.
\end{enumerate}
\end{definition}

Definition \ref{definition:blindness} is proposed in \cite{Broadbent2008UniversalComputation} as a formal description to characterize blindness. Here $X$ denotes information that the agent can obtain, and $L(X)$ is any information that can be inferred from given $X$. Now that $Y=L(X)$ is given, the agent cannot infer any algorithm information if the protocol is blind. The agent has two sources of information: the instructions it receives and the measurement outcome it gets. The first source suggests that the classical instructions obtained by the agent must be independent of the algorithm, and the second source suggests that the quantum information or measurement outcome must be independent of the algorithm. 

We show the blindness of our protocol by proving the independence between the quantum algorithm being executed and the information each agent has access to.

\begin{theorem}[Inter-block connectivity independence]\label{theorem:inter_block_connectivity_independence}
Distribution of $s_{i,k}$ is independent of the executed quantum algorithm. 
\end{theorem}
\begin{proof}
$s_{i,k}$ is only used to identify the preshared entanglement pairs; its choice is independent of the quantum algorithm. 
\end{proof}

\begin{theorem}[Inner-block connectivity leakage]\label{theorem:inner_block_connectivity_leakage}
Distribution of $\Tilde{E(V_i,V_{i^\prime})}, V_i,V_{i\prime} \in B_m$ can leak at most $max(deg(V_i))$.
\end{theorem}
\begin{proof}
For $E(\Tilde{V}_i,\Tilde{V_{i^\prime}})$, the agent $A_m$ can recover the $deg(\Tilde{V}_i)$ by reversely apply rule $\spiderrule$. With the extra dummy connectivity introduced, the degree recovered here is not necessarily the exact degree from the original algorithm, however, it is always greater or equal to $deg(\Tilde{V}_i)$.
\end{proof}

\begin{theorem}[Safety of the correction process]\label{theorem:state_recovery_safety}
The correction process does not leak information.
\end{theorem}
\begin{proof}
The correction process requires the client to modify $\alpha_i$ based on the measurement result of previous spiders. The information may leak out from the connectivity of the ZX-diagram, the phase information of each spider, and the measurement outcome distribution. We now discuss each aspect separately.
\begin{enumerate}
\item Correction doesn't change the connectivity between spiders; therefore it doesn't invalidate theorem \ref{theorem:inter_block_connectivity_independence} or \ref{theorem:inner_block_connectivity_leakage}. 
\item Note that the value of $\alpha_i$ updates with the measurement outcome from previous steps for correction, and from theorem \ref{fig:phase_splitting}, $\Tilde{\alpha}_{2i}$ and $\Tilde{\alpha}_{2i+1}$ is independent after the correction process updates $\alpha_i$. Therefore, the correction process will not invalidate the independence between the phase and the actual algorithm. 
\item The correction process would change the distribution of the measurement outcome. However, from theorem \ref{theorem:measurement_result_independence}, each agent could not obtain any information from its measurement result. 
\end{enumerate}
Therefore, the correction process does not leak information.
\end{proof}

\begin{theorem}[Extraction universality]\label{theorem:extraction_universality}
ZX-diagram generated from the proposed protocol can always be extracted into practical quantum operations for real-world devices.
\end{theorem}
\begin{proof}
The original graph-like ZX diagram was converted by a quantum circuit. Therefore, it must admit a focused Pauli flow~\cite{Duncan2019Graph-theoreticZX-calculus} and can be executed with measurement-based quantum computation~\cite{Gross2007Measurement-basedModel}. All the rewrite rules used in our protocol preserve the Pauli flow; therefore the obfuscated diagram must also admit a Pauli flow. Hadamard wires can be implemented into a CZ gate to extract the ZX-diagram into quantum operations. The regular wire only comes from splitting the spiders. So, each spider can have at most one regular wire. Such diagrams can be extracted with method form example~\ref{fig:wire_extraction}. These rules included all possible diagrams that can be generated from our protocol.
\end{proof}

\begin{theorem}[Universality and correctness]
The modified ZX-diagram $\Tilde{\mathcal{G}}$ is universal for quantum computation, can always be implemented on a quantum device, and yields the same distribution as $\mathcal{G}$.
\end{theorem}
\begin{proof}
Any arbitrary ZX-diagram $\mathcal{G}$ with flow can be converted to $\Tilde{\mathcal{G}}$ following the protocol and $\mathcal{G}$ is universal for quantum computation. Therefore $\Tilde{\mathcal{G}}$ is universal and yields the same result as $\mathcal{G}$. Then from theorem \ref{theorem:extraction_universality}, any $\mathcal{\Tilde{G}}$ generated from $\mathcal{G}$ with our protocol preserves its flow and can be extracted into quantum operations can be executed on a quantum device.
\end{proof}

\begin{definition}[$\epsilon$-private \cite{Broadbent2015DelegatingComputations}]\label{definition:leakage}

A delegated quantum computation protocol requires the implementation of a linear map $\Phi_i$ on agent $A_i \in \mathbf{A}$ given classical information $q_i$. A simulator $\mathcal{S}_i$ has the same input and output space as $\mathbf{A}-A_i$, which can simulate the interaction between $A_i$ and $\mathbf{A}-A_i$. The agent $A_i$ interacts with $\mathcal{S}_i$, producing a linear map $\Psi_i$. The protocol is $\epsilon$-private if for every agent $A_i$ there exists such simulator $\mathcal{S}_i$ that $||\Phi_i-\Psi_i||_{\diamond}<\epsilon$, where $||\Phi_i-\Psi_i||_{\diamond}$ denote the diamond distance between $\Phi_i$ and $\Psi_i$.

\end{definition}
\begin{theorem}[Private]\label{theorem:private}
Our protocol is $0$-private.
\end{theorem}
\begin{proof}

The graph of each agent constructed gives the Choi–Jamiołkowski state $J(\Phi_i)$ of the linear map $\Phi_i$ \cite{CHOI1975285}. The information obtained by an agent is $q_i =  \{G_i,\{s_i\}\}$, where $G_i(E_N,E_H,V)$ is the graph fragment assigned to agent $A_i$. From theorem \ref{fig:phase_splitting}, \ref{theorem:inter_block_connectivity_independence}, \ref{theorem:measurement_result_independence}, $q_i$ is randomly distributed and independent to the quantum algorithm for execution. Therefore $J(\Phi_i)$ is a mixed state with some layout restrictions from constructing connectivity obfuscation in theorem \ref{theorem:inter_block_connectivity_independence} and \ref{theorem:measurement_result_independence}. Consider a simulator $\mathcal{S}_i$ that keeps the pre-shared entanglement pairs but does nothing on them. See the figure below. The layout of the graph representing the corresponding Choi–Jamiołkowski state $J(\Psi_i)$ (the graph in the right solid square) is in fact, identical to $J(\Phi_i)$ (in the left solid square), therefore $||J(\Psi_i) - J(\Phi_i)|| = 0$.

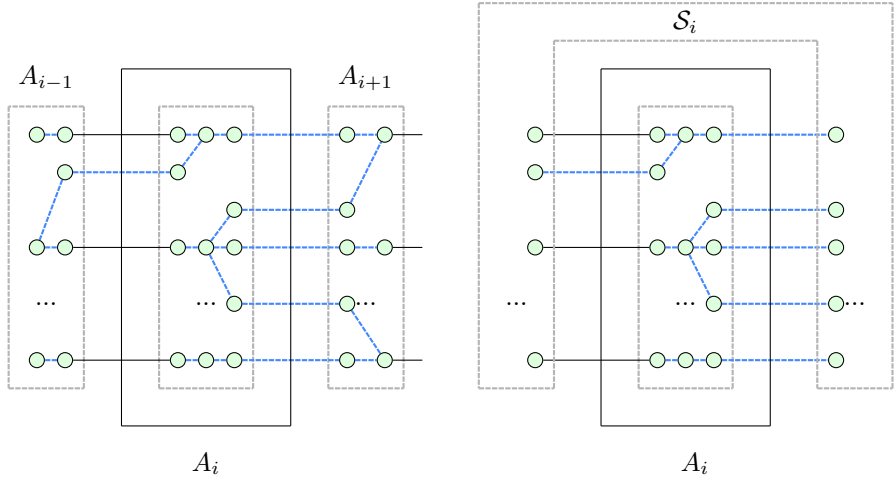
\ctikzfig{abstract/cj_private}

From relation  $\frac{1}{n}||\Phi_i-\Psi_i||_{\diamond}<||J(\Psi_i) - J(\Phi_i)||$ \cite{Distanceprocesses2005}, where $n$ is the size of the system, we conclude for our protocol is $0$-private.

\end{proof}

\begin{theorem}[Blindness] Our protocol is 0-private, and the information leakage would be at most ($max(deg(\Tilde{V}_i))$,$N(B_k)$,$n$) where $deg(\Tilde{V}_i)$ is the degree (number of wires connected to a spider) of $\Tilde{V}_i$, $max(deg(\Tilde{V}_i))$ is the maximum possible degree that $V_i$ could have. $N(B_k)$ is the qubit number of block $B_k$, $n$ is the of fragmented blocks in $\Tilde{G}$.\label{theo:blindness}
\end{theorem}
\begin{proof}
Client's input for each agent $A_m$ consists of $N(B_k)$,$n$, $\Tilde{\alpha}_i$, $s_{i,j}$ for all $V_i \in B_m$ and $V_j \in B_{l}$, where $B_l$ is all adjacent blocks of $B_m$, $E(\Tilde{V}_i,\Tilde{V}_{i^\prime})$ for $V_i,V_{i^\prime} \in B_m$.

\begin{enumerate}
    \item From theorem \ref{fig:phase_splitting}, $\Tilde{\alpha}_i$ is independent from the algorithm.
    \item From theorem \ref{theorem:inter_block_connectivity_independence}, $s_{i,j}$ is independent from the algorithm.
    \item From theorem \ref{theorem:inner_block_connectivity_leakage},at most $max(deg(\Tilde{V}_i))$ can be inferred by agent from the distribution of $E(\Tilde{V}_i,\Tilde{V}_{i^\prime})$.
    \item From theorem \ref{theorem:measurement_result_independence}, the measurement distribution of each qubit is independent of the algorithm.
    \item From theorem \ref{theorem:state_recovery_safety}, the correction process does not leak information.
    \item Each agent may know the total number of agents $m$, and infer the total block number $n$. 
    \item From example \ref{fig:wire_extraction}, each node can be extracted into quantum operation with or without a Hadamard gate. The existence of the Hadamard gate is independent of the algorithm.
    \item from theorem \ref{theorem:private}, our protocol is 0-private.
\end{enumerate}

Therefore, $A_m$ can get only $(max(deg(\Tilde{V}_i))$, $N(B_k), n)$ from the classical information it gets.
\end{proof}

The same as the UBQC protocol would inevitably disclose the size of the brickwork cluster state, our protocol also discloses some information about the resources required of the algorithm. UBQC uses a universal cluster state, which provides some surpluses of entanglement; therefore, UBQC protocol does not need to worry about the leakage of $max(deg(\Tilde{V}_i))$. Our proposal optimized the resource requirement, which discloses more information about required resources. However, such information can be hidden by allocating more resources and doing random operations on extra resources as long as it will not affect the computation result.

The secureness of our protocol requires that communication between different agents is limited. Except for the shared entanglement generated in advance, agents should not exchange any information during the execution. Such an assumption is difficult to be fulfilled indefinitely since two agents need to share entanglement. When there are collusive agents, blindness may be compromised. Here, we show that the blindness of our protocol would be compromised only when adjacent blocks are executed on two collusive agents. 

\begin{theorem}[Blindness compromise from collusive agents]
Information may leak out only when two adjacent blocks are executed on collusive agents.
\end{theorem}
\begin{proof}

When the attacker obtains information on two adjacent blocks, the attacker can apply the $\spiderrule$ rule to reverse the spider splitting and find the rotation angle or find a portion of connectivity in the original ZX-diagram. When attackers obtain information from non-adjacent blocks, it is equivalently to assign those non-adjacent blocks to the same agent. The attacker obtained the information from that single agent. Therefore from theorem \ref{theo:blindness} the information can be recovered is still $(max(deg(\Tilde{V}_i))$, $N(B_k), n)$.

\end{proof}

Although the proposed protocol only requires pre-shared bell pairs between agents, no information needs to be exchanged between agents at the run time.

Instead of physically limiting communication, the assumption can still be fulfilled with a decentralization strategy. 
For example, two quantum agents can be allocated from two different quantum service providers, and therefore it would be less likely to have two providers collude and compromise the blindness. More agents can also be introduced to have less chance of two adjacent blocks executed on collusive agents. Such relaxation is relatively weak since other strategies might be available if agents are honest and only exchange information the client allows. This relaxation allows the information exchange between agents even while executing the algorithm. Our protocol requires no information exchange between agents after the initial cluster state has been prepared. Our protocol would still be functioning if there were physical methods that could limit the communication between agents discovered in the future.

\section{A minimal example of our protocol\label{sec:swap_test}} 

In this section, we walk through a minimal example to implement a two qubits swap-test algorithm with the Hong-Ou-Mandel model~\cite{CarlosGarcia-Escartin2013TheEquivalent}. This algorithm does a CNOT gate and a Hadamard gate. The state overlap can be calculated based on the joint distribution of $O_1$ and $O_2$. We ignore the measurement obfuscation step for simplicity.

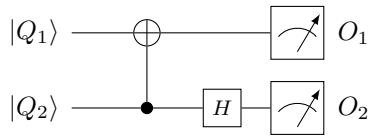
\begin{center}
    \begin{tikzpicture}
	\begin{pgfonlayer}{nodelayer}
		\node [style=none] (0) at (-6, 1) {};
		\node [style=none] (1) at (-6, -1) {};
		\node [style=target] (2) at (-4, 1) {};
		\node [style=vertex] (3) at (-4, -1) {};
		\node [style=box] (4) at (-2, -1) {$H$};
		\node [style=meter] (5) at (0, -1) {};
		\node [style=meter] (6) at (0, 1) {};
		\node [style=none] (7) at (-7, 1) {$\ket{Q_1}$};
		\node [style=none] (8) at (-7, -1) {$\ket{Q_2}$};
		\node [style=none] (9) at (1.5, 1) {$O_1$};
		\node [style=none] (10) at (1.5, -1) {$O_2$};
	\end{pgfonlayer}
	\begin{pgfonlayer}{edgelayer}
		\draw (0.center) to (2);
		\draw (2) to (6);
		\draw (1.center) to (3);
		\draw (3) to (4);
		\draw (4) to (5);
		\draw (2) to (3);
	\end{pgfonlayer}
\end{tikzpicture}
\end{center}
    
The circuit is written into the ZX-diagram as follows.

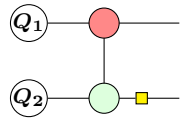
\begin{center}
    \begin{tikzpicture}
	\begin{pgfonlayer}{nodelayer}
		\node [style=vertex set] (0) at (0.5, 7.25) {$Q_1$};
		\node [style=vertex set] (1) at (0.5, 5.25) {$Q_2$};
		\node [style=X phase dot] (2) at (2.5, 7.25) {};
		\node [style=Z phase dot] (3) at (2.5, 5.25) {};
		\node [style=hadamard] (4) at (3.5, 5.25) {};
		\node [style=none] (5) at (4.5, 7.25) {};
		\node [style=none] (6) at (4.5, 5.25) {};
	\end{pgfonlayer}
	\begin{pgfonlayer}{edgelayer}
		\draw (0) to (2);
		\draw (1) to (3);
		\draw (2) to (3);
		\draw (3) to (4);
		\draw (4) to (6.center);
		\draw (2) to (5.center);
	\end{pgfonlayer}
\end{tikzpicture}
\end{center}

And then converted into a graph-like ZX-diagram. Note that $\alpha_1 = \alpha_2 = 0$.

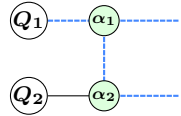
\begin{center}
    \begin{tikzpicture}
	\begin{pgfonlayer}{nodelayer}
		\node [style=vertex set] (0) at (0.5, 3) {$Q_1$};
		\node [style=vertex set] (1) at (0.5, 1) {$Q_2$};
		\node [style=Z phase dot] (2) at (2.5, 1) {$\alpha_2$};
		\node [style=none] (3) at (4.5, 3) {};
		\node [style=none] (4) at (4.5, 1) {};
		\node [style=Z phase dot] (5) at (2.5, 3) {$\alpha_1$};
	\end{pgfonlayer}
	\begin{pgfonlayer}{edgelayer}
		\draw (1) to (2);
		\draw [style=hadamard edge] (0) to (5);
		\draw [style=hadamard edge] (5) to (3.center);
		\draw [style=hadamard edge] (5) to (2);
		\draw [style=hadamard edge] (2) to (4.center);
	\end{pgfonlayer}
\end{tikzpicture}
\end{center}

Now, split each spider into two to make phase obfuscation. The value of $\Tilde{\alpha}_1$ and $\Tilde{\alpha}_3$ can be random, as long as $\Tilde{\alpha}_2 = -\Tilde{\alpha}_1$, $\Tilde{\alpha}_4 = -\Tilde{\alpha}_3$.

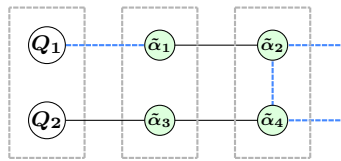
\begin{center}
    \begin{tikzpicture}
	\begin{pgfonlayer}{nodelayer}
		\node [style=Z phase dot] (0) at (7.5, -4) {$\Tilde{\alpha}_4$};
		\node [style=none] (1) at (9.5, -2) {};
		\node [style=none] (2) at (9.5, -4) {};
		\node [style=Z phase dot] (3) at (7.5, -2) {$\Tilde{\alpha}_2$};
		\node [style=Z phase dot] (4) at (4.5, -4) {$\Tilde{\alpha}_3$};
		\node [style=Z phase dot] (5) at (4.5, -2) {$\Tilde{\alpha}_1$};
		\node [style=none] (6) at (3.5, -1) {};
		\node [style=none] (7) at (5.5, -1) {};
		\node [style=none] (8) at (3.5, -5) {};
		\node [style=none] (9) at (5.5, -5) {};
		\node [style=none] (10) at (6.5, -1) {};
		\node [style=none] (11) at (6.5, -5) {};
		\node [style=none] (12) at (8.5, -5) {};
		\node [style=none] (13) at (8.5, -1) {};
		\node [style=vertex set] (14) at (1.5, -2) {$Q_1$};
		\node [style=vertex set] (15) at (1.5, -4) {$Q_2$};
		\node [style=none] (16) at (0.5, -1) {};
		\node [style=none] (17) at (0.5, -5) {};
		\node [style=none] (18) at (2.5, -5) {};
		\node [style=none] (19) at (2.5, -1) {};
	\end{pgfonlayer}
	\begin{pgfonlayer}{edgelayer}
		\draw [style=hadamard edge] (3) to (1.center);
		\draw [style=hadamard edge] (3) to (0);
		\draw [style=hadamard edge] (0) to (2.center);
		\draw (5) to (3);
		\draw (4) to (0);
		\draw [style=dash edge] (6.center) to (8.center);
		\draw [style=dash edge] (8.center) to (9.center);
		\draw [style=dash edge] (9.center) to (7.center);
		\draw [style=dash edge] (10.center) to (11.center);
		\draw [style=dash edge] (11.center) to (12.center);
		\draw [style=dash edge] (12.center) to (13.center);
		\draw [style=dash edge] (10.center) to (13.center);
		\draw [style=dash edge] (6.center) to (7.center);
		\draw [style=dash edge] (16.center) to (17.center);
		\draw [style=dash edge] (17.center) to (18.center);
		\draw [style=dash edge] (18.center) to (19.center);
		\draw [style=dash edge] (19.center) to (16.center);
		\draw [style=hadamard edge] (14) to (5);
		\draw (15) to (4);
	\end{pgfonlayer}
\end{tikzpicture}
\end{center}

Move the connectivity within the same block to another spider, making it an inter-block connectivity.

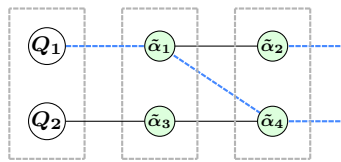
\begin{center}
    \begin{tikzpicture}
	\begin{pgfonlayer}{nodelayer}
		\node [style=Z phase dot] (0) at (7.5, -10) {$\Tilde{\alpha}_4$};
		\node [style=none] (1) at (9.5, -8) {};
		\node [style=none] (2) at (9.5, -10) {};
		\node [style=Z phase dot] (3) at (7.5, -8) {$\Tilde{\alpha}_2$};
		\node [style=Z phase dot] (4) at (4.5, -10) {$\Tilde{\alpha}_3$};
		\node [style=Z phase dot] (5) at (4.5, -8) {$\Tilde{\alpha}_1$};
		\node [style=none] (6) at (3.5, -7) {};
		\node [style=none] (7) at (3.5, -11) {};
		\node [style=none] (8) at (5.5, -11) {};
		\node [style=none] (9) at (5.5, -7) {};
		\node [style=none] (10) at (6.5, -7) {};
		\node [style=none] (11) at (6.5, -11) {};
		\node [style=none] (12) at (8.5, -11) {};
		\node [style=none] (13) at (8.5, -7) {};
		\node [style=vertex set] (14) at (1.5, -8) {$Q_1$};
		\node [style=vertex set] (15) at (1.5, -10) {$Q_2$};
		\node [style=none] (16) at (0.5, -7) {};
		\node [style=none] (17) at (0.5, -11) {};
		\node [style=none] (18) at (2.5, -11) {};
		\node [style=none] (19) at (2.5, -7) {};
	\end{pgfonlayer}
	\begin{pgfonlayer}{edgelayer}
		\draw [style=hadamard edge] (3) to (1.center);
		\draw [style=hadamard edge] (0) to (2.center);
		\draw (5) to (3);
		\draw (4) to (0);
		\draw [style=hadamard edge] (5) to (0);
		\draw [style=dash edge] (6.center) to (7.center);
		\draw [style=dash edge] (7.center) to (8.center);
		\draw [style=dash edge] (8.center) to (9.center);
		\draw [style=dash edge] (9.center) to (6.center);
		\draw [style=dash edge] (10.center) to (11.center);
		\draw [style=dash edge] (11.center) to (12.center);
		\draw [style=dash edge] (12.center) to (13.center);
		\draw [style=dash edge] (13.center) to (10.center);
		\draw [style=dash edge] (16.center) to (17.center);
		\draw [style=dash edge] (17.center) to (18.center);
		\draw [style=dash edge] (18.center) to (19.center);
		\draw [style=dash edge] (19.center) to (16.center);
		\draw [style=hadamard edge] (14) to (5);
		\draw (15) to (4);
	\end{pgfonlayer}
\end{tikzpicture}
\end{center}

Grows extra spider and finish the connectivity obfuscation.

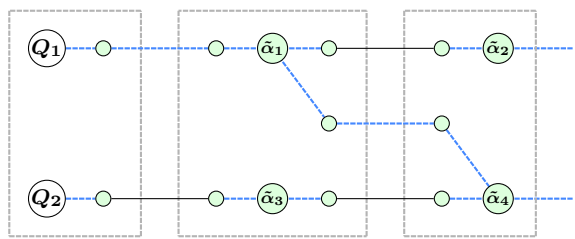
\begin{center}
    \begin{tikzpicture}
	\begin{pgfonlayer}{nodelayer}
		\node [style=Z phase dot] (0) at (4.5, -2) {$\Tilde{\alpha}_4$};
		\node [style=none] (1) at (6.5, 2) {};
		\node [style=none] (2) at (6.5, -2) {};
		\node [style=Z phase dot] (3) at (4.5, 2) {$\Tilde{\alpha}_2$};
		\node [style=vertex set] (4) at (-7.5, 2) {$Q_1$};
		\node [style=vertex set] (5) at (-7.5, -2) {$Q_2$};
		\node [style=Z phase dot] (6) at (-1.5, -2) {$\Tilde{\alpha}_3$};
		\node [style=Z phase dot] (7) at (-1.5, 2) {$\Tilde{\alpha}_1$};
		\node [style=none] (10) at (-4, 3) {};
		\node [style=none] (11) at (-4, -3) {};
		\node [style=none] (12) at (1, -3) {};
		\node [style=none] (13) at (1, 3) {};
		\node [style=none] (14) at (2, 3) {};
		\node [style=none] (15) at (2, -3) {};
		\node [style=none] (16) at (5.5, -3) {};
		\node [style=none] (17) at (5.5, 3) {};
		\node [style=none] (18) at (-8.5, 3) {};
		\node [style=none] (19) at (-8.5, -3) {};
		\node [style=none] (20) at (-5, -3) {};
		\node [style=none] (21) at (-5, 3) {};
		\node [style=Z dot] (22) at (0, 2) {};
		\node [style=Z dot] (23) at (0, -2) {};
		\node [style=Z dot] (24) at (3, -2) {};
		\node [style=Z dot] (25) at (3, 2) {};
		\node [style=Z dot] (26) at (0, 0) {};
		\node [style=Z dot] (27) at (3, 0) {};
		\node [style=Z dot] (28) at (-3, 2) {};
		\node [style=Z dot] (29) at (-3, -2) {};
		\node [style=Z dot] (30) at (-6, 2) {};
		\node [style=Z dot] (31) at (-6, -2) {};
	\end{pgfonlayer}
	\begin{pgfonlayer}{edgelayer}
		\draw [style=hadamard edge] (3) to (1.center);
		\draw [style=hadamard edge] (0) to (2.center);
		\draw [style=dash edge] (10.center) to (11.center);
		\draw [style=dash edge] (11.center) to (12.center);
		\draw [style=dash edge] (12.center) to (13.center);
		\draw [style=dash edge] (13.center) to (10.center);
		\draw [style=dash edge] (14.center) to (15.center);
		\draw [style=dash edge] (15.center) to (16.center);
		\draw [style=dash edge] (16.center) to (17.center);
		\draw [style=dash edge] (17.center) to (14.center);
		\draw [style=dash edge] (18.center) to (19.center);
		\draw [style=dash edge] (19.center) to (20.center);
		\draw [style=dash edge] (20.center) to (21.center);
		\draw [style=dash edge] (21.center) to (18.center);
		\draw [style=hadamard edge] (22) to (7);
		\draw [style=hadamard edge] (25) to (3);
		\draw [style=hadamard edge] (7) to (26);
		\draw [style=hadamard edge] (26) to (27);
		\draw [style=hadamard edge] (27) to (0);
		\draw [style=hadamard edge] (6) to (23);
		\draw [style=hadamard edge] (24) to (0);
		\draw (22) to (25);
		\draw (23) to (24);
		\draw (31) to (29);
		\draw [style=hadamard edge] (29) to (6);
		\draw [style=hadamard edge] (5) to (31);
		\draw [style=hadamard edge] (4) to (30);
		\draw [style=hadamard edge] (30) to (28);
		\draw [style=hadamard edge] (28) to (7);
	\end{pgfonlayer}
\end{tikzpicture}
\end{center}

To construct the diagram above, each agent only requires a shared bell state at the beginning.

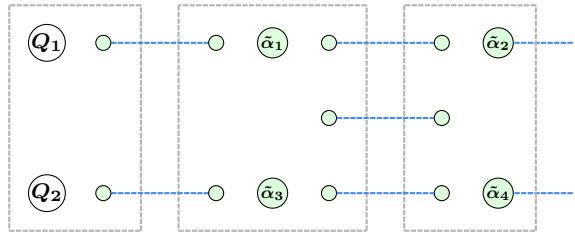
\begin{center}
    \begin{tikzpicture}
	\begin{pgfonlayer}{nodelayer}
		\node [style=Z phase dot] (30) at (13.5, 7) {$\Tilde{\alpha}_4$};
		\node [style=none] (31) at (15.5, 11) {};
		\node [style=none] (32) at (15.5, 7) {};
		\node [style=Z phase dot] (33) at (13.5, 11) {$\Tilde{\alpha}_2$};
		\node [style=vertex set] (34) at (1.5, 11) {$Q_1$};
		\node [style=vertex set] (35) at (1.5, 7) {$Q_2$};
		\node [style=Z phase dot] (36) at (7.5, 7) {$\Tilde{\alpha}_3$};
		\node [style=Z phase dot] (37) at (7.5, 11) {$\Tilde{\alpha}_1$};
		\node [style=none] (38) at (5, 12) {};
		\node [style=none] (39) at (5, 6) {};
		\node [style=none] (40) at (10, 6) {};
		\node [style=none] (41) at (10, 12) {};
		\node [style=none] (42) at (11, 12) {};
		\node [style=none] (43) at (11, 6) {};
		\node [style=none] (44) at (14.5, 6) {};
		\node [style=none] (45) at (14.5, 12) {};
		\node [style=none] (46) at (0.5, 12) {};
		\node [style=none] (47) at (0.5, 6) {};
		\node [style=none] (48) at (4, 6) {};
		\node [style=none] (49) at (4, 12) {};
		\node [style=Z dot] (50) at (9, 11) {};
		\node [style=Z dot] (51) at (9, 7) {};
		\node [style=Z dot] (52) at (12, 7) {};
		\node [style=Z dot] (53) at (12, 11) {};
		\node [style=Z dot] (54) at (9, 9) {};
		\node [style=Z dot] (55) at (12, 9) {};
		\node [style=Z dot] (56) at (6, 11) {};
		\node [style=Z dot] (57) at (6, 7) {};
		\node [style=Z dot] (58) at (3, 11) {};
		\node [style=Z dot] (59) at (3, 7) {};
	\end{pgfonlayer}
	\begin{pgfonlayer}{edgelayer}
		\draw [style=hadamard edge] (33) to (31.center);
		\draw [style=hadamard edge] (30) to (32.center);
		\draw [style=dash edge] (38.center) to (39.center);
		\draw [style=dash edge] (39.center) to (40.center);
		\draw [style=dash edge] (40.center) to (41.center);
		\draw [style=dash edge] (41.center) to (38.center);
		\draw [style=dash edge] (42.center) to (43.center);
		\draw [style=dash edge] (43.center) to (44.center);
		\draw [style=dash edge] (44.center) to (45.center);
		\draw [style=dash edge] (45.center) to (42.center);
		\draw [style=dash edge] (46.center) to (47.center);
		\draw [style=dash edge] (47.center) to (48.center);
		\draw [style=dash edge] (48.center) to (49.center);
		\draw [style=dash edge] (49.center) to (46.center);
		\draw [style=hadamard edge] (54) to (55);
		\draw [style=hadamard edge] (58) to (56);
		\draw [style=hadamard edge] (59) to (57);
		\draw [style=hadamard edge] (51) to (52);
		\draw [style=hadamard edge] (50) to (53);
	\end{pgfonlayer}
\end{tikzpicture}
\end{center}

Then apply the Hadamard gate as an example \ref{fig:wire_extraction} to convert Hadamard edges to normal edges.  
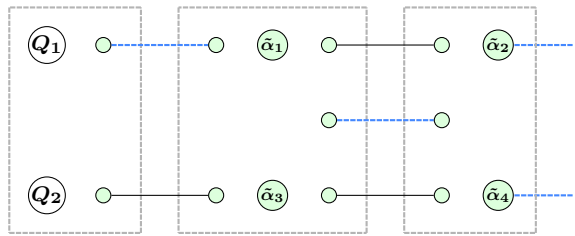
\begin{center}
    \begin{tikzpicture}
	\begin{pgfonlayer}{nodelayer}
		\node [style=Z phase dot] (0) at (13.5, -1) {$\Tilde{\alpha}_4$};
		\node [style=none] (1) at (15.5, 3) {};
		\node [style=none] (2) at (15.5, -1) {};
		\node [style=Z phase dot] (3) at (13.5, 3) {$\Tilde{\alpha}_2$};
		\node [style=vertex set] (4) at (1.5, 3) {$Q_1$};
		\node [style=vertex set] (5) at (1.5, -1) {$Q_2$};
		\node [style=Z phase dot] (6) at (7.5, -1) {$\Tilde{\alpha}_3$};
		\node [style=Z phase dot] (7) at (7.5, 3) {$\Tilde{\alpha}_1$};
		\node [style=none] (8) at (5, 4) {};
		\node [style=none] (9) at (5, -2) {};
		\node [style=none] (10) at (10, -2) {};
		\node [style=none] (11) at (10, 4) {};
		\node [style=none] (12) at (11, 4) {};
		\node [style=none] (13) at (11, -2) {};
		\node [style=none] (14) at (14.5, -2) {};
		\node [style=none] (15) at (14.5, 4) {};
		\node [style=none] (16) at (0.5, 4) {};
		\node [style=none] (17) at (0.5, -2) {};
		\node [style=none] (18) at (4, -2) {};
		\node [style=none] (19) at (4, 4) {};
		\node [style=Z dot] (20) at (9, 3) {};
		\node [style=Z dot] (21) at (9, -1) {};
		\node [style=Z dot] (22) at (12, -1) {};
		\node [style=Z dot] (23) at (12, 3) {};
		\node [style=Z dot] (24) at (9, 1) {};
		\node [style=Z dot] (25) at (12, 1) {};
		\node [style=Z dot] (26) at (6, 3) {};
		\node [style=Z dot] (27) at (6, -1) {};
		\node [style=Z dot] (28) at (3, 3) {};
		\node [style=Z dot] (29) at (3, -1) {};
	\end{pgfonlayer}
	\begin{pgfonlayer}{edgelayer}
		\draw [style=hadamard edge] (3) to (1.center);
		\draw [style=hadamard edge] (0) to (2.center);
		\draw [style=dash edge] (8.center) to (9.center);
		\draw [style=dash edge] (9.center) to (10.center);
		\draw [style=dash edge] (10.center) to (11.center);
		\draw [style=dash edge] (11.center) to (8.center);
		\draw [style=dash edge] (12.center) to (13.center);
		\draw [style=dash edge] (13.center) to (14.center);
		\draw [style=dash edge] (14.center) to (15.center);
		\draw [style=dash edge] (15.center) to (12.center);
		\draw [style=dash edge] (16.center) to (17.center);
		\draw [style=dash edge] (17.center) to (18.center);
		\draw [style=dash edge] (18.center) to (19.center);
		\draw [style=dash edge] (19.center) to (16.center);
		\draw [style=hadamard edge] (24) to (25);
		\draw (20) to (23);
		\draw (21) to (22);
		\draw (29) to (27);
		\draw [style=hadamard edge] (28) to (26);
	\end{pgfonlayer}
\end{tikzpicture}
\end{center}

Then apply entanglement operation within each agent.

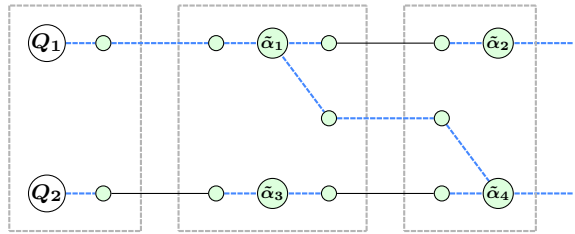
\begin{center}
    \begin{tikzpicture}
	\begin{pgfonlayer}{nodelayer}
		\node [style=Z phase dot] (0) at (13.5, -9) {$\Tilde{\alpha}_4$};
		\node [style=none] (1) at (15.5, -5) {};
		\node [style=none] (2) at (15.5, -9) {};
		\node [style=Z phase dot] (3) at (13.5, -5) {$\Tilde{\alpha}_2$};
		\node [style=vertex set] (4) at (1.5, -5) {$Q_1$};
		\node [style=vertex set] (5) at (1.5, -9) {$Q_2$};
		\node [style=Z phase dot] (6) at (7.5, -9) {$\Tilde{\alpha}_3$};
		\node [style=Z phase dot] (7) at (7.5, -5) {$\Tilde{\alpha}_1$};
		\node [style=none] (8) at (5, -4) {};
		\node [style=none] (9) at (5, -10) {};
		\node [style=none] (10) at (10, -10) {};
		\node [style=none] (11) at (10, -4) {};
		\node [style=none] (12) at (11, -4) {};
		\node [style=none] (13) at (11, -10) {};
		\node [style=none] (14) at (14.5, -10) {};
		\node [style=none] (15) at (14.5, -4) {};
		\node [style=none] (16) at (0.5, -4) {};
		\node [style=none] (17) at (0.5, -10) {};
		\node [style=none] (18) at (4, -10) {};
		\node [style=none] (19) at (4, -4) {};
		\node [style=Z dot] (20) at (9, -5) {};
		\node [style=Z dot] (21) at (9, -9) {};
		\node [style=Z dot] (22) at (12, -9) {};
		\node [style=Z dot] (23) at (12, -5) {};
		\node [style=Z dot] (24) at (9, -7) {};
		\node [style=Z dot] (25) at (12, -7) {};
		\node [style=Z dot] (26) at (6, -5) {};
		\node [style=Z dot] (27) at (6, -9) {};
		\node [style=Z dot] (28) at (3, -5) {};
		\node [style=Z dot] (29) at (3, -9) {};
	\end{pgfonlayer}
	\begin{pgfonlayer}{edgelayer}
		\draw [style=hadamard edge] (3) to (1.center);
		\draw [style=hadamard edge] (0) to (2.center);
		\draw [style=dash edge] (8.center) to (9.center);
		\draw [style=dash edge] (9.center) to (10.center);
		\draw [style=dash edge] (10.center) to (11.center);
		\draw [style=dash edge] (11.center) to (8.center);
		\draw [style=dash edge] (12.center) to (13.center);
		\draw [style=dash edge] (13.center) to (14.center);
		\draw [style=dash edge] (14.center) to (15.center);
		\draw [style=dash edge] (15.center) to (12.center);
		\draw [style=dash edge] (16.center) to (17.center);
		\draw [style=dash edge] (17.center) to (18.center);
		\draw [style=dash edge] (18.center) to (19.center);
		\draw [style=dash edge] (19.center) to (16.center);
		\draw [style=hadamard edge] (20) to (7);
		\draw [style=hadamard edge] (23) to (3);
		\draw [style=hadamard edge] (7) to (24);
		\draw [style=hadamard edge] (24) to (25);
		\draw [style=hadamard edge] (25) to (0);
		\draw [style=hadamard edge] (6) to (21);
		\draw [style=hadamard edge] (22) to (0);
		\draw (20) to (23);
		\draw (21) to (22);
		\draw (29) to (27);
		\draw [style=hadamard edge] (27) to (6);
		\draw [style=hadamard edge] (5) to (29);
		\draw [style=hadamard edge] (4) to (28);
		\draw [style=hadamard edge] (28) to (26);
		\draw [style=hadamard edge] (26) to (7);
	\end{pgfonlayer}
\end{tikzpicture}
\end{center}

To execute the algorithm, here we follow the standard MBQC protocol. First, the regular edges are merged, and the new spider represents the sum of the phase from two old spiders as example \ref{fig:execution_strategy}. These spiders are labelled with the red star symbol. 

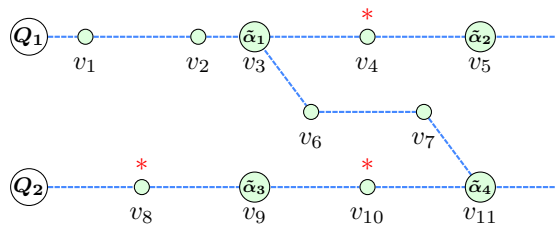
\begin{center}
    \begin{tikzpicture}
	\begin{pgfonlayer}{nodelayer}
		\node [style=Z phase dot] (0) at (13.5, -9) {$\Tilde{\alpha}_4$};
		\node [style=none] (1) at (15.5, -5) {};
		\node [style=none] (2) at (15.5, -9) {};
		\node [style=Z phase dot] (3) at (13.5, -5) {$\Tilde{\alpha}_2$};
		\node [style=vertex set] (4) at (1.5, -5) {$Q_1$};
		\node [style=vertex set] (5) at (1.5, -9) {$Q_2$};
		\node [style=Z phase dot] (6) at (7.5, -9) {$\Tilde{\alpha}_3$};
		\node [style=Z phase dot] (7) at (7.5, -5) {$\Tilde{\alpha}_1$};
		\node [style=Z dot] (20) at (10.5, -5) {};
		\node [style=Z dot] (21) at (10.5, -9) {};
		\node [style=Z dot] (24) at (9, -7) {};
		\node [style=Z dot] (25) at (12, -7) {};
		\node [style=Z dot] (26) at (6, -5) {};
		\node [style=Z dot] (28) at (3, -5) {};
		\node [style=Z dot] (29) at (4.5, -9) {};
		\node [style=none] (30) at (4.5, -8.5) {\textcolor{red}{*}};
		\node [style=none] (31) at (10.5, -8.5) {\textcolor{red}{*}};
		\node [style=none] (32) at (10.5, -4.5) {\textcolor{red}{*}};
		\node [style=none] (33) at (3, -5.75) {$v_1$};
		\node [style=none] (34) at (6, -5.75) {$v_2$};
		\node [style=none] (35) at (7.5, -5.75) {$v_3$};
		\node [style=none] (36) at (10.5, -5.75) {$v_4$};
		\node [style=none] (37) at (13.5, -5.75) {$v_5$};
		\node [style=none] (38) at (9, -7.75) {$v_6$};
		\node [style=none] (39) at (12, -7.75) {$v_7$};
		\node [style=none] (40) at (4.5, -9.75) {$v_8$};
		\node [style=none] (41) at (7.5, -9.75) {$v_9$};
		\node [style=none] (42) at (10.5, -9.75) {$v_{10}$};
		\node [style=none] (43) at (13.5, -9.75) {$v_{11}$};
	\end{pgfonlayer}
	\begin{pgfonlayer}{edgelayer}
		\draw [style=hadamard edge] (3) to (1.center);
		\draw [style=hadamard edge] (0) to (2.center);
		\draw [style=hadamard edge] (20) to (7);
		\draw [style=hadamard edge] (7) to (24);
		\draw [style=hadamard edge] (24) to (25);
		\draw [style=hadamard edge] (25) to (0);
		\draw [style=hadamard edge] (6) to (21);
		\draw [style=hadamard edge] (5) to (29);
		\draw [style=hadamard edge] (4) to (28);
		\draw [style=hadamard edge] (28) to (26);
		\draw [style=hadamard edge] (26) to (7);
		\draw [style=hadamard edge] (29) to (6);
		\draw [style=hadamard edge] (20) to (3);
		\draw [style=hadamard edge] (21) to (0);
	\end{pgfonlayer}
\end{tikzpicture}
\end{center}

Here we present a Pauli flow and the corresponding correction set. Define partial order  

\begin{equation}
    \prec := \{(v_1,v_2,v_6,v_7,v_8,v_8,v_4,v_{10})<(Q_1,Q_2)\}\bigcup\{(Q_1,Q_2)<(v_3,v_9)\}\bigcup\{(v_3,v_9)<(v_5,v_{11})\}
\end{equation}

where we define $A<B := \bigcup{\{(a,b)\}},\forall a\in A \text{~and~} \forall b\in B$.

\begin{center}
\begin{tabular}{ |c|c|c|c| } 
\hline
Vertex $u$ & Measurement Plane $\lambda(u)$& Correction set $f(u)$ & $Odd(f(u))$ \\
\hline
$Q_1$,& XY &  \{$v_1,v_3,v_5,v_7$\}   & \{$Q_1$,$v_{11}$\}\\ 
$Q_2$,& XY  & \{$v_8$\}   & \{$Q_2$,$v_9$\}\\
$v_1$ & X  &  \{$v_2$\}   & \{$Q_1$,$v_3$\}\\
$v_2$ & X  &  \{$v_3,v_5,v_7$\}   & \{$v_2$,$v_{11}$\}\\
$v_3$ & XY  & \{$v_4$\}   & \{$v_3$,$v_5$\}\\
$v_4$ & X   & \{$v_5$\}   & \{$v_4$\}\\
$v_5$ & XY  & N/A         & N/A\\
$v_6$ & X   & \{$v_7$\}   & \{$v_6$,$v_{11}$\}\\
$v_7$ & X   & \{$v_{6}$\}& \{$v_{3}$,$v_7$\}\\
$v_8$ & X   & \{$v_9,v_{11}$\}   & \{$v_8$\}\\
$v_9$ & XY  & \{$v_{10}$\}& \{$v_9$,$v_{11}$\}\\
$v_{10}$& X & \{$v_{11},v_{6}$\}& \{$v_{10}$,$v_3$\}\\
$v_{11}$& XY& N/A&N/A\\
\hline
\end{tabular}
\end{center}

Based on the Pauli flow configuration, the algorithm can be executed with 4 steps. First all the qubits corresponds to vertices$ \{v_1,v_2,v_6,v_7,v_8,v_8,v_4,v_{10}\}$ and then measure $\{Q_1,Q_2\}$,then $\{v_3,v_9\}$, and finally $\{v_5,v_{11}\}$. As an example here we demonstrate the correction when $v_2$ is measured with unexpected results. First, we highlight all the vertices from the correction set of $v_2$.

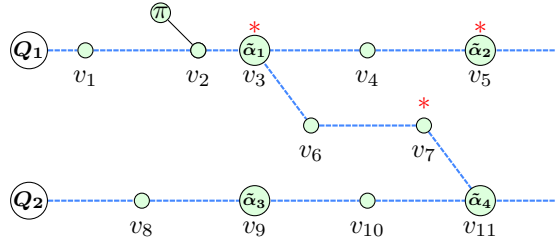
\begin{center}
    \begin{tikzpicture}
	\begin{pgfonlayer}{nodelayer}
		\node [style=Z phase dot] (0) at (13.5, -9) {$\Tilde{\alpha}_4$};
		\node [style=none] (1) at (15.5, -5) {};
		\node [style=none] (2) at (15.5, -9) {};
		\node [style=Z phase dot] (3) at (13.5, -5) {$\Tilde{\alpha}_2$};
		\node [style=vertex set] (4) at (1.5, -5) {$Q_1$};
		\node [style=vertex set] (5) at (1.5, -9) {$Q_2$};
		\node [style=Z phase dot] (6) at (7.5, -9) {$\Tilde{\alpha}_3$};
		\node [style=Z phase dot] (7) at (7.5, -5) {$\Tilde{\alpha}_1$};
		\node [style=Z dot] (20) at (10.5, -5) {};
		\node [style=Z dot] (21) at (10.5, -9) {};
		\node [style=Z dot] (24) at (9, -7) {};
		\node [style=Z dot] (25) at (12, -7) {};
		\node [style=Z dot] (26) at (6, -5) {};
		\node [style=Z dot] (28) at (3, -5) {};
		\node [style=Z dot] (29) at (4.5, -9) {};
		\node [style=none] (30) at (7.5, -4.5) {\textcolor{red}{*}};
		\node [style=none] (31) at (12, -6.5) {\textcolor{red}{*}};
		\node [style=none] (32) at (13.5, -4.5) {\textcolor{red}{*}};
		\node [style=none] (33) at (3, -5.75) {$v_1$};
		\node [style=none] (34) at (6, -5.75) {$v_2$};
		\node [style=none] (35) at (7.5, -5.75) {$v_3$};
		\node [style=none] (36) at (10.5, -5.75) {$v_4$};
		\node [style=none] (37) at (13.5, -5.75) {$v_5$};
		\node [style=none] (38) at (9, -7.75) {$v_6$};
		\node [style=none] (39) at (12, -7.75) {$v_7$};
		\node [style=none] (40) at (4.5, -9.75) {$v_8$};
		\node [style=none] (41) at (7.5, -9.75) {$v_9$};
		\node [style=none] (42) at (10.5, -9.75) {$v_{10}$};
		\node [style=none] (43) at (13.5, -9.75) {$v_{11}$};
		\node [style=Z dot] (44) at (5, -4) {$\pi$};
		\node [style=Z dot] (45) at (6, -5) {};
	\end{pgfonlayer}
	\begin{pgfonlayer}{edgelayer}
		\draw [style=hadamard edge] (3) to (1.center);
		\draw [style=hadamard edge] (0) to (2.center);
		\draw [style=hadamard edge] (20) to (7);
		\draw [style=hadamard edge] (7) to (24);
		\draw [style=hadamard edge] (24) to (25);
		\draw [style=hadamard edge] (25) to (0);
		\draw [style=hadamard edge] (6) to (21);
		\draw [style=hadamard edge] (5) to (29);
		\draw [style=hadamard edge] (4) to (28);
		\draw [style=hadamard edge] (28) to (26);
		\draw [style=hadamard edge] (26) to (7);
		\draw [style=hadamard edge] (29) to (6);
		\draw [style=hadamard edge] (20) to (3);
		\draw [style=hadamard edge] (21) to (0);
		\draw (44) to (45);
	\end{pgfonlayer}
\end{tikzpicture}
\end{center}

Now we apply $\pirule$ and $\identityrule$ and $\spiderrule$ to emit red spiders with pi phase on all the vertices in the correction set.
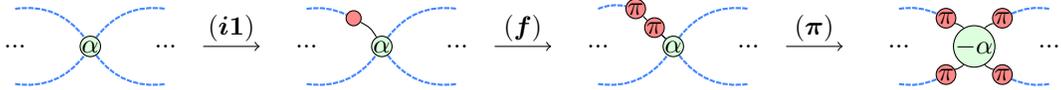
\begin{center}
    \begin{tikzpicture}
	\begin{pgfonlayer}{nodelayer}
		\node [style=Z dot] (0) at (-11, 0) {$\alpha$};
		\node [style=none] (1) at (-13, 1) {};
		\node [style=none] (2) at (-13, -1) {};
		\node [style=none] (3) at (-9, 1) {};
		\node [style=none] (4) at (-9, -1) {};
		\node [style=none] (5) at (-13, 0) {...};
		\node [style=none] (6) at (-9, 0) {...};
		\node [style=Z dot] (7) at (-3.25, 0) {$\alpha$};
		\node [style=none] (8) at (-5.25, 1) {};
		\node [style=none] (9) at (-5.25, -1) {};
		\node [style=none] (10) at (-1.25, 1) {};
		\node [style=none] (11) at (-1.25, -1) {};
		\node [style=none] (12) at (-5.25, 0) {...};
		\node [style=none] (13) at (-1.25, 0) {...};
		\node [style=none] (14) at (-8, 0) {};
		\node [style=none] (15) at (-6.5, 0) {};
		\node [style=none] (17) at (-7.25, 0.5) {$\identityrule$};
		\node [style=X dot] (18) at (-4, 0.75) {};
		\node [style=Z dot] (19) at (4.5, 0) {$\alpha$};
		\node [style=none] (20) at (2.5, 1) {};
		\node [style=none] (21) at (2.5, -1) {};
		\node [style=none] (22) at (6.5, 1) {};
		\node [style=none] (23) at (6.5, -1) {};
		\node [style=none] (24) at (2.5, 0) {...};
		\node [style=none] (25) at (6.5, 0) {...};
		\node [style=none] (26) at (-0.25, 0) {};
		\node [style=none] (27) at (1.25, 0) {};
		\node [style=none] (29) at (0.5, 0.5) {$\spiderrule$};
		\node [style=X dot] (30) at (3.5, 1) {$\pi$};
		\node [style=X dot] (31) at (4, 0.5) {$\pi$};
		\node [style=Z dot] (32) at (12.5, 0) {$-\alpha$};
		\node [style=none] (33) at (10.5, 1) {};
		\node [style=none] (34) at (10.5, -1) {};
		\node [style=none] (35) at (14.5, 1) {};
		\node [style=none] (36) at (14.5, -1) {};
		\node [style=none] (37) at (10.5, 0) {...};
		\node [style=none] (38) at (14.5, 0) {...};
		\node [style=none] (39) at (7.5, 0) {};
		\node [style=none] (40) at (9, 0) {};
		\node [style=none] (42) at (8.25, 0.5) {$\pirule$};
		\node [style=X dot] (43) at (11.75, 0.75) {$\pi$};
		\node [style=X dot] (44) at (11.75, -0.75) {$\pi$};
		\node [style=X dot] (45) at (13.25, 0.75) {$\pi$};
		\node [style=X dot] (46) at (13.25, -0.75) {$\pi$};
	\end{pgfonlayer}
	\begin{pgfonlayer}{edgelayer}
		\draw [style=hadamard edge, bend left] (1.center) to (0);
		\draw [style=hadamard edge, bend right] (2.center) to (0);
		\draw [style=hadamard edge, bend left] (0) to (3.center);
		\draw [style=hadamard edge, bend right] (0) to (4.center);
		\draw [style=hadamard edge, bend right] (9.center) to (7);
		\draw [style=hadamard edge, bend left] (7) to (10.center);
		\draw [style=hadamard edge, bend right] (7) to (11.center);
		\draw [style=diredge] (14.center) to (15.center);
		\draw [style=hadamard edge, bend left=15] (8.center) to (18);
		\draw [bend left=15, looseness=1.25] (18) to (7);
		\draw [style=hadamard edge, bend right] (21.center) to (19);
		\draw [style=hadamard edge, bend left] (19) to (22.center);
		\draw [style=hadamard edge, bend right] (19) to (23.center);
		\draw [style=diredge] (26.center) to (27.center);
		\draw [style=hadamard edge, bend left=15] (20.center) to (30);
		\draw (30) to (31);
		\draw (19) to (31);
		\draw (31) to (19);
		\draw [style=diredge] (39.center) to (40.center);
		\draw [style=hadamard edge, bend left=15] (33.center) to (43);
		\draw [style=hadamard edge, bend right=15] (34.center) to (44);
		\draw [style=hadamard edge, bend right=15] (46) to (36.center);
		\draw [style=hadamard edge, bend left=15] (45) to (35.center);
		\draw [bend right=15] (44) to (32);
		\draw [bend left=15] (32) to (45);
		\draw [bend left=15] (43) to (32);
		\draw [bend right=15] (32) to (46);
	\end{pgfonlayer}
\end{tikzpicture}
\end{center}

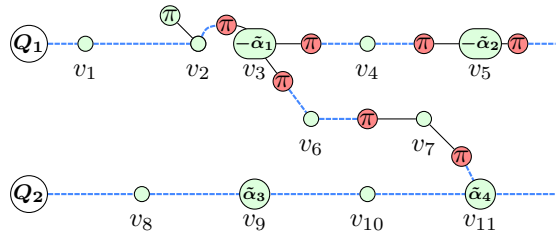
\begin{center}
    \begin{tikzpicture}
	\begin{pgfonlayer}{nodelayer}
		\node [style=Z phase dot] (0) at (13.5, -9) {$\Tilde{\alpha}_4$};
		\node [style=none] (1) at (15.5, -5) {};
		\node [style=none] (2) at (15.5, -9) {};
		\node [style=Z phase dot] (3) at (13.5, -5) {$-\Tilde{\alpha}_2$};
		\node [style=vertex set] (4) at (1.5, -5) {$Q_1$};
		\node [style=vertex set] (5) at (1.5, -9) {$Q_2$};
		\node [style=Z phase dot] (6) at (7.5, -9) {$\Tilde{\alpha}_3$};
		\node [style=Z phase dot] (7) at (7.5, -5) {$-\Tilde{\alpha}_1$};
		\node [style=Z dot] (20) at (10.5, -5) {};
		\node [style=Z dot] (21) at (10.5, -9) {};
		\node [style=Z dot] (24) at (9, -7) {};
		\node [style=Z dot] (25) at (12, -7) {};
		\node [style=Z dot] (26) at (6, -5) {};
		\node [style=Z dot] (28) at (3, -5) {};
		\node [style=Z dot] (29) at (4.5, -9) {};
		\node [style=none] (33) at (3, -5.75) {$v_1$};
		\node [style=none] (34) at (6, -5.75) {$v_2$};
		\node [style=none] (35) at (7.5, -5.75) {$v_3$};
		\node [style=none] (36) at (10.5, -5.75) {$v_4$};
		\node [style=none] (37) at (13.5, -5.75) {$v_5$};
		\node [style=none] (38) at (9, -7.75) {$v_6$};
		\node [style=none] (39) at (12, -7.75) {$v_7$};
		\node [style=none] (40) at (4.5, -9.75) {$v_8$};
		\node [style=none] (41) at (7.5, -9.75) {$v_9$};
		\node [style=none] (42) at (10.5, -9.75) {$v_{10}$};
		\node [style=none] (43) at (13.5, -9.75) {$v_{11}$};
		\node [style=Z dot] (44) at (5.25, -4.25) {$\pi$};
		\node [style=Z dot] (45) at (6, -5) {};
		\node [style=X dot] (46) at (6.75, -4.5) {};
		\node [style=X dot] (47) at (8.25, -6) {$\pi$};
		\node [style=X dot] (48) at (9, -5) {$\pi$};
		\node [style=X dot] (49) at (6.75, -4.5) {$\pi$};
		\node [style=X dot] (50) at (12, -5) {$\pi$};
		\node [style=X dot] (51) at (10.5, -7) {$\pi$};
		\node [style=X dot] (52) at (13, -8) {$\pi$};
		\node [style=X dot] (53) at (14.5, -5) {};
		\node [style=X dot] (54) at (14.5, -5) {$\pi$};
	\end{pgfonlayer}
	\begin{pgfonlayer}{edgelayer}
		\draw [style=hadamard edge] (0) to (2.center);
		\draw [style=hadamard edge] (6) to (21);
		\draw [style=hadamard edge] (5) to (29);
		\draw [style=hadamard edge] (4) to (28);
		\draw [style=hadamard edge] (28) to (26);
		\draw [style=hadamard edge] (29) to (6);
		\draw [style=hadamard edge] (21) to (0);
		\draw (44) to (45);
		\draw [bend left, looseness=1.25] (49) to (7);
		\draw (7) to (48);
		\draw (7) to (47);
		\draw [style=hadamard edge, bend left, looseness=1.25] (45) to (49);
		\draw [style=hadamard edge] (48) to (20);
		\draw [style=hadamard edge] (47) to (24);
		\draw (50) to (3);
		\draw (51) to (25);
		\draw (25) to (52);
		\draw [style=hadamard edge] (52) to (0);
		\draw [style=hadamard edge] (24) to (51);
		\draw [style=hadamard edge] (20) to (50);
		\draw (3) to (54);
		\draw [style=hadamard edge] (54) to (1.center);
	\end{pgfonlayer}
\end{tikzpicture}
\end{center}

Now push these red spiders through the Hadamard edge, and they become green spiders.

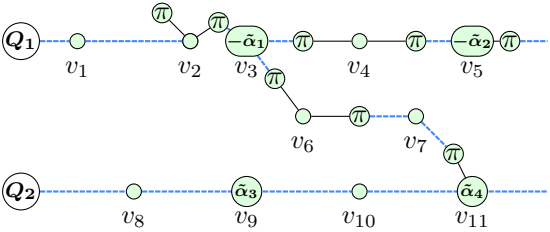
\begin{center}
    \begin{tikzpicture}
	\begin{pgfonlayer}{nodelayer}
		\node [style=Z phase dot] (0) at (13.5, -9) {$\Tilde{\alpha}_4$};
		\node [style=none] (1) at (15.5, -5) {};
		\node [style=none] (2) at (15.5, -9) {};
		\node [style=Z phase dot] (3) at (13.5, -5) {$-\Tilde{\alpha}_2$};
		\node [style=vertex set] (4) at (1.5, -5) {$Q_1$};
		\node [style=vertex set] (5) at (1.5, -9) {$Q_2$};
		\node [style=Z phase dot] (6) at (7.5, -9) {$\Tilde{\alpha}_3$};
		\node [style=Z phase dot] (7) at (7.5, -5) {$-\Tilde{\alpha}_1$};
		\node [style=Z dot] (20) at (10.5, -5) {};
		\node [style=Z dot] (21) at (10.5, -9) {};
		\node [style=Z dot] (24) at (9, -7) {};
		\node [style=Z dot] (25) at (12, -7) {};
		\node [style=Z dot] (26) at (6, -5) {};
		\node [style=Z dot] (28) at (3, -5) {};
		\node [style=Z dot] (29) at (4.5, -9) {};
		\node [style=none] (33) at (3, -5.75) {$v_1$};
		\node [style=none] (34) at (6, -5.75) {$v_2$};
		\node [style=none] (35) at (7.5, -5.75) {$v_3$};
		\node [style=none] (36) at (10.5, -5.75) {$v_4$};
		\node [style=none] (37) at (13.5, -5.75) {$v_5$};
		\node [style=none] (38) at (9, -7.75) {$v_6$};
		\node [style=none] (39) at (12, -7.75) {$v_7$};
		\node [style=none] (40) at (4.5, -9.75) {$v_8$};
		\node [style=none] (41) at (7.5, -9.75) {$v_9$};
		\node [style=none] (42) at (10.5, -9.75) {$v_{10}$};
		\node [style=none] (43) at (13.5, -9.75) {$v_{11}$};
		\node [style=Z dot] (44) at (5.25, -4.25) {$\pi$};
		\node [style=Z dot] (45) at (6, -5) {};
		\node [style=Z dot] (46) at (6.75, -4.5) {};
		\node [style=Z dot] (47) at (8.25, -6) {$\pi$};
		\node [style=Z dot] (48) at (9, -5) {$\pi$};
		\node [style=Z dot] (49) at (6.75, -4.5) {$\pi$};
		\node [style=Z dot] (50) at (12, -5) {$\pi$};
		\node [style=Z dot] (51) at (10.5, -7) {$\pi$};
		\node [style=Z dot] (52) at (13, -8) {$\pi$};
		\node [style=Z dot] (53) at (14.5, -5) {};
		\node [style=Z dot] (54) at (14.5, -5) {$\pi$};
	\end{pgfonlayer}
	\begin{pgfonlayer}{edgelayer}
		\draw [style=hadamard edge] (0) to (2.center);
		\draw [style=hadamard edge] (6) to (21);
		\draw [style=hadamard edge] (5) to (29);
		\draw [style=hadamard edge] (4) to (28);
		\draw [style=hadamard edge] (28) to (26);
		\draw [style=hadamard edge] (29) to (6);
		\draw [style=hadamard edge] (21) to (0);
		\draw (44) to (45);
		\draw (3) to (54);
		\draw [style=hadamard edge] (54) to (1.center);
		\draw (52) to (0);
		\draw (24) to (51);
		\draw (47) to (24);
		\draw (45) to (49);
		\draw (50) to (20);
		\draw (48) to (20);
		\draw [style=hadamard edge] (50) to (3);
		\draw [style=hadamard edge] (51) to (25);
		\draw [style=hadamard edge] (25) to (52);
		\draw [style=hadamard edge] (47) to (7);
		\draw [style=hadamard edge] (7) to (48);
		\draw [style=hadamard edge] (7) to (49);
	\end{pgfonlayer}
\end{tikzpicture}
\end{center}

Merge green spiders to cancel the unexpected measurement result. 
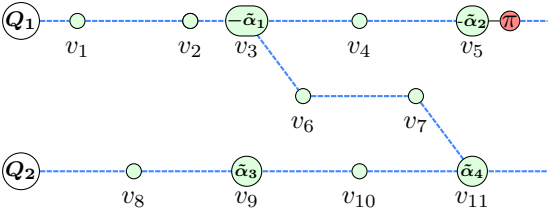
\begin{center}
    \begin{tikzpicture}
	\begin{pgfonlayer}{nodelayer}
		\node [style=Z phase dot] (0) at (13.5, -9) {$\Tilde{\alpha}_4$};
		\node [style=none] (1) at (15.5, -5) {};
		\node [style=none] (2) at (15.5, -9) {};
		\node [style=Z phase dot] (3) at (13.5, -5) {-$\Tilde{\alpha}_2$};
		\node [style=vertex set] (4) at (1.5, -5) {$Q_1$};
		\node [style=vertex set] (5) at (1.5, -9) {$Q_2$};
		\node [style=Z phase dot] (6) at (7.5, -9) {$\Tilde{\alpha}_3$};
		\node [style=Z phase dot] (7) at (7.5, -5) {$-\Tilde{\alpha}_1$};
		\node [style=Z dot] (20) at (10.5, -5) {};
		\node [style=Z dot] (21) at (10.5, -9) {};
		\node [style=Z dot] (24) at (9, -7) {};
		\node [style=Z dot] (25) at (12, -7) {};
		\node [style=Z dot] (26) at (6, -5) {};
		\node [style=Z dot] (28) at (3, -5) {};
		\node [style=Z dot] (29) at (4.5, -9) {};
		\node [style=none] (33) at (3, -5.75) {$v_1$};
		\node [style=none] (34) at (6, -5.75) {$v_2$};
		\node [style=none] (35) at (7.5, -5.75) {$v_3$};
		\node [style=none] (36) at (10.5, -5.75) {$v_4$};
		\node [style=none] (37) at (13.5, -5.75) {$v_5$};
		\node [style=none] (38) at (9, -7.75) {$v_6$};
		\node [style=none] (39) at (12, -7.75) {$v_7$};
		\node [style=none] (40) at (4.5, -9.75) {$v_8$};
		\node [style=none] (41) at (7.5, -9.75) {$v_9$};
		\node [style=none] (42) at (10.5, -9.75) {$v_{10}$};
		\node [style=none] (43) at (13.5, -9.75) {$v_{11}$};
		\node [style=Z dot] (45) at (6, -5) {};
		\node [style=X dot] (46) at (14.5, -5) {$\pi$};
	\end{pgfonlayer}
	\begin{pgfonlayer}{edgelayer}
		\draw [style=hadamard edge] (0) to (2.center);
		\draw [style=hadamard edge] (6) to (21);
		\draw [style=hadamard edge] (5) to (29);
		\draw [style=hadamard edge] (4) to (28);
		\draw [style=hadamard edge] (28) to (26);
		\draw [style=hadamard edge] (29) to (6);
		\draw [style=hadamard edge] (21) to (0);
		\draw [style=hadamard edge] (45) to (7);
		\draw [style=hadamard edge] (7) to (20);
		\draw [style=hadamard edge] (20) to (3);
		\draw [style=hadamard edge] (7) to (24);
		\draw [style=hadamard edge] (24) to (25);
		\draw [style=hadamard edge] (25) to (0);
		\draw [style=hadamard edge] (46) to (1.center);
		\draw (3) to (46);
	\end{pgfonlayer}
\end{tikzpicture}
\end{center}

\section{Discussion\label{sec:discussion}}

\subsection{Resource cost comparison between UBQC and our protocol} 

In this section, we quantify the resource requirements for implementing the UBQC and our proposed protocols and demonstrate the significant advantage of resource cost reduction compared to the UBQC protocol. Our proposed protocol distinguishes itself from the UBQC protocol primarily in how it implements connectivity obfuscation of the quantum algorithm. In the context of GBQC, connectivity refers to the layout of the quantum circuit, while in MBQC, connectivity is represented by the edges of the cluster state. In the ZX diagram, connectivity is depicted using a similar representation as wires between spiders. The resources we consider include the total number of qubits utilized by the protocol, the quantity of pre-shared Bell pairs necessary to establish entanglement between agents, and the overall number of two-qubit entanglement gates required within each agent.

The UBQC protocol achieves connectivity obfuscation by creating identical graphs for all algorithms, requiring the preparation of a universal cluster state that is algorithm-independent and information-free. Such universal cluster states restrict information to be stored in the measurement angles, thereby introducing redundancy to the graph. In contrast, our proposed protocol obfuscates connectivity by separating the two endpoints of the wires into two agents, enabling the layout of the ZX-diagram to carry information. Each agent is aware of a pre-shared Bell pair between itself and a neighbouring agent, however, it is uncertain which qubit the Bell pair is entangled with, resulting in a loss of information carried by the entanglement of the Bell pair. An example of resource reduction is shown in the following:

\begin{example}
An example of comparing the required resources to implement an arbitrary gate on two qubits. Here we denote the blue line as the entanglement generated from local entanglement gates within an agent, and the red lines are entanglement established by pre-shared Bell pairs or teleportation between agents. Note that this gate doesn't have to be a two-qubit entanglement gate; it can be two single qubit gate acts on two qubits separately. 
 (a) The brickwork cluster state is used in the UBQC protocol for implementing a quantum gate. (b-d) The implementation of the protocol proposed in this study. The entanglement gate can be implemented by both (b) and (c), and the two single qubit gates can be implemented by (d). The number of qubits and entanglement gates and pre-shared Bell pairs from our protocol (b-d) are less than UBQC protocol (a). \SXE{ The required measurement steps are one fewer than those in the brickwork state. As each column can be measured together, each block showing in (b,c,d) requires a maximum of 4 steps to execute. On the other hand, the brickwork state (a) incorporates an additional step: measuring all remotely entangled qubits (attached by red-coloured edges) to implement the teleportation. As for the number of measurements executing each block, the brickwork typically requires measuring 16 qubits, whereas our protocol requires a maximum of 10.}\vskip 1em

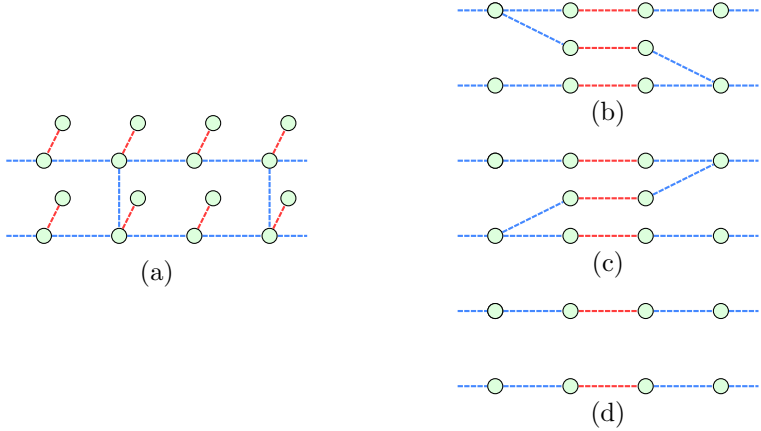
\centering
\begin{tikzpicture}
	\begin{pgfonlayer}{nodelayer}
		\node [style=Z dot] (0) at (-9, 7) {};
		\node [style=Z dot] (1) at (-7, 7) {};
		\node [style=Z dot] (2) at (-5, 7) {};
		\node [style=Z dot] (3) at (-3, 7) {};
		\node [style=Z dot] (4) at (-9, 5) {};
		\node [style=Z dot] (5) at (-7, 5) {};
		\node [style=Z dot] (6) at (-5, 5) {};
		\node [style=Z dot] (7) at (-3, 5) {};
		\node [style=none] (8) at (-10, 7) {};
		\node [style=none] (9) at (-10, 5) {};
		\node [style=none] (10) at (-2, 5) {};
		\node [style=none] (11) at (-2, 7) {};
		\node [style=Z dot] (12) at (-8.5, 8) {};
		\node [style=Z dot] (13) at (-8.5, 6) {};
		\node [style=Z dot] (14) at (-6.5, 8) {};
		\node [style=Z dot] (15) at (-6.5, 6) {};
		\node [style=Z dot] (16) at (-4.5, 8) {};
		\node [style=Z dot] (17) at (-4.5, 6) {};
		\node [style=Z dot] (18) at (-2.5, 8) {};
		\node [style=Z dot] (19) at (-2.5, 6) {};
		\node [style=none] (20) at (2, 7) {};
		\node [style=none] (21) at (2, 5) {};
		\node [style=Z dot] (22) at (3, 7) {};
		\node [style=Z dot] (23) at (3, 7) {};
		\node [style=Z dot] (24) at (5, 7) {};
		\node [style=Z dot] (25) at (7, 7) {};
		\node [style=Z dot] (26) at (9, 7) {};
		\node [style=Z dot] (27) at (5, 6) {};
		\node [style=Z dot] (28) at (7, 6) {};
		\node [style=Z dot] (29) at (3, 5) {};
		\node [style=Z dot] (30) at (5, 5) {};
		\node [style=Z dot] (31) at (7, 5) {};
		\node [style=Z dot] (32) at (9, 5) {};
		\node [style=none] (33) at (10, 7) {};
		\node [style=none] (34) at (10, 5) {};
		\node [style=none] (35) at (-6, 4) {(a)};
		\node [style=none] (36) at (6, 0.25) {(d)};
		\node [style=none] (37) at (2, 3) {};
		\node [style=none] (38) at (2, 1) {};
		\node [style=Z dot] (39) at (3, 3) {};
		\node [style=Z dot] (40) at (3, 3) {};
		\node [style=Z dot] (41) at (5, 3) {};
		\node [style=Z dot] (42) at (7, 3) {};
		\node [style=Z dot] (43) at (9, 3) {};
		\node [style=Z dot] (46) at (3, 1) {};
		\node [style=Z dot] (47) at (5, 1) {};
		\node [style=Z dot] (48) at (7, 1) {};
		\node [style=Z dot] (49) at (9, 1) {};
		\node [style=none] (50) at (10, 3) {};
		\node [style=none] (51) at (10, 1) {};
		\node [style=none] (52) at (2, 11) {};
		\node [style=none] (53) at (2, 9) {};
		\node [style=Z dot] (54) at (3, 11) {};
		\node [style=Z dot] (55) at (3, 11) {};
		\node [style=Z dot] (56) at (5, 11) {};
		\node [style=Z dot] (57) at (7, 11) {};
		\node [style=Z dot] (58) at (9, 11) {};
		\node [style=Z dot] (59) at (5, 10) {};
		\node [style=Z dot] (60) at (7, 10) {};
		\node [style=Z dot] (61) at (3, 9) {};
		\node [style=Z dot] (62) at (5, 9) {};
		\node [style=Z dot] (63) at (7, 9) {};
		\node [style=Z dot] (64) at (9, 9) {};
		\node [style=none] (65) at (10, 11) {};
		\node [style=none] (66) at (10, 9) {};
		\node [style=none] (67) at (6, 4.25) {(c)};
		\node [style=none] (68) at (6, 8.25) {(b)};
	\end{pgfonlayer}
	\begin{pgfonlayer}{edgelayer}
		\draw [style=hadamard edge] (0) to (1);
		\draw [style=hadamard edge] (1) to (2);
		\draw [style=hadamard edge] (2) to (3);
		\draw [style=hadamard edge] (3) to (7);
		\draw [style=hadamard edge] (4) to (5);
		\draw [style=hadamard edge] (5) to (6);
		\draw [style=hadamard edge] (6) to (7);
		\draw [style=hadamard edge] (1) to (5);
		\draw [style=hadamard edge] (8.center) to (0);
		\draw [style=hadamard edge] (9.center) to (4);
		\draw [style=hadamard edge] (7) to (10.center);
		\draw [style=hadamard edge] (3) to (11.center);
		\draw [style=red dash edge] (12) to (0);
		\draw [style=red dash edge] (13) to (4);
		\draw [style=red dash edge] (5) to (15);
		\draw [style=red dash edge] (1) to (14);
		\draw [style=red dash edge] (2) to (16);
		\draw [style=red dash edge] (6) to (17);
		\draw [style=red dash edge] (7) to (19);
		\draw [style=red dash edge] (3) to (18);
		\draw [style=red dash edge] (24) to (25);
		\draw [style=red dash edge] (27) to (28);
		\draw [style=red dash edge] (30) to (31);
		\draw [style=hadamard edge] (23) to (24);
		\draw [style=hadamard edge] (29) to (30);
		\draw [style=hadamard edge] (31) to (32);
		\draw [style=hadamard edge] (25) to (26);
		\draw [style=hadamard edge] (20.center) to (23);
		\draw [style=hadamard edge] (21.center) to (29);
		\draw [style=hadamard edge] (26) to (33.center);
		\draw [style=hadamard edge] (32) to (34.center);
		\draw [style=red dash edge] (41) to (42);
		\draw [style=red dash edge] (47) to (48);
		\draw [style=hadamard edge] (40) to (41);
		\draw [style=hadamard edge] (46) to (47);
		\draw [style=hadamard edge] (48) to (49);
		\draw [style=hadamard edge] (42) to (43);
		\draw [style=hadamard edge] (37.center) to (40);
		\draw [style=hadamard edge] (38.center) to (46);
		\draw [style=hadamard edge] (43) to (50.center);
		\draw [style=hadamard edge] (49) to (51.center);
		\draw [style=red dash edge] (56) to (57);
		\draw [style=red dash edge] (59) to (60);
		\draw [style=red dash edge] (62) to (63);
		\draw [style=hadamard edge] (55) to (56);
		\draw [style=hadamard edge] (55) to (59);
		\draw [style=hadamard edge] (61) to (62);
		\draw [style=hadamard edge] (63) to (64);
		\draw [style=hadamard edge] (60) to (64);
		\draw [style=hadamard edge] (57) to (58);
		\draw [style=hadamard edge] (52.center) to (55);
		\draw [style=hadamard edge] (53.center) to (61);
		\draw [style=hadamard edge] (58) to (65.center);
		\draw [style=hadamard edge] (64) to (66.center);
		\draw [style=hadamard edge] (29) to (27);
		\draw [style=hadamard edge] (28) to (26);
	\end{pgfonlayer}
\end{tikzpicture}

\end{example}

Now we move on to more general cases. For simplicity, we suppose the algorithm is decomposed into a gate set containing only local single-qubit gates and a CZ gate on nearest neighbours. These gates can be implemented directly with one ``brick'', the fundamental component in the brickwork cluster state. Denote $d$ as the circuit depth, $w$ as the circuit width or the number of qubits and $t$ as the number of two-qubit gates. For the worst that all gates are two-qubit gates, there are $\frac{1}{2}dw$ two-qubit gates. Therefore $t \leq \frac{1}{2}dw$. Here the entanglements between different agents are established by pre-shared Bell pairs, and entanglements within the same agent are implemented by entanglement gates performed by the agent. The transmission of quantum data between agents and clients is also considered using a Bell pair.%

For the brickwork state, each ``brick'' includes eight qubits and eight internal entanglement operations. Each ``brick'' hosts two qubits from the original algorithm. For both single-agent and multi-agent versions, each qubit needs to share a Bell pair with another agent or the client. In total, the qubits required to implement brickwork state is $\frac{1}{2}w\times8\times d+w = (4d+1)w$. For the semi-classical client UBQC, the brickwork state can be constructed and executed in sequence to recycle qubits. This strategy reduces the requirement qubit number to $2w+1$. The amount of Bell pairs is $(4d+1)w$. The amount of local entanglement gates is $\frac{1}{2}w\times8\times d$. For the single-agent version, an extra qubit is required; for the multi-agent version $(4d+1)w$ qubits are required. 

For our proposal, without any simplification, each qubit is split into three, and the qubits in the input and output blocks are split into two. Each two-qubit gate requires two more qubits, two more internal entanglements, and one extra external entanglement. For each dummy connection, two extra qubits, one Bell pair, and two local entanglements gates are required. Therefore our proposal requires $3(d-2)w + 2t$ qubits, $(d-1)w + t$ Bell pairs, $2(d-2)w+2t$ local entanglement gates.

See the table below to summarize the comparison.

\[
\centering
\begin{tabular}{llll}
                    Number of & Single-agent UBQC& Multi-agent UBQC & Our proposal \\ \hline
Agents        & $1$          & $\geq2$            & $\geq2$\\                  
Qubits        & $2w+1$ &$2(4d+1)w$          & $3(d-2)w + 2t$\\
Bell pairs & $(4d+1)w$ & $(4d+1)w$    & $(d-1)w + t$\\     
Local entanglement gates & $8dw$ & $8dw$     & $2(d-2)w+2t$\\     
\end{tabular}
\]

There are extra advantages to our protocol compared to UBQC. First, our protocol can implement non-nearest-neighbor entanglement directly. With the UBQC protocol, two qubits must be swapped to an adjacent position to perform the two-qubit gate, which adds extra cost to the implementation. Secondly, although the brickwork cluster state is universal, it is not intuitive to directly implement gates such as controlled single-qubit arbitrary rotation. Such gates can be decomposed into a ZX-diagram and directly implemented. Also, the quantum circuits can be simplified first with existing techniques from ZX-Calculus~\cite{Duncan2019Graph-theoreticZX-calculus} before applying our protocol. The graph-like ZX-diagram can be optimized until it only contains nodes representing non-Clifford operations. Such optimization can be considered the classically simulatable part of the quantum algorithm simplified from the diagram. Suppose the non-Clifford operation count is $c$, then the optimal qubit number would be $2c$, and the entanglement number would depend on the algorithm. Such a method can significantly reduce the resource requirement of our protocol.

\subsection{Compatibility with existing verification protocols}
\SXE{The universality of ZX-Diagram provides compatibility with most of the existing verification protocols. However, since some verification protocols require a universal cluster state, combining these verification protocols would invalidate our resource requirement advantage compared to the UBQC protocol. Here we discuss the ``first-order'' compatibility of verification protocols \cite{Gheorghiu2019VerificationApproaches}. The rigorous compatibility and full analysis of security with detailed proof, however, is beyond the topic of this paper.}

\SXE{Verification can be implemented by embedding a quantum circuit into the original algorithm that gives a deterministic result when the algorithm has been faithfully computed. The authentication-based verification method~\cite{Aharonov2017InteractiveComputations} extends the Quantum Authentication Schemes (QAS) as the embedded circuit. The trap-based verification method  \cite{Fitzsimons2017UnconditionallyComputation,Fitzsimons2017PrivateProtocols} utilize tapped wires or stabilizer codes for the embedded circuit. These embedded circuits can be converted into ZX-diagram and processed with our protocol. It is worth mentioning that work from \cite{Backens2013TheMechanics} makes it even more convenient to embed the stabilizer codes based on ZX-diagram.}

\SXE{Verification can also be implemented with the run and test scheme. The agent is asked to do calculations multiple times. The client randomly selects some of these calculations as test runs that run an algorithm in which the measurement distribution is known. The proposal from \cite{Broadbent2018HowComputation} suggests running the circuit in different initial states indistinguishable from the agent and using some of them as the tests. This method is compatible with our protocol since the initial state can be prepared arbitrarily and indistinguishable from the agent. The proposal from \cite{Hayashi2015VerifiableTesting} implement the test run with the same cluster state as the computation but modifies the measurement angle. Since our protocol no longer uses a universal cluster state, this proposal is invalid. However, we can still use the ZX Calculus to find phases for the same diagram layout but it gives a known probability distribution. If each spider's phase is chosen carefully, the ZX-diagram can be efficiently simulated~\cite{Duncan2019Graph-theoreticZX-calculus}.}

\SXE{Verification can be implemented with entanglement-based protocols. Proposals from \cite{Gheorghiu2015RobustnessComputing,Reichardt2013ClassicalSystems} make use of CHSH games and proposals from \cite{McKague2016InteractiveStates} utilize a self-testing graph states for verification. These methods all use self-testing results and pass the verification when the winning rate agrees with the prediction of quantum mechanics. These proposals are all compatible with ours; however, the self-testing graph protocol requires implementing a complicated graph state, which would invalidate our advantage compared to the UBQC protocol.
}

\section{Conclusion}\label{sec:conclusion}

We propose a multi-agent blind quantum computation protocol based on ZX-Calculus in this work. The quantum algorithm is first written into a ZX-diagram and then modified to be extracted into an MBQC-style algorithm. Then the algorithm is executed across multiple agents. We show that the information leakage to every agent is minimal, and our protocol's security can be guaranteed under the assumption that communication between agents is limited. Our proposal does not require a universal cluster state compared to the UBQC protocol. This advantage makes our protocol more flexible and efficient.

\section*{Acknowledgement}

We thank Miriam~Backens for providing the proof of theorem \ref{fig:double_extension} and \ref{fig:appendage_extension} from her work before publication. We thank John van de Wetering for reviewing this work and providing constructive comments. We thank Anne Broadbent, Lia~Yeh, Niel~de~Beaudrap, Quanlong~Wang, Xiao~Yuan, Brian~Vlastakis, and Peter~Leek for insightful discussions. We acknowledge the usage of the PyZX package~\cite{Kissinger2019PyZX:Reasoning} and the TikZiT tool.

\bibliography{main}

\begin{thebibliography}{48}%
\makeatletter
\providecommand \@ifxundefined [1]{%
 \@ifx{#1\undefined}
}%
\providecommand \@ifnum [1]{%
 \ifnum #1\expandafter \@firstoftwo
 \else \expandafter \@secondoftwo
 \fi
}%
\providecommand \@ifx [1]{%
 \ifx #1\expandafter \@firstoftwo
 \else \expandafter \@secondoftwo
 \fi
}%
\providecommand \natexlab [1]{#1}%
\providecommand \enquote  [1]{``#1''}%
\providecommand \bibnamefont  [1]{#1}%
\providecommand \bibfnamefont [1]{#1}%
\providecommand \citenamefont [1]{#1}%
\providecommand \href@noop [0]{\@secondoftwo}%
\providecommand \href [0]{\begingroup \@sanitize@url \@href}%
\providecommand \@href[1]{\@@startlink{#1}\@@href}%
\providecommand \@@href[1]{\endgroup#1\@@endlink}%
\providecommand \@sanitize@url [0]{\catcode `\\12\catcode `\$12\catcode
  `\&12\catcode `\#12\catcode `\^12\catcode `\_12\catcode `\%12\relax}%
\providecommand \@@startlink[1]{}%
\providecommand \@@endlink[0]{}%
\providecommand \url  [0]{\begingroup\@sanitize@url \@url }%
\providecommand \@url [1]{\endgroup\@href {#1}{\urlprefix }}%
\providecommand \urlprefix  [0]{URL }%
\providecommand \Eprint [0]{\href }%
\providecommand \doibase [0]{https://doi.org/}%
\providecommand \selectlanguage [0]{\@gobble}%
\providecommand \bibinfo  [0]{\@secondoftwo}%
\providecommand \bibfield  [0]{\@secondoftwo}%
\providecommand \translation [1]{[#1]}%
\providecommand \BibitemOpen [0]{}%
\providecommand \bibitemStop [0]{}%
\providecommand \bibitemNoStop [0]{.\EOS\space}%
\providecommand \EOS [0]{\spacefactor3000\relax}%
\providecommand \BibitemShut  [1]{\csname bibitem#1\endcsname}%
\let\auto@bib@innerbib\@empty
\bibitem [{\citenamefont {Krantz}\ \emph {et~al.}(2019)\citenamefont {Krantz},
  \citenamefont {Kjaergaard}, \citenamefont {Yan}, \citenamefont {Orlando},
  \citenamefont {Gustavsson},\ and\ \citenamefont
  {Oliver}}]{Krantz2019AQubits}%
  \BibitemOpen
  \bibfield  {author} {\bibinfo {author} {\bibfnamefont {P.}~\bibnamefont
  {Krantz}}, \bibinfo {author} {\bibfnamefont {M.}~\bibnamefont {Kjaergaard}},
  \bibinfo {author} {\bibfnamefont {F.}~\bibnamefont {Yan}}, \bibinfo {author}
  {\bibfnamefont {T.~P.}\ \bibnamefont {Orlando}}, \bibinfo {author}
  {\bibfnamefont {S.}~\bibnamefont {Gustavsson}},\ and\ \bibinfo {author}
  {\bibfnamefont {W.~D.}\ \bibnamefont {Oliver}},\ }\bibfield  {title}
  {\bibinfo {title} {{A quantum engineer's guide to superconducting qubits}},\
  }\href {https://doi.org/10.1063/1.5089550} {\bibfield  {journal} {\bibinfo
  {journal} {Applied Physics Reviews}\ }\textbf {\bibinfo {volume} {6}},\
  \bibinfo {pages} {021318} (\bibinfo {year} {2019})}\BibitemShut {NoStop}%
\bibitem [{\citenamefont {Bruzewicz}\ \emph {et~al.}(2019)\citenamefont
  {Bruzewicz}, \citenamefont {Chiaverini}, \citenamefont {McConnell},\ and\
  \citenamefont {Sage}}]{Bruzewicz2019Trapped-ionChallenges}%
  \BibitemOpen
  \bibfield  {author} {\bibinfo {author} {\bibfnamefont {C.~D.}\ \bibnamefont
  {Bruzewicz}}, \bibinfo {author} {\bibfnamefont {J.}~\bibnamefont
  {Chiaverini}}, \bibinfo {author} {\bibfnamefont {R.}~\bibnamefont
  {McConnell}},\ and\ \bibinfo {author} {\bibfnamefont {J.~M.}\ \bibnamefont
  {Sage}},\ }\bibfield  {title} {\bibinfo {title} {{Trapped-ion quantum
  computing: Progress and challenges}},\ }\href
  {https://doi.org/10.1063/1.5088164} {\bibfield  {journal} {\bibinfo
  {journal} {Applied Physics Reviews}\ }\textbf {\bibinfo {volume} {6}},\
  \bibinfo {pages} {021314} (\bibinfo {year} {2019})}\BibitemShut {NoStop}%
\bibitem [{\citenamefont {Kloeffel}\ and\ \citenamefont
  {Loss}(2013)}]{Kloeffel2013ProspectsDots}%
  \BibitemOpen
  \bibfield  {author} {\bibinfo {author} {\bibfnamefont {C.}~\bibnamefont
  {Kloeffel}}\ and\ \bibinfo {author} {\bibfnamefont {D.}~\bibnamefont
  {Loss}},\ }\bibfield  {title} {\bibinfo {title} {{Prospects for Spin-Based
  Quantum Computing in Quantum Dots}},\ }\href
  {https://doi.org/10.1146/annurev-conmatphys-030212-184248} {\bibfield
  {journal} {\bibinfo  {journal} {Annual Review of Condensed Matter Physics}\
  }\textbf {\bibinfo {volume} {4}},\ \bibinfo {pages} {51} (\bibinfo {year}
  {2013})}\BibitemShut {NoStop}%
\bibitem [{\citenamefont {Slussarenko}\ and\ \citenamefont
  {Pryde}(2019)}]{Slussarenko2019PhotonicReview}%
  \BibitemOpen
  \bibfield  {author} {\bibinfo {author} {\bibfnamefont {S.}~\bibnamefont
  {Slussarenko}}\ and\ \bibinfo {author} {\bibfnamefont {G.~J.}\ \bibnamefont
  {Pryde}},\ }\bibfield  {title} {\bibinfo {title} {{Photonic quantum
  information processing: A concise review}},\ }\href
  {https://doi.org/10.1063/1.5115814} {\bibfield  {journal} {\bibinfo
  {journal} {Applied Physics Reviews}\ }\textbf {\bibinfo {volume} {6}},\
  \bibinfo {pages} {041303} (\bibinfo {year} {2019})}\BibitemShut {NoStop}%
\bibitem [{\citenamefont {Fitzsimons}(2017)}]{Fitzsimons2017PrivateProtocols}%
  \BibitemOpen
  \bibfield  {author} {\bibinfo {author} {\bibfnamefont {J.~F.}\ \bibnamefont
  {Fitzsimons}},\ }\bibfield  {title} {\bibinfo {title} {{Private quantum
  computation: an introduction to blind quantum computing and related
  protocols}},\ }\href {https://doi.org/10.1038/s41534-017-0025-3} {\bibfield
  {journal} {\bibinfo  {journal} {npj Quantum Information}\ }\textbf {\bibinfo
  {volume} {3}},\ \bibinfo {pages} {23} (\bibinfo {year} {2017})}\BibitemShut
  {NoStop}%
\bibitem [{\citenamefont {Broadbent}\ \emph {et~al.}(2008)\citenamefont
  {Broadbent}, \citenamefont {Fitzsimons},\ and\ \citenamefont
  {Kashefi}}]{Broadbent2008UniversalComputation}%
  \BibitemOpen
  \bibfield  {author} {\bibinfo {author} {\bibfnamefont {A.}~\bibnamefont
  {Broadbent}}, \bibinfo {author} {\bibfnamefont {J.}~\bibnamefont
  {Fitzsimons}},\ and\ \bibinfo {author} {\bibfnamefont {E.}~\bibnamefont
  {Kashefi}},\ }\bibfield  {title} {\bibinfo {title} {{Universal blind quantum
  computation}},\ }\href {https://doi.org/10.1109/FOCS.2009.36} {\bibfield
  {journal} {\bibinfo  {journal} {Proceedings - Annual IEEE Symposium on
  Foundations of Computer Science, FOCS}\ ,\ \bibinfo {pages} {517}} (\bibinfo
  {year} {2008})}\BibitemShut {NoStop}%
\bibitem [{\citenamefont {Michielsen}\ \emph {et~al.}(2017)\citenamefont
  {Michielsen}, \citenamefont {Nocon}, \citenamefont {Willsch}, \citenamefont
  {Jin}, \citenamefont {Lippert},\ and\ \citenamefont
  {De~Raedt}}]{Michielsen2017BenchmarkingComputers}%
  \BibitemOpen
  \bibfield  {author} {\bibinfo {author} {\bibfnamefont {K.}~\bibnamefont
  {Michielsen}}, \bibinfo {author} {\bibfnamefont {M.}~\bibnamefont {Nocon}},
  \bibinfo {author} {\bibfnamefont {D.}~\bibnamefont {Willsch}}, \bibinfo
  {author} {\bibfnamefont {F.}~\bibnamefont {Jin}}, \bibinfo {author}
  {\bibfnamefont {T.}~\bibnamefont {Lippert}},\ and\ \bibinfo {author}
  {\bibfnamefont {H.}~\bibnamefont {De~Raedt}},\ }\bibfield  {title} {\bibinfo
  {title} {{Benchmarking gate-based quantum computers}},\ }\href
  {https://doi.org/10.1016/j.cpc.2017.06.011} {\bibfield  {journal} {\bibinfo
  {journal} {Computer Physics Communications}\ }\textbf {\bibinfo {volume}
  {220}},\ \bibinfo {pages} {44} (\bibinfo {year} {2017})}\BibitemShut
  {NoStop}%
\bibitem [{\citenamefont {Briegel}\ \emph {et~al.}(2009)\citenamefont
  {Briegel}, \citenamefont {Browne}, \citenamefont {D{\"{u}}r}, \citenamefont
  {Raussendorf},\ and\ \citenamefont {Van~den
  Nest}}]{Briegel2009Measurement-basedComputation}%
  \BibitemOpen
  \bibfield  {author} {\bibinfo {author} {\bibfnamefont {H.~J.}\ \bibnamefont
  {Briegel}}, \bibinfo {author} {\bibfnamefont {D.~E.}\ \bibnamefont {Browne}},
  \bibinfo {author} {\bibfnamefont {W.}~\bibnamefont {D{\"{u}}r}}, \bibinfo
  {author} {\bibfnamefont {R.}~\bibnamefont {Raussendorf}},\ and\ \bibinfo
  {author} {\bibfnamefont {M.}~\bibnamefont {Van~den Nest}},\ }\bibfield
  {title} {\bibinfo {title} {{Measurement-based quantum computation}},\ }\href
  {https://doi.org/10.1038/nphys1157} {\bibfield  {journal} {\bibinfo
  {journal} {Nature Physics}\ }\textbf {\bibinfo {volume} {5}},\ \bibinfo
  {pages} {19} (\bibinfo {year} {2009})}\BibitemShut {NoStop}%
\bibitem [{\citenamefont {Raussendorf}\ \emph {et~al.}(2006)\citenamefont
  {Raussendorf}, \citenamefont {Harrington},\ and\ \citenamefont
  {Goyal}}]{Raussendorf2006AComputer}%
  \BibitemOpen
  \bibfield  {author} {\bibinfo {author} {\bibfnamefont {R.}~\bibnamefont
  {Raussendorf}}, \bibinfo {author} {\bibfnamefont {J.}~\bibnamefont
  {Harrington}},\ and\ \bibinfo {author} {\bibfnamefont {K.}~\bibnamefont
  {Goyal}},\ }\bibfield  {title} {\bibinfo {title} {{A fault-tolerant one-way
  quantum computer}},\ }\href {https://doi.org/10.1016/j.aop.2006.01.012}
  {\bibfield  {journal} {\bibinfo  {journal} {Annals of Physics}\ }\textbf
  {\bibinfo {volume} {321}},\ \bibinfo {pages} {2242} (\bibinfo {year}
  {2006})}\BibitemShut {NoStop}%
\bibitem [{\citenamefont {Gross}\ \emph {et~al.}(2007)\citenamefont {Gross},
  \citenamefont {Eisert}, \citenamefont {Schuch},\ and\ \citenamefont
  {Perez-Garcia}}]{Gross2007Measurement-basedModel}%
  \BibitemOpen
  \bibfield  {author} {\bibinfo {author} {\bibfnamefont {D.}~\bibnamefont
  {Gross}}, \bibinfo {author} {\bibfnamefont {J.}~\bibnamefont {Eisert}},
  \bibinfo {author} {\bibfnamefont {N.}~\bibnamefont {Schuch}},\ and\ \bibinfo
  {author} {\bibfnamefont {D.}~\bibnamefont {Perez-Garcia}},\ }\bibfield
  {title} {\bibinfo {title} {{Measurement-based quantum computation beyond the
  one-way model}},\ }\href {https://doi.org/10.1103/PhysRevA.76.052315}
  {\bibfield  {journal} {\bibinfo  {journal} {Physical Review A}\ }\textbf
  {\bibinfo {volume} {76}},\ \bibinfo {pages} {052315} (\bibinfo {year}
  {2007})}\BibitemShut {NoStop}%
\bibitem [{\citenamefont {Kissinger}\ and\ \citenamefont {van~de
  Wetering}(2019{\natexlab{a}})}]{Kissinger2019UniversalMeasurements}%
  \BibitemOpen
  \bibfield  {author} {\bibinfo {author} {\bibfnamefont {A.}~\bibnamefont
  {Kissinger}}\ and\ \bibinfo {author} {\bibfnamefont {J.}~\bibnamefont {van~de
  Wetering}},\ }\bibfield  {title} {\bibinfo {title} {{Universal MBQC with
  generalised parity-phase interactions and Pauli measurements}},\ }\href
  {https://doi.org/10.22331/q-2019-04-26-134} {\bibfield  {journal} {\bibinfo
  {journal} {Quantum}\ }\textbf {\bibinfo {volume} {3}},\ \bibinfo {pages}
  {134} (\bibinfo {year} {2019}{\natexlab{a}})}\BibitemShut {NoStop}%
\bibitem [{\citenamefont
  {Nielsen}(2006)}]{Nielsen2006Cluster-stateComputation}%
  \BibitemOpen
  \bibfield  {author} {\bibinfo {author} {\bibfnamefont {M.~A.}\ \bibnamefont
  {Nielsen}},\ }\bibfield  {title} {\bibinfo {title} {{Cluster-state quantum
  computation}},\ }\href {https://doi.org/10.1016/S0034-4877(06)80014-5}
  {\bibfield  {journal} {\bibinfo  {journal} {Reports on Mathematical Physics}\
  }\textbf {\bibinfo {volume} {57}},\ \bibinfo {pages} {147} (\bibinfo {year}
  {2006})}\BibitemShut {NoStop}%
\bibitem [{\citenamefont {Mahadev}(2022)}]{Mahadev2018ClassicalComputations}%
  \BibitemOpen
  \bibfield  {author} {\bibinfo {author} {\bibfnamefont {U.}~\bibnamefont
  {Mahadev}},\ }\bibfield  {title} {\bibinfo {title} {Classical verification of
  quantum computations},\ }\href {https://doi.org/10.1137/20M1371828}
  {\bibfield  {journal} {\bibinfo  {journal} {SIAM Journal on Computing}\
  }\textbf {\bibinfo {volume} {51}},\ \bibinfo {pages} {1172} (\bibinfo {year}
  {2022})},\ \Eprint {https://arxiv.org/abs/https://doi.org/10.1137/20M1371828}
  {https://doi.org/10.1137/20M1371828} \BibitemShut {NoStop}%
\bibitem [{\citenamefont {Brakerski}(2018)}]{BrakerskiQuantumClassical}%
  \BibitemOpen
  \bibfield  {author} {\bibinfo {author} {\bibfnamefont {Z.}~\bibnamefont
  {Brakerski}},\ }\bibfield  {title} {\bibinfo {title} {Quantum fhe (almost) as
  secure as classical},\ }in\ \href@noop {} {\emph {\bibinfo {booktitle}
  {Advances in Cryptology -- CRYPTO 2018}}},\ \bibinfo {editor} {edited by\
  \bibinfo {editor} {\bibfnamefont {H.}~\bibnamefont {Shacham}}\ and\ \bibinfo
  {editor} {\bibfnamefont {A.}~\bibnamefont {Boldyreva}}}\ (\bibinfo
  {publisher} {Springer International Publishing},\ \bibinfo {address} {Cham},\
  \bibinfo {year} {2018})\ pp.\ \bibinfo {pages} {67--95}\BibitemShut {NoStop}%
\bibitem [{\citenamefont {Cojocaru}\ \emph {et~al.}(2019)\citenamefont
  {Cojocaru}, \citenamefont {Colisson}, \citenamefont {Kashefi},\ and\
  \citenamefont {Wallden}}]{Cojocaru2019QFactory:Preparation}%
  \BibitemOpen
  \bibfield  {author} {\bibinfo {author} {\bibfnamefont {A.}~\bibnamefont
  {Cojocaru}}, \bibinfo {author} {\bibfnamefont {L.}~\bibnamefont {Colisson}},
  \bibinfo {author} {\bibfnamefont {E.}~\bibnamefont {Kashefi}},\ and\ \bibinfo
  {author} {\bibfnamefont {P.}~\bibnamefont {Wallden}},\ }\bibfield  {title}
  {\bibinfo {title} {Qfactory: Classically-instructed remote secret qubits
  preparation},\ }in\ \href@noop {} {\emph {\bibinfo {booktitle} {Advances in
  Cryptology -- ASIACRYPT 2019}}},\ \bibinfo {editor} {edited by\ \bibinfo
  {editor} {\bibfnamefont {S.~D.}\ \bibnamefont {Galbraith}}\ and\ \bibinfo
  {editor} {\bibfnamefont {S.}~\bibnamefont {Moriai}}}\ (\bibinfo  {publisher}
  {Springer International Publishing},\ \bibinfo {address} {Cham},\ \bibinfo
  {year} {2019})\ pp.\ \bibinfo {pages} {615--645}\BibitemShut {NoStop}%
\bibitem [{\citenamefont {Morimae}\ and\ \citenamefont
  {Fujii}(2013)}]{Morimae2013BlindMeasurements}%
  \BibitemOpen
  \bibfield  {author} {\bibinfo {author} {\bibfnamefont {T.}~\bibnamefont
  {Morimae}}\ and\ \bibinfo {author} {\bibfnamefont {K.}~\bibnamefont
  {Fujii}},\ }\bibfield  {title} {\bibinfo {title} {{Blind quantum computation
  protocol in which Alice only makes measurements}},\ }\href
  {https://doi.org/10.1103/PhysRevA.87.050301} {\bibfield  {journal} {\bibinfo
  {journal} {Physical Review A}\ }\textbf {\bibinfo {volume} {87}},\ \bibinfo
  {pages} {050301} (\bibinfo {year} {2013})}\BibitemShut {NoStop}%
\bibitem [{\citenamefont {Sano}(2020)}]{Sano2020BlindComputer}%
  \BibitemOpen
  \bibfield  {author} {\bibinfo {author} {\bibfnamefont {Y.}~\bibnamefont
  {Sano}},\ }\href {http://arxiv.org/abs/2006.06255} {\bibinfo {title} {{Blind
  Quantum Computation Using a Circuit-Based Quantum Computer}}} (\bibinfo
  {year} {2020})\BibitemShut {NoStop}%
\bibitem [{\citenamefont {Broadbent}(2018)}]{Broadbent2018HowComputation}%
  \BibitemOpen
  \bibfield  {author} {\bibinfo {author} {\bibfnamefont {A.}~\bibnamefont
  {Broadbent}},\ }\bibfield  {title} {\bibinfo {title} {{How to verify a
  quantum computation}},\ }\href {https://doi.org/10.4086/toc.2018.v014a011}
  {\bibfield  {journal} {\bibinfo  {journal} {Theory of Computing}\ }\textbf
  {\bibinfo {volume} {14}},\ \bibinfo {pages} {1} (\bibinfo {year}
  {2018})}\BibitemShut {NoStop}%
\bibitem [{\citenamefont {Fisher}\ \emph {et~al.}(2014)\citenamefont {Fisher},
  \citenamefont {Broadbent}, \citenamefont {Shalm}, \citenamefont {Yan},
  \citenamefont {Lavoie}, \citenamefont {Prevedel}, \citenamefont {Jennewein},\
  and\ \citenamefont {Resch}}]{Fisher2014QuantumData}%
  \BibitemOpen
  \bibfield  {author} {\bibinfo {author} {\bibfnamefont {K.~A.}\ \bibnamefont
  {Fisher}}, \bibinfo {author} {\bibfnamefont {A.}~\bibnamefont {Broadbent}},
  \bibinfo {author} {\bibfnamefont {L.~K.}\ \bibnamefont {Shalm}}, \bibinfo
  {author} {\bibfnamefont {Z.}~\bibnamefont {Yan}}, \bibinfo {author}
  {\bibfnamefont {J.}~\bibnamefont {Lavoie}}, \bibinfo {author} {\bibfnamefont
  {R.}~\bibnamefont {Prevedel}}, \bibinfo {author} {\bibfnamefont
  {T.}~\bibnamefont {Jennewein}},\ and\ \bibinfo {author} {\bibfnamefont
  {K.~J.}\ \bibnamefont {Resch}},\ }\bibfield  {title} {\bibinfo {title}
  {{Quantum computing on encrypted data}},\ }\href
  {https://doi.org/10.1038/ncomms4074} {\bibfield  {journal} {\bibinfo
  {journal} {Nature Communications}\ }\textbf {\bibinfo {volume} {5}},\
  \bibinfo {pages} {1} (\bibinfo {year} {2014})}\BibitemShut {NoStop}%
\bibitem [{\citenamefont
  {Broadbent}(2015)}]{Broadbent2015DelegatingComputations}%
  \BibitemOpen
  \bibfield  {author} {\bibinfo {author} {\bibfnamefont {A.}~\bibnamefont
  {Broadbent}},\ }\bibfield  {title} {\bibinfo {title} {{Delegating Private
  Quantum Computations}},\ }\href {https://doi.org/10.1139/cjp-2015-0030}
  {\bibfield  {journal} {\bibinfo  {journal} {Canadian Journal of Physics}\
  }\textbf {\bibinfo {volume} {93}},\ \bibinfo {pages} {941} (\bibinfo {year}
  {2015})}\BibitemShut {NoStop}%
\bibitem [{\citenamefont {Coecke}\ and\ \citenamefont
  {Kissinger}(2018)}]{Coecke2018PicturingProcesses}%
  \BibitemOpen
  \bibfield  {author} {\bibinfo {author} {\bibfnamefont {B.}~\bibnamefont
  {Coecke}}\ and\ \bibinfo {author} {\bibfnamefont {A.}~\bibnamefont
  {Kissinger}},\ }\bibfield  {title} {\bibinfo {title} {{Picturing Quantum
  Processes}},\ }in\ \href {https://doi.org/10.1007/978-3-319-91376-6{\_}6}
  {\emph {\bibinfo {booktitle} {Lecture Notes in Computer Science}}},\ Vol.\
  \bibinfo {volume} {10871 LNAI}\ (\bibinfo  {publisher} {Springer Verlag},\
  \bibinfo {year} {2018})\ pp.\ \bibinfo {pages} {28--31}\BibitemShut {NoStop}%
\bibitem [{\citenamefont {Biamonte}\ and\ \citenamefont
  {Bergholm}(2017)}]{Biamonte2017TensorNutshell}%
  \BibitemOpen
  \bibfield  {author} {\bibinfo {author} {\bibfnamefont {J.}~\bibnamefont
  {Biamonte}}\ and\ \bibinfo {author} {\bibfnamefont {V.}~\bibnamefont
  {Bergholm}},\ }\href {http://arxiv.org/abs/1708.00006} {\bibinfo {title}
  {{Tensor Networks in a Nutshell}}} (\bibinfo {year} {2017})\BibitemShut
  {NoStop}%
\bibitem [{\citenamefont {Peng}\ \emph {et~al.}(2020)\citenamefont {Peng},
  \citenamefont {Harrow}, \citenamefont {Ozols},\ and\ \citenamefont
  {Wu}}]{PhysRevLett.125.150504}%
  \BibitemOpen
  \bibfield  {author} {\bibinfo {author} {\bibfnamefont {T.}~\bibnamefont
  {Peng}}, \bibinfo {author} {\bibfnamefont {A.~W.}\ \bibnamefont {Harrow}},
  \bibinfo {author} {\bibfnamefont {M.}~\bibnamefont {Ozols}},\ and\ \bibinfo
  {author} {\bibfnamefont {X.}~\bibnamefont {Wu}},\ }\bibfield  {title}
  {\bibinfo {title} {Simulating large quantum circuits on a small quantum
  computer},\ }\href {https://doi.org/10.1103/PhysRevLett.125.150504}
  {\bibfield  {journal} {\bibinfo  {journal} {Phys. Rev. Lett.}\ }\textbf
  {\bibinfo {volume} {125}},\ \bibinfo {pages} {150504} (\bibinfo {year}
  {2020})}\BibitemShut {NoStop}%
\bibitem [{\citenamefont {Backens}(2016)}]{Backens2016CompletenessZX-calculus}%
  \BibitemOpen
  \bibfield  {author} {\bibinfo {author} {\bibfnamefont {M.}~\bibnamefont
  {Backens}},\ }\href {http://arxiv.org/abs/1602.08954} {\bibinfo {title}
  {{Completeness and the ZX-calculus}}} (\bibinfo {year} {2016})\BibitemShut
  {NoStop}%
\bibitem [{\citenamefont {Jeandel}\ \emph {et~al.}(2020)\citenamefont
  {Jeandel}, \citenamefont {Perdrix},\ and\ \citenamefont
  {Vilmart}}]{Jeandel2020COMPLETENESSZX-CALCULUS}%
  \BibitemOpen
  \bibfield  {author} {\bibinfo {author} {\bibfnamefont {E.}~\bibnamefont
  {Jeandel}}, \bibinfo {author} {\bibfnamefont {S.}~\bibnamefont {Perdrix}},\
  and\ \bibinfo {author} {\bibfnamefont {R.}~\bibnamefont {Vilmart}},\
  }\bibfield  {title} {\bibinfo {title} {{COMPLETENESS OF THE ZX-CALCULUS}},\
  }\href {https://doi.org/10.23638/LMCS-16(2:11)2020} {\bibfield  {journal}
  {\bibinfo  {journal} {Logical Methods in Computer Science}\ }\textbf
  {\bibinfo {volume} {16}},\ \bibinfo {pages} {72} (\bibinfo {year}
  {2020})}\BibitemShut {NoStop}%
\bibitem [{\citenamefont {Schr{\"{o}}der~de Witt}\ and\ \citenamefont
  {Zamdzhiev}(2014)}]{SchroderdeWitt2014TheMechanics}%
  \BibitemOpen
  \bibfield  {author} {\bibinfo {author} {\bibfnamefont {C.}~\bibnamefont
  {Schr{\"{o}}der~de Witt}}\ and\ \bibinfo {author} {\bibfnamefont
  {V.}~\bibnamefont {Zamdzhiev}},\ }\bibfield  {title} {\bibinfo {title} {{The
  ZX-calculus is incomplete for quantum mechanics}},\ }\href
  {https://doi.org/10.4204/EPTCS.172.20} {\bibfield  {journal} {\bibinfo
  {journal} {Electronic Proceedings in Theoretical Computer Science}\ }\textbf
  {\bibinfo {volume} {172}},\ \bibinfo {pages} {285} (\bibinfo {year}
  {2014})}\BibitemShut {NoStop}%
\bibitem [{\citenamefont {Backens}\ \emph {et~al.}(2020)\citenamefont
  {Backens}, \citenamefont {Miller-Bakewell}, \citenamefont {de~Felice},
  \citenamefont {Lobski},\ and\ \citenamefont {van~de
  Wetering}}]{Backens2020ThereTale}%
  \BibitemOpen
  \bibfield  {author} {\bibinfo {author} {\bibfnamefont {M.}~\bibnamefont
  {Backens}}, \bibinfo {author} {\bibfnamefont {H.}~\bibnamefont
  {Miller-Bakewell}}, \bibinfo {author} {\bibfnamefont {G.}~\bibnamefont
  {de~Felice}}, \bibinfo {author} {\bibfnamefont {L.}~\bibnamefont {Lobski}},\
  and\ \bibinfo {author} {\bibfnamefont {J.}~\bibnamefont {van~de Wetering}},\
  }\href {http://arxiv.org/abs/2003.01664} {\bibinfo {title} {{There and back
  again: A circuit extraction tale}}} (\bibinfo {year} {2020})\BibitemShut
  {NoStop}%
\bibitem [{\citenamefont {Duncan}\ \emph {et~al.}(2019)\citenamefont {Duncan},
  \citenamefont {Kissinger}, \citenamefont {Perdrix},\ and\ \citenamefont
  {van~de Wetering}}]{Duncan2019Graph-theoreticZX-calculus}%
  \BibitemOpen
  \bibfield  {author} {\bibinfo {author} {\bibfnamefont {R.}~\bibnamefont
  {Duncan}}, \bibinfo {author} {\bibfnamefont {A.}~\bibnamefont {Kissinger}},
  \bibinfo {author} {\bibfnamefont {S.}~\bibnamefont {Perdrix}},\ and\ \bibinfo
  {author} {\bibfnamefont {J.}~\bibnamefont {van~de Wetering}},\ }\href
  {http://arxiv.org/abs/1902.03178} {\bibinfo {title} {{Graph-theoretic
  Simplification of Quantum Circuits with the ZX-calculus}}} (\bibinfo {year}
  {2019})\BibitemShut {NoStop}%
\bibitem [{\citenamefont {Aaronson}\ and\ \citenamefont
  {Gottesman}(2004)}]{Aaronson2004ImprovedCircuits}%
  \BibitemOpen
  \bibfield  {author} {\bibinfo {author} {\bibfnamefont {S.}~\bibnamefont
  {Aaronson}}\ and\ \bibinfo {author} {\bibfnamefont {D.}~\bibnamefont
  {Gottesman}},\ }\bibfield  {title} {\bibinfo {title} {{Improved simulation of
  stabilizer circuits}},\ }\bibfield  {journal} {\bibinfo  {journal} {Physical
  Review A - Atomic, Molecular, and Optical Physics}\ }\textbf {\bibinfo
  {volume} {70}},\ \href {https://doi.org/10.1103/PhysRevA.70.052328}
  {10.1103/PhysRevA.70.052328} (\bibinfo {year} {2004})\BibitemShut {NoStop}%
\bibitem [{\citenamefont {Danos}\ and\ \citenamefont
  {Kashefi}(2005)}]{Danos2005DeterminismModel}%
  \BibitemOpen
  \bibfield  {author} {\bibinfo {author} {\bibfnamefont {V.}~\bibnamefont
  {Danos}}\ and\ \bibinfo {author} {\bibfnamefont {E.}~\bibnamefont
  {Kashefi}},\ }\bibfield  {title} {\bibinfo {title} {{Determinism in the
  one-way model}},\ }\bibfield  {journal} {\bibinfo  {journal} {Physical Review
  A - Atomic, Molecular, and Optical Physics}\ }\textbf {\bibinfo {volume}
  {74}},\ \href {https://doi.org/10.1103/PhysRevA.74.052310}
  {10.1103/PhysRevA.74.052310} (\bibinfo {year} {2005})\BibitemShut {NoStop}%
\bibitem [{\citenamefont {Browne}\ \emph {et~al.}(2007)\citenamefont {Browne},
  \citenamefont {Kashefi}, \citenamefont {Mhalla},\ and\ \citenamefont
  {Perdrix}}]{Browne2007GeneralizedComputation}%
  \BibitemOpen
  \bibfield  {author} {\bibinfo {author} {\bibfnamefont {D.~E.}\ \bibnamefont
  {Browne}}, \bibinfo {author} {\bibfnamefont {E.}~\bibnamefont {Kashefi}},
  \bibinfo {author} {\bibfnamefont {M.}~\bibnamefont {Mhalla}},\ and\ \bibinfo
  {author} {\bibfnamefont {S.}~\bibnamefont {Perdrix}},\ }\bibfield  {title}
  {\bibinfo {title} {{Generalized flow and determinism in measurement-based
  quantum computation}},\ }\href {https://doi.org/10.1088/1367-2630/9/8/250}
  {\bibfield  {journal} {\bibinfo  {journal} {New Journal of Physics}\ }\textbf
  {\bibinfo {volume} {9}},\ \bibinfo {pages} {250} (\bibinfo {year}
  {2007})}\BibitemShut {NoStop}%
\bibitem [{\citenamefont {Duncan}\ and\ \citenamefont
  {Perdrix}(2010{\natexlab{a}})}]{Duncan2010RewritingFlow}%
  \BibitemOpen
  \bibfield  {author} {\bibinfo {author} {\bibfnamefont {R.}~\bibnamefont
  {Duncan}}\ and\ \bibinfo {author} {\bibfnamefont {S.}~\bibnamefont
  {Perdrix}},\ }\bibfield  {title} {\bibinfo {title} {{Rewriting
  Measurement-Based Quantum Computations with Generalised Flow}},\ }in\ \href
  {https://doi.org/10.1007/978-3-642-14162-1{\_}24} {\emph {\bibinfo
  {booktitle} {Lecture Notes in Computer Science}}},\ \bibinfo {series and
  number} {\bibinfo {number} {PART 2}}\ (\bibinfo  {publisher} {Springer
  Verlag},\ \bibinfo {year} {2010})\ pp.\ \bibinfo {pages}
  {285--296}\BibitemShut {NoStop}%
\bibitem [{\citenamefont {Duncan}(2012)}]{Duncan2012AComputing}%
  \BibitemOpen
  \bibfield  {author} {\bibinfo {author} {\bibfnamefont {R.}~\bibnamefont
  {Duncan}},\ }\bibfield  {title} {\bibinfo {title} {{A graphical approach to
  measurement-based quantum computing}},\ }\href
  {http://arxiv.org/abs/1203.6242} {\bibfield  {journal} {\bibinfo  {journal}
  {Quantum Physics and Linguistics}\ ,\ \bibinfo {pages} {50}} (\bibinfo {year}
  {2012})}\BibitemShut {NoStop}%
\bibitem [{\citenamefont {Backens}(2013)}]{Backens2013TheMechanics}%
  \BibitemOpen
  \bibfield  {author} {\bibinfo {author} {\bibfnamefont {M.}~\bibnamefont
  {Backens}},\ }\bibfield  {title} {\bibinfo {title} {{The ZX-calculus is
  complete for stabilizer quantum mechanics}},\ }\bibfield  {journal} {\bibinfo
   {journal} {New Journal of Physics}\ }\textbf {\bibinfo {volume} {16}},\
  \href {https://doi.org/10.1088/1367-2630/16/9/093021}
  {10.1088/1367-2630/16/9/093021} (\bibinfo {year} {2013})\BibitemShut
  {NoStop}%
\bibitem [{\citenamefont {Duncan}\ and\ \citenamefont
  {Perdrix}(2010{\natexlab{b}})}]{DuncanRewritingFlow}%
  \BibitemOpen
  \bibfield  {author} {\bibinfo {author} {\bibfnamefont {R.}~\bibnamefont
  {Duncan}}\ and\ \bibinfo {author} {\bibfnamefont {S.}~\bibnamefont
  {Perdrix}},\ }\bibfield  {title} {\bibinfo {title} {Rewriting
  measurement-based quantum computations with generalised flow},\ }in\
  \href@noop {} {\emph {\bibinfo {booktitle} {Automata, Languages and
  Programming}}},\ \bibinfo {editor} {edited by\ \bibinfo {editor}
  {\bibfnamefont {S.}~\bibnamefont {Abramsky}}, \bibinfo {editor}
  {\bibfnamefont {C.}~\bibnamefont {Gavoille}}, \bibinfo {editor}
  {\bibfnamefont {C.}~\bibnamefont {Kirchner}}, \bibinfo {editor}
  {\bibfnamefont {F.}~\bibnamefont {Meyer auf~der Heide}},\ and\ \bibinfo
  {editor} {\bibfnamefont {P.~G.}\ \bibnamefont {Spirakis}}}\ (\bibinfo
  {publisher} {Springer Berlin Heidelberg},\ \bibinfo {address} {Berlin,
  Heidelberg},\ \bibinfo {year} {2010})\ pp.\ \bibinfo {pages}
  {285--296}\BibitemShut {NoStop}%
\bibitem [{\citenamefont {McElvanney}\ and\ \citenamefont
  {Backens}(2022)}]{arxiv_2205.02009}%
  \BibitemOpen
  \bibfield  {author} {\bibinfo {author} {\bibfnamefont {T.}~\bibnamefont
  {McElvanney}}\ and\ \bibinfo {author} {\bibfnamefont {M.}~\bibnamefont
  {Backens}},\ }\href@noop {} {\bibinfo {title} {Complete flow-preserving
  rewrite rules for mbqc patterns with pauli measurements}} (\bibinfo {year}
  {2022}),\ \Eprint {https://arxiv.org/abs/arXiv:2205.02009} {arXiv:2205.02009}
  \BibitemShut {NoStop}%
\bibitem [{\citenamefont {McElvanney}\ and\ \citenamefont
  {Backens}(2023)}]{arxiv_2304.08166}%
  \BibitemOpen
  \bibfield  {author} {\bibinfo {author} {\bibfnamefont {T.}~\bibnamefont
  {McElvanney}}\ and\ \bibinfo {author} {\bibfnamefont {M.}~\bibnamefont
  {Backens}},\ }\href@noop {} {\bibinfo {title} {Flow-preserving zx-calculus
  rewrite rules for optimisation and obfuscation}} (\bibinfo {year} {2023}),\
  \Eprint {https://arxiv.org/abs/arXiv:2304.08166} {arXiv:2304.08166}
  \BibitemShut {NoStop}%
\bibitem [{\citenamefont {Choi}(1975)}]{CHOI1975285}%
  \BibitemOpen
  \bibfield  {author} {\bibinfo {author} {\bibfnamefont {M.-D.}\ \bibnamefont
  {Choi}},\ }\bibfield  {title} {\bibinfo {title} {Completely positive linear
  maps on complex matrices},\ }\href
  {https://doi.org/https://doi.org/10.1016/0024-3795(75)90075-0} {\bibfield
  {journal} {\bibinfo  {journal} {Linear Algebra and its Applications}\
  }\textbf {\bibinfo {volume} {10}},\ \bibinfo {pages} {285} (\bibinfo {year}
  {1975})}\BibitemShut {NoStop}%
\bibitem [{\citenamefont {Gilchrist}\ \emph {et~al.}(2005)\citenamefont
  {Gilchrist}, \citenamefont {Langford},\ and\ \citenamefont
  {Nielsen}}]{Distanceprocesses2005}%
  \BibitemOpen
  \bibfield  {author} {\bibinfo {author} {\bibfnamefont {A.}~\bibnamefont
  {Gilchrist}}, \bibinfo {author} {\bibfnamefont {N.~K.}\ \bibnamefont
  {Langford}},\ and\ \bibinfo {author} {\bibfnamefont {M.~A.}\ \bibnamefont
  {Nielsen}},\ }\bibfield  {title} {\bibinfo {title} {Distance measures to
  compare real and ideal quantum processes},\ }\href
  {https://doi.org/10.1103/PhysRevA.71.062310} {\bibfield  {journal} {\bibinfo
  {journal} {Phys. Rev. A}\ }\textbf {\bibinfo {volume} {71}},\ \bibinfo
  {pages} {062310} (\bibinfo {year} {2005})}\BibitemShut {NoStop}%
\bibitem [{\citenamefont {Garcia-Escartin}\ and\ \citenamefont
  {Chamorro-Posada}(2013)}]{CarlosGarcia-Escartin2013TheEquivalent}%
  \BibitemOpen
  \bibfield  {author} {\bibinfo {author} {\bibfnamefont {J.~C.}\ \bibnamefont
  {Garcia-Escartin}}\ and\ \bibinfo {author} {\bibfnamefont {P.}~\bibnamefont
  {Chamorro-Posada}},\ }\bibfield  {title} {\bibinfo {title} {swap test and
  hong-ou-mandel effect are equivalent},\ }\href
  {https://doi.org/10.1103/PhysRevA.87.052330} {\bibfield  {journal} {\bibinfo
  {journal} {Phys. Rev. A}\ }\textbf {\bibinfo {volume} {87}},\ \bibinfo
  {pages} {052330} (\bibinfo {year} {2013})}\BibitemShut {NoStop}%
\bibitem [{\citenamefont {Gheorghiu}\ \emph {et~al.}(2019)\citenamefont
  {Gheorghiu}, \citenamefont {Kapourniotis},\ and\ \citenamefont
  {Kashefi}}]{Gheorghiu2019VerificationApproaches}%
  \BibitemOpen
  \bibfield  {author} {\bibinfo {author} {\bibfnamefont {A.}~\bibnamefont
  {Gheorghiu}}, \bibinfo {author} {\bibfnamefont {T.}~\bibnamefont
  {Kapourniotis}},\ and\ \bibinfo {author} {\bibfnamefont {E.}~\bibnamefont
  {Kashefi}},\ }\bibfield  {title} {\bibinfo {title} {{Verification of Quantum
  Computation: An Overview of Existing Approaches}},\ }\href
  {https://doi.org/10.1007/s00224-018-9872-3} {\bibfield  {journal} {\bibinfo
  {journal} {Theory of Computing Systems}\ }\textbf {\bibinfo {volume} {63}},\
  \bibinfo {pages} {715} (\bibinfo {year} {2019})}\BibitemShut {NoStop}%
\bibitem [{\citenamefont {Aharonov}\ \emph {et~al.}(2017)\citenamefont
  {Aharonov}, \citenamefont {Ben-Or}, \citenamefont {Eban},\ and\ \citenamefont
  {Mahadev}}]{Aharonov2017InteractiveComputations}%
  \BibitemOpen
  \bibfield  {author} {\bibinfo {author} {\bibfnamefont {D.}~\bibnamefont
  {Aharonov}}, \bibinfo {author} {\bibfnamefont {M.}~\bibnamefont {Ben-Or}},
  \bibinfo {author} {\bibfnamefont {E.}~\bibnamefont {Eban}},\ and\ \bibinfo
  {author} {\bibfnamefont {U.}~\bibnamefont {Mahadev}},\ }\bibfield  {title}
  {\bibinfo {title} {{Interactive Proofs for Quantum Computations}},\ }\href
  {http://arxiv.org/abs/1704.04487} {\bibfield  {journal} {\bibinfo  {journal}
  {Lecture Notes in Computer Science (including subseries Lecture Notes in
  Artificial Intelligence and Lecture Notes in Bioinformatics)}\ }\textbf
  {\bibinfo {volume} {2906}},\ \bibinfo {pages} {1} (\bibinfo {year}
  {2017})}\BibitemShut {NoStop}%
\bibitem [{\citenamefont {Fitzsimons}\ and\ \citenamefont
  {Kashefi}(2017)}]{Fitzsimons2017UnconditionallyComputation}%
  \BibitemOpen
  \bibfield  {author} {\bibinfo {author} {\bibfnamefont {J.~F.}\ \bibnamefont
  {Fitzsimons}}\ and\ \bibinfo {author} {\bibfnamefont {E.}~\bibnamefont
  {Kashefi}},\ }\bibfield  {title} {\bibinfo {title} {{Unconditionally
  verifiable blind quantum computation}},\ }\href
  {https://doi.org/10.1103/PhysRevA.96.012303} {\bibfield  {journal} {\bibinfo
  {journal} {Physical Review A}\ }\textbf {\bibinfo {volume} {96}},\ \bibinfo
  {pages} {517} (\bibinfo {year} {2017})}\BibitemShut {NoStop}%
\bibitem [{\citenamefont {Hayashi}\ and\ \citenamefont
  {Morimae}(2015)}]{Hayashi2015VerifiableTesting}%
  \BibitemOpen
  \bibfield  {author} {\bibinfo {author} {\bibfnamefont {M.}~\bibnamefont
  {Hayashi}}\ and\ \bibinfo {author} {\bibfnamefont {T.}~\bibnamefont
  {Morimae}},\ }\bibfield  {title} {\bibinfo {title} {{Verifiable
  Measurement-Only Blind Quantum Computing with Stabilizer Testing}},\ }\href
  {https://doi.org/10.1103/PhysRevLett.115.220502} {\bibfield  {journal}
  {\bibinfo  {journal} {Physical Review Letters}\ }\textbf {\bibinfo {volume}
  {115}},\ \bibinfo {pages} {220502} (\bibinfo {year} {2015})}\BibitemShut
  {NoStop}%
\bibitem [{\citenamefont {Gheorghiu}\ \emph {et~al.}(2015)\citenamefont
  {Gheorghiu}, \citenamefont {Kashefi},\ and\ \citenamefont
  {Wallden}}]{Gheorghiu2015RobustnessComputing}%
  \BibitemOpen
  \bibfield  {author} {\bibinfo {author} {\bibfnamefont {A.}~\bibnamefont
  {Gheorghiu}}, \bibinfo {author} {\bibfnamefont {E.}~\bibnamefont {Kashefi}},\
  and\ \bibinfo {author} {\bibfnamefont {P.}~\bibnamefont {Wallden}},\
  }\bibfield  {title} {\bibinfo {title} {{Robustness and device independence of
  verifiable blind quantum computing}},\ }\bibfield  {journal} {\bibinfo
  {journal} {New Journal of Physics}\ }\textbf {\bibinfo {volume} {17}},\ \href
  {https://doi.org/10.1088/1367-2630/17/8/083040}
  {10.1088/1367-2630/17/8/083040} (\bibinfo {year} {2015})\BibitemShut
  {NoStop}%
\bibitem [{\citenamefont {Reichardt}\ \emph {et~al.}(2013)\citenamefont
  {Reichardt}, \citenamefont {Unger},\ and\ \citenamefont
  {Vazirani}}]{Reichardt2013ClassicalSystems}%
  \BibitemOpen
  \bibfield  {author} {\bibinfo {author} {\bibfnamefont {B.~W.}\ \bibnamefont
  {Reichardt}}, \bibinfo {author} {\bibfnamefont {F.}~\bibnamefont {Unger}},\
  and\ \bibinfo {author} {\bibfnamefont {U.}~\bibnamefont {Vazirani}},\
  }\bibfield  {title} {\bibinfo {title} {{Classical command of quantum
  systems}},\ }\href {https://doi.org/10.1038/nature12035} {\bibfield
  {journal} {\bibinfo  {journal} {Nature}\ }\textbf {\bibinfo {volume} {496}},\
  \bibinfo {pages} {456} (\bibinfo {year} {2013})}\BibitemShut {NoStop}%
\bibitem [{\citenamefont {McKague}(2016)}]{McKague2016InteractiveStates}%
  \BibitemOpen
  \bibfield  {author} {\bibinfo {author} {\bibfnamefont {M.}~\bibnamefont
  {McKague}},\ }\bibfield  {title} {\bibinfo {title} {{Interactive proofs for
  BQP via self-tested graph states}},\ }\href
  {https://doi.org/10.4086/toc.2016.v012a003} {\bibfield  {journal} {\bibinfo
  {journal} {Theory of Computing}\ }\textbf {\bibinfo {volume} {12}},\ \bibinfo
  {pages} {1} (\bibinfo {year} {2016})}\BibitemShut {NoStop}%
\bibitem [{\citenamefont {Kissinger}\ and\ \citenamefont {van~de
  Wetering}(2019{\natexlab{b}})}]{Kissinger2019PyZX:Reasoning}%
  \BibitemOpen
  \bibfield  {author} {\bibinfo {author} {\bibfnamefont {A.}~\bibnamefont
  {Kissinger}}\ and\ \bibinfo {author} {\bibfnamefont {J.}~\bibnamefont {van~de
  Wetering}},\ }\href {http://arxiv.org/abs/1904.04735} {\bibinfo {title}
  {{PyZX: Large Scale Automated Diagrammatic Reasoning}}} (\bibinfo {year}
  {2019}{\natexlab{b}})\BibitemShut {NoStop}%
\end{thebibliography}%

\end{document}